\documentclass[a4paper, fleqn, halfparskip,twoside,headsepline]{report} 

\usepackage[english]{babel}

\usepackage{tcolorbox}
\usepackage{comment}
\usepackage{placeins}
\usepackage{amsmath, amssymb,amsthm}
\usepackage{enumitem}
\usepackage{setspace} 
\newtheorem{thm}{Theorem}[chapter]
\newtheorem*{thm*}{Theorem}
\newtheorem{definition}[thm]{Definition}
\newtheorem{prop}[thm]{Proposition}
\newtheorem{lem}[thm]{Lemma}
\newtheorem{cor}[thm]{Corollary}
\newtheorem{ex}[thm]{Example}
\newtheorem{rem}[thm]{Remark}

\usepackage{mathtools}
\usepackage{dsfont}
\usepackage{enumitem}
\usepackage{units}
\usepackage{verbatim} 
\usepackage{courier}
\usepackage{tikz}
\usepackage{float}

\usepackage{verbatim} 

\usepackage{graphicx}
\usepackage{subfig}
\usepackage{textcomp} 

\usepackage{braket}
\usepackage[left=3cm,right=3cm,top=3.5cm,bottom=3.5cm]{geometry}

\renewcommand{\emph}[1]{\textsf{\textit{#1}}}

\DeclareMathOperator{\vol}{vol}
\DeclareMathOperator{\diam}{diam}
\DeclareMathOperator{\len}{len}
\DeclareMathOperator{\cay}{Cay}
\DeclareMathOperator{\poly}{poly}
\DeclareMathOperator{\spanN}{span}
\DeclareMathOperator{\dist}{dist}

\usepackage{scrpage2} 

\begin{document} 

\title{Adiabatic Quantum Computation}
\author{Friederike Anna Dziemba}
\date{September 25, 2014}

\begin{titlepage}
\,
\vspace{124pt}

\centering{
\begin{onehalfspace}

\begin{Large}
\quad

\begin{huge}
\vspace{4cm}
\sffamily Adiabatic Quantum Computation
\end{huge}

by

Friederike Anna Dziemba (2728680)

\vspace{2cm}
A thesis submitted in partial fulfillment\\ of the requirements for the degree of\\
Master of Science in Physics

\vspace{1cm}
at
\vspace{1cm}

The Leibniz University of Hanover\\
Institute of Theoretical Physics\\
Quantum Information Group

\vspace{1cm}
September 25, 2014\\
\end{Large}

\end{onehalfspace}
}
\end{titlepage}

\clearpage
\thispagestyle{empty}\quad
\clearpage

\thispagestyle{empty}
\quad
\vspace{50em}

\noindent \begin{tabular}{@{}ll}
Advisor: &Prof. Dr. Tobias J. Osborne\\
Co-Advisor: &Prof. Dr. Reinhard F. Werner
\end{tabular}

\vspace{5em}

\noindent I affirm that I have written the thesis myself and have not used any sources and aids other than those indicated.

\vspace{2em}

\noindent Date / Signature: \rule{0.5\textwidth}{0.4pt}

\clearpage
\thispagestyle{empty}\quad
\clearpage

\thispagestyle{empty}\cleardoublepage

\thispagestyle{empty}\chapter*{Abstract}

The quantum adiabatic theorem ensures that a slowly changing system, initially prepared in its ground-state, will evolve to its final ground-state with arbitrary precision. This fact can be exploited for a computational model with the ground-state carrying the computation information. The necessary evolution time of the adiabatic quantum computation increases with the inverse energy gap of the Hamiltonian. Currently a construction by Kitaev is the standard Hamiltonian used for simulation of an arbitrary quantum circuit via adiabatic quantum computation. The energy gap in this construction is mainly determined by the spectral gap  $\mathcal{O}\left(\frac{1}{L^2}\right)$ of an underlying path graph, with $L$ the length of the simulated circuit. In this thesis, we will broaden the concept of Kitaev to a class of ``standard graph Hamiltonians'' which allows us to substitute the path graph in the Kitaev Hamiltonian with different graph families possessing an improved spectral gap.

However, it turns out that restrictions on the graph families will make an improvement over the Kitaev construction difficult. On the one hand, we present some Hamiltonians based on particular graphs that show the same efficiency performance as the Kitaev Hamiltonian. On the other hand, we prove some restrictions on graphs in order to use them for an efficient adiabatic quantum computation by standard graph Hamiltonians. We will show that graphs with spectral gap  $\mathcal{O}\left(\frac{1}{L^k}\right)$, $k<1$, cannot be used for an efficient adiabatic quantum computation at all and that graphs with constant degree ratio and spectral gap $\mathcal{O}\left(\frac{1}{L^k}\right)$, $k<2$ need, as a minimal requirement, that their vertex set grows faster than polynomial in the circuit length.

Moreover we will prove in this thesis a new quantum adiabatic theorem for projection operators that expands the statement of the original adiabatic theorem to Hamiltonians with a degenerate ground-state.

\pagenumbering{gobble}
\thispagestyle{empty}\tableofcontents
\clearpage

\pagestyle{scrheadings}
\setheadsepline{0.4pt} 
\automark[chapter]{chapter} 

\pagenumbering{arabic}
 \chapter{Introduction}

\setcounter{page}{1}

This thesis presents a Hamiltonian model for adiabatic quantum computation with particular regard to its efficiency. Adiabatic quantum computation is based on the quantum adiabatic theorem that ensures that a slowly changing system initially prepared in its ground-state will evolve to its final ground-state with arbitrary precision. This fact motivates a computational model that encodes the computation input into the initial ground-state and ensures that the final ground-state encodes the desired computation output. The actual computation will then be achieved by time evolution. The neccessary evolution time depends, in addition to the error treshhold, on the energy gap of the Hamiltonian and the norm of its time derivatives. Since we will be interested in a particular efficient computation we will motivate a list of efficiency requirements for Hamiltonians in which the energy gap and the derivatives are the central optimization quantities.

The concept of quantum circuits will serve as a formalization of a computational task. A Hamiltonian that is capable of simulating quantum circuits via an efficient adiabatic quantum computation was introduced by Kitaev \cite{aharonov2}. Indeed the so-called Kitaev Hamiltonian is based on the normalized Laplacian of a graph whose spectral gap mainly domiates the actual energy gap of the Hamiltonian. We take advantage of the widely explored field of spectral graph theory to extend the idea of the Kitaev Hamiltonian to derive the more general concept of a ``standard graph Hamiltonian''. This Hamiltonian admits different underlying graph families and therefore the optimization of the energy gap turns into the task of finding an appropriate graph family.

However the optimization is more difficult than just choosing a graph family with a large gap, since the construction of a standard graph Hamiltonian and the efficiency requirements restrict the set of possible graphs we can use. In particular we will see that the need for a neccessary minimum diameter combined with the hope for a large spectral gap implies a problematic vertex expansion of the graph. We will actually not succeed in finding a construction that outperforms the Kitaev Hamiltonian, but we will present on the one hand some new constructions with the same efficiency and on the other hand prove some negative results, for example that standard graph Hamiltonians based on expander graphs will never be capable of an efficient adiabatic quantum computation.

Chapters 2--4 will each give an introduction to an important field: to spectral graph theory, quantum circuits and adiabatic quantum computation. Chapter 4 will also comprise a list of requirements for an efficient adiabatic quantum computation. In Chapter 5 we apply the combined knowledge of the basic chapters to define what we consider a ``standard graph Hamiltonian'' and specify the efficiency requirements for this particular kind of Hamiltonian. Section 5.3 will show that the important time derivatives of a standard graph Hamiltonian are always constant while Section 5.4 will contain some important restriction results. In the last three chapters we present, along with the Kitaev construction, some examples of standard graph Hamiltonians that allow an adiabatic quantum computation of comparable efficiency.

 \chapter{Graphs and Parallel Transport Networks}

\section{Basic definitions of graph theory}

In this first chapter we will lay out some basics of spectral graph theory. The following definitions and lemmata can be found in most standard books on spectral graph theory. As a reference see e.g. \cite[chapter 1]{chung}.

\begin{definition}\label{def:graph}
A \textsf{graph} $G=(V,E)$ consists of a finite, nonempty \textsf{vertex} set $V$ and a set $E \subseteq V \times V$ of \textsf{edges} such that $\forall u,v \in V: (u,v) \in E \Leftrightarrow (v,u) \in E$.
\end{definition}

\begin{definition}
A \textsf{weighted graph} $G=(V,E,w)$ is a graph $(V,E)$ with a \textsf{weight function} $w: V \times V \rightarrow \mathbb{R}^+_0$ with $w(v,u)=w(u,v)$ and $w(v,u)= 0 \Leftrightarrow (v,u) \notin E$  for all $v,u \in V$.
\end{definition}

Actually the edge set $E$ could be omitted in the definition of a weighted graph, because it is fully implied by the weight function. But for convenience it is advantageous to have a definition of an edge set.

Graphs are visualized by drawing for each edge $(u,v) \in E$ a line between two nodes representing the vertices and labelled respectively. For weighted graphs we add the edge weights as labels to the lines.

\begin{ex}
The weighted graph $G=(V, E,w)$ with
\begin{align*}
V&=\{0,1,2,3,4,5\}\\
w(5,0)&=w(0,5)=w(v,v+1)=w(v+1,v)=1   \quad\text{ for all } 0\le v\le 4\\
w(1,4)&=w(4,1)= w(2,5)=w(5,2)=2
\end{align*}
and all not defined function values of $w$ equaling $0$, is graphically represented as

\vspace{6pt}
\begin{figure}[\textwidth]
\centering
\begin{tikzpicture}[scale=0.9]
\node (0) at (-2,0) [circle,shade,draw,minimum size=8mm] {$0$};
\node (1) at (0,1.5) [circle,shade,draw,minimum size=8mm] {$1$};
\node (2) at (3,1.5) [circle,shade,draw,minimum size=8mm] {$2$};
\node (5) at (0,-1.5) [circle,shade,draw,minimum size=8mm] {$5$};
\node (4) at (3,-1.5) [circle,shade,draw,minimum size=8mm] {$4$};
\node (3) at (5,0) [circle,shade,draw,minimum size=8mm] {$3$};
\draw[line width=1pt] (0) to (1);;
\draw[line width=1pt] (1) to (2);
\draw[line width=1pt] (2) to (3);
\draw[line width=1pt] (3) to (4);
\draw[line width=1pt] (4) to (5);
\draw[line width=1pt] (5) to (0);
\draw[line width=1pt] (1) to (4);
\draw[line width=1pt] (2) to (5);
\node at (1.5,1.8) {$1$};
\node at (1.5,-1.8) {$1$};
\node at (-1.05,1) {$1$};
\node at (4.05,1) {$1$};
\node at (-1.05,-1.05) {$1$};
\node at (4.05,-1.05) {$1$};
\node at (2.35,-0.5) {$2$};
\node at (0.65,0.5) {$2$};
\end{tikzpicture}
\end{figure}

\end{ex}

Unless otherwise stated, throughout the thesis the term ``graph'' refers to that of Definition \ref{def:graph}, in some context we may say ``unweighted graph'' to point out that we do not talk about weighted graphs.

Most books first derive a lot of useful properties for unweighted graphs and then later generalize the proofs for weighted graphs. To avoid redundancy we will give definitions and results directly for weighted graphs in this section. Whatever property applies to weighted graphs applies of course also for an unweighted graph as the latter one can be regarded as a special case of a weighted graph with weight function
\begin{align*}
w(v,u) =
\begin{cases}
1 &\text{ if } (v,u) \in E \\
0 &\text{ if } (v,u) \notin E\text{.}
\end{cases}
\end{align*}

Therefore keep in mind that the fundamental graph vocabulary given by the next definition also applies for unweighted graphs accordingly:
\begin{definition}\label{def:graphBasics}
Let $G=(V,E,w)$ be a weighted graph.
\begin{itemize}
\item  Two vertices $u,v \in V$ are called \textsf{adjacent} or \textsf{neighbored} iff $(u,v) \in E$.
\item A \textsf{path} $p$ is a sequence of two or more vertices $p=(v_0,v_1, \dots v_n)$ such that $\forall i, 1\le i \le n: (v_{i-1}, v_i) \in E$. The vertices $v_0$ and $v_n$ are called \textsf{connected} via the path $p$. The \textsf{length} of the path $p$ is defined as $len(p)=n$.
\item A \textsf{connected component} of a graph $G=(V,E)$ is a subgraph $(V', E \cap V' \times V'), V' \subseteq V$ such that there is a path between any two distinct vertices $u, v \in V'$ and for all vertices $ v' \in V'$ and all vertices $v \in V \setminus V'$ it holds $(v',v) \notin E$. A graph is called \textsf{connected} iff it consists of just one connected component.
\item The \textsf{distance} $\dist(u,v)$ of two vertices $u,v \in V$ in a connected graph is defined as $0$ if the vertices are the same or otherwise as the minimum length of all paths connecting these two vertices.
\item The distance between vertex sets $X, Y\subseteq V$ is defined as $\dist(X,Y) = \min_{x\in X, y\in Y} \dist(x,y)$.
\item For a vertex set $W \subseteq V$ we define $\dist_{\rho}(W) := \left\{v\in V \,\big\vert\,  \dist(W, \{v\})\le \rho \right\}$.
\item The \textsf{diameter} of a connected graph $G$, denoted $diam$, is the largest distance between any two vertices of the graph.
\item The \textsf{degree} $d_v$ of a vertex $v \in V$ is defined as $d_v := \sum_{u \in V} w(v,u)$. Moreover let $d_{max}:=\max_{v\in V} d_v$ and $d_{min}:=\min_{v\in V} d_v$.
\item The \textsf{degree matrix} $T$ is the ${\vert V\vert}\times{\vert V\vert}$ diagonal matrix with the $(v,v)$-th entry having value $d_v$. Its ``inverse'' $T^{-1}$ is the diagonal matrix with the $(v,v)$-th entry having value $0$ if $d_v=0$ and $\frac{1}{d_v}$ otherwise.
\item A graph is called $d$-\textsf{regular}, iff all vertices have the same degree $d >0$.
\item The \textsf{volume} of a vertex set $S \subseteq V$ is defined as $\vol(S):= \sum_{v\in S} d_v$.
\item $\partial S:= \{ (u,v) \,\vert\, (u,v) \in E,\, u \in S, \, v\in \overline{S}\}$ is called the \textsf{edge boundary} of a vertex set $S \subseteq V$.
\item $\delta S :=\{v \,\vert\, (u,v)\in E,\, u\in S,\, v\in \overline{S}\}$ is called the \textsf{vertex boundary} of a vertex set $S \subseteq V$.
\end{itemize}
\end{definition}

Note that, with the exception of the definition of ``degree'', the above definitions do not even make use of the weight function. The degree of a vertex in the case of an unweighted graph reduces to the number of adjacent vertices. Similarily, in the next definition the adjacency matrix elements of an unweighted graph turn out to be simply $1$ or $0$ depending on whether the respective vertices are connected via an edge or not:

\begin{definition}
The \textsf{adjacency matrix} $A$ of a weighted graph $G=(V,E,w)$ is a linear operator on a $\vert V\vert$-dimensional vector space, called \textsf{graph space}, and defined as
\begin{align*}
A= \sum_{v,u\in V} w(v,u) \ket{u}\bra{v}
\end{align*}
with $\{\ket{v}\}_{v\in V}$ an orthonormal basis of the space, called the \textsf{vertex basis}.
\end{definition}

We will use the conventional short-hand
\begin{align*}
A_{u,v}&:=\braket{u | A | v}\\
f_v &:= \braket{v | f}
\end{align*}
for arbitrary operators $A$ and arbitrary vectors $\ket{f}$ for all $u,v \in V$.

\begin{definition}
The \textsf{standard Laplacian} $L$ of a weighted graph $G$ is defined as degree matrix minus adjacency matrix: $L= T-A$. The \textsf{normalized Laplcian} $\mathcal{L}$ is defined as $\mathcal{L}=T^{-1/2} L T^{-1/2}$.
\end{definition}

The elements of the normalized Laplacian can be calculated as follows
\begin{align*}
\mathcal{L}_{u,v}=
\begin{cases}
1 - \frac{w(v,v)}{d_v}\quad&\text{if } v=u \text{ and } d_v \ne 0\\
- \frac{w(v,u)}{\sqrt{d_v d_u}} &\text{if } v\ne u\\
0 &\text{otherwise.}
\end{cases}
\end{align*}
Since $L$ and $\mathcal{L}$ are symmetric, their eigevalues are real. $\mathcal{L}$ is called normalized Laplacian, because its eigenvalues are bounded by a constant as proven by the next theorem. But first we give a helpful expression for the two Laplacians in terms of the so-called Rayleigh quotient:
\begin{definition}
The \textsf{Rayleigh quotient} $R(M,\ket{x})$ of a matrix $M$ regarding a vector $\ket{x}$, $\ket{x} \ne 0$, is the expectation value of this matrix according to this vector and is given as follows:
\begin{align*}
R(M,\ket{x})=\frac{\braket{x | M | x}}{\braket{x | x}}\text{.}
\end{align*}
\end{definition}

\begin{lem}\label{lem:expectationL}
Let $L$ be the standard Laplacian of a weighted graph $G=(V,E,w)$ and $\ket{f}$ an arbitrary vector. Then it holds:
\begin{align*}
\braket{f | L | f} = \sum_{u \sim v} \left(f_v - f_u\right)^2 w(u,v)
\end{align*}
where we sum over all unordered pairs $u \sim v$ of adjacent vertices.
\end{lem}
\begin{proof}
\begin{align*}
\braket{f | L| f} &=\braket{f | T | f}- \braket{f | A | f}\\
&= \sum\limits_{v\in V} d_v f_v^2 - \sum\limits_{(u,v)\in E} w(u,v) f_v f_u\\
&= \sum\limits_{v\in V} \sum\limits_{u \in V} \left(w(u,v) f_v^2 -  w(u,v) f_v f_u\right)\\
&=\sum\limits_{u \sim v} \left(f_v^2 - 2 f_v f_u + f_u^2\right) w(u,v)\\
&= \sum\limits_{u \sim v}  \left(f_v - f_u\right)^2 w(u,v)\text{.}\qedhere
\end{align*}
\end{proof}

\begin{cor}\label{cor:rayleigh}
The Rayleigh quotient $R(\mathcal{L},\ket{g})$ with $\mathcal{L}$ being the normalized Laplacian of a connected weighted graph $G=(V,E,w)$ with more than one vertex fulfills the identity
\begin{align*}
R(\mathcal{L},\ket{g}) =  \frac{\sum\limits_{u \sim v} \left(f_v - f_u\right)^2 w(u,v)}{\sum\limits_{v \in V} f_v^2 d_v}
\end{align*}
with $\ket{f}:= T^{-1/2}\ket{g}$.
\end{cor}
\begin{proof}
As $G$ is a connected weighted graph with more than one vertex, $d_v \ne 0$ for all $v \in V$ and the degree matrix $T$ is invertible.
\begin{align*}
\frac{\braket{g | \mathcal{L} | g}}{\braket{g | g}}
&=\frac{\braket{f | L| f}}{\braket{f | T | f}}
= \frac{\sum\limits_{u \sim v}  \left(f_v - f_u\right)^2 w(u,v)}{\sum\limits_{v \in V} f_v^2 d_v}\text{.}\qedhere
\end{align*}
\end{proof}

\begin{thm}\label{thm:spectra}
\quad

\begin{enumerate}[label={\normalfont(\roman*)}]
\item The eigenvalues $\mu_i$ of the standard Laplacian $L$ of a connected weighted graph $G$ fulfill the relation $0 = \mu_0 < \mu_1 \le \dots \le \mu_{N-1} \le 2 d_{max}$ and the nondegenerate null-space is spanned by the equal distribution vector $\sum_{v\in V} \ket{v}=:\mathbf{1}$.
\item The eigenvalues $\lambda_i$ of the normalized Laplacian $\mathcal{L}$ of a connected weighted graph $G$ fulfill the relation $0 = \lambda_0 < \lambda_1 \le \dots \le \lambda_{N-1} \le 2$ and the nondegenerate null-space is spanned by $T^{1/2}\mathbf{1}$.
\item The eigenvectors and eigenvalues of a graph are given by the eigenvectors and -values of its connected components.
\end{enumerate}
\end{thm}
\begin{proof}
\quad

\begin{enumerate}[label={\normalfont(\roman*)}]
\item 
If $G$ is a graph with just one vertex, the standard Laplacian is the zero matrix and the statement obviously true. So let us assume from now on that $G$ has more than one vertex.

From Lemma \ref{lem:expectationL} we know
\begin{align*}
\braket{f | L | f} = \sum\limits_{u \sim v}  \left(f_v - f_u\right)^2 w(u,v)\text{,}
\end{align*}
from which we can conclude that all eigenvalues of $L$ are non-negative. An eigenvector with the eigenvalue $0$ requires each summand on the right side to be zero. As $G$ is connected, this is only fulfilled for $\ket{f}$ a multiple of the equal distribution vector $\mathbf{1}$. Since for all $v\in V$
\begin{align*}
(L \mathbf{1})_v = \sum_{u\in V} L_{v,u} = 0\text{,}
\end{align*}
$\mathbf{1}$ is indeed the nondegenerated eigenvector corresponding to the eigenvalue $0$ of $L$.

Now we derive the upper bound on the eigenvalues. Assume that $\ket{f}$ is a normalized eigenvector of $L$ corresponding to the eigenvalue $\mu$. It holds that
\begin{align*}
\mu &=\braket{f | L | f}\\
&= \frac{1}{2} \sum_{(u,v) \in E}\left(f_v - f_u\right)^2 w(u,v)\\
&\le \sum_{(u,v) \in E} \left( f_v^2 + f_u^2 \right) w(u,v)\\
&= 2 \sum_{v \in V} \sum_{u\in V} f_v^2 w(u,v)\\
&\le \sum_{v\in V} 2 d_v f_v^2\text{.}
\end{align*}
Since $f$ is normalized it follows directly that $\mu \le 2  d_{max}$.

\item
If $G$ is a graph with just one vertex, the normalized Laplacian is the zero matrix and the statement obviously true. So let us assume that $G$ has more than one vertex. As $G$ is connected, $d_v > 0$ for every vertex $v \in V$ and the degree matrix $T$ is invertible.

Since the null-space of $L$ is spanned by $\mathbf{1}$ according to (i), the null-space of $\mathcal{L}= T^{-1/2} L T^{-1/2}$ is spanned by the $T^{1/2}\mathbf{1}$.

For the upper bound let $\ket{g}$ be an eigenvector of $\mathcal{L}$ with the eigenvalue $\lambda$. By analogy to the calculation in (i) we can show:
\begin{align*}
\lambda&=\frac{\braket{g | \mathcal{L}| g}}{\braket{g|g}}\\
&= \frac{\braket{f | L| f}}{\braket{f|T|f}}\\
&\stackrel{(i)}{\le} \frac{\sum\limits_{v\in V} 2 d_v f_v^2}{\sum\limits_{v\in V} d_v f_v^2}\\
&= 2\text{.}
\end{align*}

\item
The statement follows directly from the fact that $\mathcal{L}$ (or $L$ or $A$) is blockdiagonal with one block for each connected component.\qedhere
\end{enumerate}
\end{proof}

\begin{cor}\label{cor:infRayleigh}
The second smallest eigenvalue $\lambda$ of the normalized Laplacian $ \mathcal{L}$ of a connected weighted graph $G$ is its lowest nonzero eigenvalue and is given by
\begin{align*}
\lambda &= \inf_{\ket{g} \perp T^{1/2} \mathbf{1}} \frac{\braket{g| \mathcal{L} | g}}{\braket{g | g}}\\
&= \inf_{\ket{f} \perp T\mathbf{1}} \frac{\sum\limits_{u\sim v}  \big(f_v - f_u\big)^2 w(u,v)}{\sum\limits_{v \in V} f_v^2 d_v}\text{.}\hspace{23.2em}\qed
\end{align*}
\end{cor}

In this thesis we are mostly interested in the eigenvalues and eigenvectors of the normalized Laplacian of a (weighted) graph. That's why we also call them just the eigenvalues and eigenvectors of the graph. But as it is sometimes easier to derive expressions for the spectra of the standard Laplacian or the adjacency matrix of a graph, it is worthwhile to know the relationships between the different spectra:

\begin{definition}
For a weighted graph we denote the increasing eigenvalues of the normalized Laplacian by $(\lambda_0, \lambda_1, \dots \lambda_{N-1})$, the increasing eigenvalues of the standard Laplacian by $(\mu_0, \mu_1, \dots\linebreak[3] \mu_{N-1})$ and the decreasing eigenvalues of the adjacency matrix by $(\alpha_0, \alpha_1, \dots \alpha_{N-1})$.
\end{definition}

\begin{prop}
For a $d$-regular weighted graph $G$ it holds:
\begin{align*}
\lambda_i = \frac{\mu_i}{d} = 1 - \frac{\alpha_i}{d}\text{.}
\end{align*}
\end{prop}
\begin{proof}
Since every vertex has degree $d>0$ the degree matrix is simply $T=d \,\mathbb{I}$ and according to the definitions of the different matrices
\begin{align*}
\mathcal{L} &= \frac{1}{d} L = \mathbb{I} - \frac{1}{d} A\text{.}\qedhere 
\end{align*}
\end{proof}

If a graph is not regular, the relationship between the spectra of normalized Laplacian, standard Laplacian and adjacency matrix are non-trivial. But one can derive the following theorem similar to \cite{zumstein} based on the Courant-Fischer theorem:

\begin{thm}\label{thm:relationship}
For a connected weighted graph $G$ with more than one vertex the following holds:
\begin{enumerate}[label={\normalfont(\roman*)}]
\item
${\displaystyle \quad\quad \frac{\mu_i}{d_{max}} \le \lambda_i \le \frac{\mu_i}{d_{min}}}$
\item
${\displaystyle \quad\quad 1-\frac{\alpha_i}{d_{min}} \le \lambda_i \le 1-\frac{\alpha_i}{d_{max}}}$.
\end{enumerate}
\end{thm}
\begin{proof}
\quad

\begin{enumerate}[label={\normalfont(\roman*)}]
\item
As $G$ is connected and has more than one vertex, $d_v >0$ for all $v \in V$ and $T$ is invertible. It holds that
\begin{align*}
R(\mathcal{L},\ket{g}) &=  \frac{\braket{g T^{-1/2} | L | T^{-1/2} g}}{\braket{g T^{-1/2} | T | T^{-1/2}g}}\\
&= \frac{\braket{g T^{-1/2} | L | T^{-1/2} g}}{d^* \braket{g T^{-1/2} | T^{-1/2}g}}\\
&= \frac{R(L, T^{-1/2}\ket{g})}{d^*}
\end{align*}
for some $d^*$ with $d_{min} \le d^* \le d_{max}$.

We now use the characterization of the $i$-th eigenvalue via the Courant-Fischer theorem (min-max theorem):
\begin{align*}
\lambda_i &= \min_{\substack{\text{linearly independent}\\ \ket{g_{i+1}}, \dots \ket{g_{N-1}}}} \max_{\substack{\ket{g_i} \ne 0\\ \ket{g_i} \perp  \ket{g_{i+1}},  \dots \ket{g_{N-1}}}} R(\mathcal{L},\ket{g_i})\\
&= \frac{1}{d^*} \min_{\substack{\text{linearly independent}\\ \ket{g_{i+1}}, \dots \ket{g_{N-1}}}} \max_{\substack{\ket{g_i} \ne 0\\ \ket{g_i} \perp  \ket{g_{i+1}},  \dots \ket{g_{N-1}}}} R(L,T^{-1/2} \ket{g_i})\text{.}
\end{align*}
Define $\ket{f_i} := T^{-1/2} \ket{g_i}$ and $\ket{f_j} := T^{1/2} \ket{g_j}$ for $i+1 \le j \le N-1$. Then it holds
\begin{align*}
\ket{g_i} \perp \ket{g_j} \quad\Longleftrightarrow\quad \ket{f_i} \perp \ket{f_j} \quad \forall\, i+1\le j \le N-1
\end{align*}
and
\begin{align*}
\ket{g_{i+1}}, \dots, \ket{g_{N-1}} \text{ lin. independent} \quad\Longleftrightarrow\quad\ket{f_{i+1}}, \dots, \ket{f_{N-1}} \text{ lin. independent}
\end{align*}
since $T$ is invertible. Thus we can rewrite
\begin{align*}
\lambda_i
&= \frac{1}{d^*} \min_{\substack{\text{linearly independent}\\ \ket{f_{i+1}}, \dots \ket{f_{N-1}}}} \max_{\substack{\ket{f_i} \ne 0\\ \ket{f_i} \perp  \ket{f_{i+1}},  \dots \ket{f_{N-1}}}} R(L,\ket{f_i})\\
&= \frac{\mu_i}{d^*}\text{.}
\end{align*}
Therefore
\begin{align*}
 \frac{\mu_i}{d_{max}} &\le \lambda_i \le \frac{\mu_i}{d_{min}}\text{.}
\end{align*}

\item
The inequality for the eigenvalues of the adjacency matrix follows analogously:
\begin{align*}
R(\mathcal{L},\ket{g}) &= \frac{\braket{g | \mathbb{I} | g}}{\braket{g | g}} - \frac{ \braket{g | T^{-1/2} A T^{-1/2} | g}}{\braket{g T^{-1/2} | T | T^{-1/2} g}}\\
&= 1 - \frac{R(A, T^{-1/2} \ket{g})}{d^*}
\end{align*}
for some $d^*$ with $d_{min} \le d^* \le d_{max}$.
\begin{align*}
\lambda_i &=  \min_{\substack{\text{linearly independent}\\ \ket{g_{i+1}}, \dots \ket{g_{N-1}}}} \max_{\substack{\ket{g_i} \ne 0\\ \ket{g_i} \perp  \ket{g_{i+1}},  \dots \ket{g_{N-1}}}} R(\mathcal{L},\ket{g_i})\\
&= 1- \frac{1}{d^*} \max_{\substack{\text{linearly independent}\\ \ket{g_{i+1}}, \dots \ket{g_{N-1}}}} \min_{\substack{\ket{g_i} \ne 0\\ \ket{g_i} \perp  \ket{g_{i+1}},  \dots \ket{g_{N-1}}}} R(A,T^{-1/2} \ket{g_i})\\
&=1- \frac{1}{d^*} \max_{\substack{\text{linearly independent}\\ \ket{f_{i+1}}, \dots \ket{f_{N-1}}}} \min_{\substack{\ket{f_i} \ne 0\\ \ket{f_i} \perp  \ket{f_{i+1}},  \dots \ket{f_{N-1}}}} R(A,\ket{f_i})\\
&= 1- \frac{\alpha_i}{d^*}
\end{align*}
and hence
\begin{align*}
1-\frac{\alpha_i}{d_{min}} &\le \lambda_i \le 1-\frac{\alpha_i}{d_{max}}\text{.}\qedhere
\end{align*}
\end{enumerate}
\end{proof}


\section{The spectral gap and expander graphs}

In this thesis the seond lowest eigenvalue of the normalized Laplacian of a (weighted) graph is of particular interest, justifying its own name:
\begin{definition}
Let $\mathcal{L}$ be the normalized Laplacian of a (weighted) graph $G$. Its second smallest eigenvalue $\lambda$ is called the \textsf{spectral gap} of the graph $G$.
\end{definition}

From now on we drop the notation $\lambda_1$ and simply write $\lambda$ or $\lambda(G)$ for the spectral gap of a graph $G$. Notice that in literature sometimes the spectral gap is defined by $\min \{ \lambda_1, \lambda_{N-1}-\lambda_{N-2}\}$, which equals the second largest absolute eigenvalue of the ``normalized adjacency matrix'' $D^{-1/2}AD^{-1/2}$.

We are particulary interested in the shrinking behaviour of the spectral gap while looking at a whole set or family of (weighted) graphs with strictly increasing vertex sets. If the spectral gap of an unweighted graph family is lower bounded by a constant, we have a special name for this family:
\begin{definition}
An $\epsilon$\textsf{-expander} is an element of a infinite graph family $\mathcal{A}$ with $\lambda(G) \ge \epsilon >0$ for all graphs $G \in \mathcal{A}$.
\end{definition}

The name ``expander'' graph comes from the fact that a spectral gap lower-bounded by a constant is equivalent to a constantly lower bounded edge expansion factor $h_G$ and vertex expansion factor $g_G$, which are defined in the following way:
\begin{definition}
For a nonempty set of vertices $S \subsetneq V$ of a connected graph $G$ with more than one vertex define
\begin{align*}
h_G(S) := \frac{\left\vert \partial S\right\vert}{\min \{ \vol(S), \vol(\overline{S})\}}\text{.}
\end{align*}
The \textsf{edge expansion factor} or \textsf{Cheeger constant} of the graph $G$ is defined as $h_G:= \min\limits_{\emptyset \ne S\subsetneq V} h_G(S)$.
\end{definition}

\begin{definition}
For a nonempty set of vertices $S \subsetneq V$ of a connected graph $G$ with more than one vertex define
\begin{align*}
g_G(S) := \frac{\vol(\delta S)}{\min \{ \vol(S), \vol(\overline{S})\}}\text{.}
\end{align*}
The \textsf{vertex expansion factor} of the graph $G$ is defined as $g_G:=\min\limits_{\emptyset \ne S\subsetneq V} g_G(S)$.
\end{definition}

The proof of the next theorem is based on \cite[chapter 2]{chung}.

\begin{thm}\label{thm:vertexExpansion}
For any connected graph $G$ with more than one vertex the following holds:
\begin{align*}
g_G \ge h_G \ge \frac{1}{2}\lambda\text{.}
\end{align*}
\end{thm}
\begin{proof}
The first inequality sign is easy to verify since every nonempty subset $S \subsetneq V$ fulfills
\begin{align*}
\vol(\delta S)&=\sum_{v \in \delta S}d_v
=\left\vert \{ (v,u) \in E \,\vert\, v \in \delta S\}\right\vert
\ge\left\vert \{ (v,u) \in E \,\vert\, v \in \delta S,\, u \in S\}\right\vert
= \left\vert\partial S \right\vert
\end{align*}
and hence $g_G \ge h_G$.

For the second inequality sign let's take a look at the vector $\ket{f}$ with
\begin{align*}
f_v=
\begin{cases}
\frac{1}{\vol(S)} \quad &\text{ if } v\in S \\
\frac{-1}{\vol(\overline{S})} \quad &\text{ if } v\in \overline{S}
\end{cases}
\end{align*}
with $S \subsetneq V$ the vertex set that achieves the Cheeger constant: $g_G(S) = \min\limits_{\emptyset \ne S'\subsetneq V} g_G(S')$.

$\ket{f}$ is orthogonal to $T \ket{1}$, thus according to Lemma \ref{cor:infRayleigh}  
\begin{align*}
\lambda &\le \frac{\sum\limits_{u \sim v} (f_v - f_u)^2}{\sum\limits_v f_v^2 d_v}\\
&= \frac{\vert\partial S\vert \left(\frac{1}{\vol(S)} + \frac{1}{\vol(\overline{S})}\right)^2}{\frac{1}{\vol(S)} + \frac{1}{\vol(\overline{S})}}\\
&\le 2 \frac{\vert\partial S\vert}{\min\{\vol(S), \vol(\overline{S})\}}\\
&= 2 h_G\text{.}\qedhere
\end{align*}
\end{proof}

We have shown that eigenvalue expansion ($\lambda \ge \epsilon >0$) implies edge and vertex expansion ($g_G \ge h_G \ge \epsilon >0$). As mentioned above the opposite also holds, but since for our purposes only the first direction is interesting, we omit to proof the other direction and refer the interested reader to \cite[chapter 2]{chung}.

The next technical lemma is a consequence of vertex expansion and will help us later to derive some important restrictive results for the graph families in our constructions:
\begin{lem}\label{lem:expVertexExp}
Let $G=(V,E)$ be a graph family depending on the parameter $L$. Let $a$ be a positive constant and $W \subseteq V$ such that $\vol\left(\dist_{\frac{L}{a}+1}(W)\right) \le \frac{\vol(V)}{2}$. Then it holds
\begin{align*}
\frac{\vol(V)}{\vol(W)}\ge 2 \left( 1+ \frac{\lambda}{2}\right)^{\frac{L}{a}}\text{.}
\end{align*}
\end{lem}
\begin{proof}
From Theorem \ref{thm:vertexExpansion} about vertex expansion we know that for any $\dist_\rho(W)$, $0 \le \rho \le \frac{L}{a}+1$ the following inequality holds:
\begin{align*}
\frac{1}{2} \lambda \le g_G \le \frac{\vol\left(\delta\dist_\rho(W)\right))}{\vol\left(\dist_\rho(W)\right)}
\end{align*}
and hence
\begin{align*}
\vol\left(\dist_\rho(W) \cup \delta \dist_\rho(W)\right) \ge \left( 1+ \frac{\lambda}{2}\right) \vol\left(\dist_\rho(W)\right)\text{.}
\end{align*} 
With this equality we can derive successively
\begin{align*}
\frac{1}{2} \vol(V) &\ge \vol\left( \dist_{\frac{L}{a}+1}(W)\right)\\
&\ge \left( 1+ \frac{\lambda}{2}\right) \vol\left( \dist_{\frac{L}{a}}(W)\right)\\
&\ge \left( 1+ \frac{\lambda}{2}\right)^{\left\lfloor \frac{L}{a}+1\right\rfloor} \vol\left( W \right)\\
\frac{\vol(V)}{\vol(W)} &\ge 2 \left( 1+ \frac{\lambda}{2}\right)^{\frac{L}{a}}\text{.}\qedhere
\end{align*}
\end{proof}

We can conclude from the above lemma that the diameter of an expander graph is always upper bounded by a logarithmic function in the size of the vertex set (consider $L$ in the above lemma as diameter). In addition to this result we will also use later the following diameter bound by Alon and Milman \cite{alon} since it offers an even stricter bound for graphs with constant degree ratio:
\begin{thm}\label{thm:diameter}
Let $G=(V,E)$ be a connected graph with $\left\vert V\right\vert = N > 1$ and $\mu$ be the second smallest eigenvalue of the standard Laplacian. Then
\begin{align*}
\diam \le 2 \sqrt{2d_{max} \frac{1}{\mu}} \log_2 N\text{.}
\end{align*}
\end{thm}
 
Knowing the relationship between the spectra of standard and normalized Laplacian from Theorem \ref{thm:relationship} we can conclude directly

\begin{cor}\label{cor:diameter}
Let $G=(V,E)$ be a connected graph with $\left\vert V\right\vert = N > 1$ and $\lambda$ its spectral gap. Then
\begin{align*}
\diam \le 2 \sqrt{2\frac{d_{max}}{d_{min}} \frac{1}{\lambda}} \log_2 N\text{.}\hspace{25.8em}\qed
\end{align*}
\end{cor}


\section{Cayley graphs}

In this section we introduce Cayley graphs as a special kind of \textit{unweighted} graphs whose eigenvectors and eigenvalues can be calculated with some background knowledge about group theory. In order to do so we first need to recall some properties about characters of a group:

\begin{definition}
A \textsf{character} $\chi$ of a finite abelian group $H$ is a group homomorphism $\chi : G \rightarrow \mathbb{C}^\times$.
\end{definition}

As every element of a finite group has finite order according to Langrange's theorem, the image of a character are roots of unity and thus actually even comprised in $S^1$.

\begin{lem}\label{lem:charLin}
Distinct characters of a finite abelian group $H$ are linearly independent.
\end{lem}
\begin{proof}
We follow the induction proof of \cite{conrad}: Let $\chi_1,\chi_2, \dots \chi_n$ be distinct characters $H \rightarrow \mathbb{C}^\times$. For $n=1$ the character(s) is trivially linear independant. So let $n \ge 2$ and suppose $\chi_1, \dots \chi_{n-1}$ are linearly independent characters. Consider the linear independence relation
\begin{align}
\sum_{j=1}^n c_j \chi_j(g) = 0 \quad \quad \forall \,g\in H\text{.} \label{eq:linIndep}
\end{align}
As this equation should hold for any group element, we can equivalently write
\begin{align*}
\sum_{j=1}^n c_j \chi_j(g' g) = \sum_{j=1}^n c_j \chi_j(g') \chi_j(g)=0 \quad\quad \forall\, g', g\in H\text{.}
\end{align*}
We subtract the equation
\begin{align*}
\chi_n(g') \sum_{j=1}^n c_j \chi_j(g) = 0 \quad\quad\forall\, g', g\in H
\end{align*}
to get the condition
\begin{align*}
\sum_{j=1}^{n-1} c_j \big(\chi_j(g')-\chi_n(g')\big) \chi_j(g) \quad\quad\forall\, g', g\in H\text{.}
\end{align*}
As $\chi_1, \dots \chi_{n-1}$ are linear independent according to the induction assumption, $c_j\big(\chi_j(g')-\chi_n(g')\big)$ has to vanish for all $j \in\{ 1, \dots n-1\}$ and all  $g'\in H$. But as for every $j \in\{ 1, \dots n-1\}$ $\chi_j$ and $\chi_n$ are distinct characters, there exists an element $g_j \in H$, such that $\chi_j(g_j) \ne \chi_n(g_j)$. Consequently $c_j = 0$ for all $j \in\{ 1, \dots n-1\}$ and according to \eqref{eq:linIndep} also $c_n=0$. This means that $\chi_1, \dots \chi_n$ are linearly independent.
\end{proof}

\begin{lem}\label{lem:nChar}
A finite abelian group $H$ has exactly ${\vert H \vert}$ distinct characters.
\end{lem}
\begin{proof}
For now let $H$ be cyclic, $g \in H$ a generator of the group and $\chi$ a character. Because of $1=\chi(e)= \chi(g^{\vert H\vert})=\chi(g)^{\vert H\vert}$, the image of the generator has to be a ${\vert H \vert}$-root of unity, of which there are exactly ${\vert H\vert}$. As each possible image of the generator fully defines a distinct character, the group $H$ has exactly ${\vert H \vert}$ distinct characters.

Now assume that $H$ is an arbitrary finite abelian group. According to the structure theorem of finite abelian groups (see as reference theorem 14.2 in \cite{saracino}), $H$ is isomorphic to a direct product of cyclic groups $H_i, \, i\in\{1,\dots k\}$ with ${\vert H\vert}=\prod_{i=1}^k {\vert H_i\vert}$:
\begin{align*}
H \simeq \bigotimes_{i=1}^k H_i\text{.}
\end{align*}
Let $\chi_{i,1}, \chi_{i,2}, \dots \chi_{i,{\vert H_i\vert}}$ be the characters of $H_i$. It can easily be checked that $\bigotimes_{i=1}^k \chi_{i,j_i}$ with $j_i \in \{1, \dots {\vert H_i\vert}\}$ are the characters of $\bigotimes_{i=1}^k H_i$. Thus $H$ has $\prod_{i=1}^k {\vert H_i\vert}={\vert H\vert}$ distinct characters.
\end{proof}

Now we can define Cayley graphs and use the above results to derive a useful theorem about their eigenvectors and -values:

\begin{definition} 
For a finite abelian group $H$ and a symmetric subset $S$ (i.e. $g \in S \Leftrightarrow g^{-1} \in S$) the \textsf{Cayley graph} $\cay(H,S)$ is definded as the graph with vertex set $H$ and edge set $E=\{(u,v) \,\vert\, uv^{-1} \in S\}$. 
\end{definition}

Notice that a Cayley graph is indeed a correctly defined undirected graph and that it is ${\vert S\vert}$-regular.

The following theorem adapted from \cite{trevisan} gives the promised result about the eigenvectors and -values of the normalized Laplacian of a Cayley graph:

\begin{thm}\label{thm:eigenChar}
Let $\chi$ be a character of the finite abelian group $H$. Then $\chi$ is an eigenvector of the normalized Laplacian of the Cayley graph $\cay(H,S)$ with eigenvalue
\begin{align*}
1 -\frac{1}{\vert S \vert} \sum_{s \in S} \chi(s)\text{.}
\end{align*}
\end{thm}
\begin{proof}
Let A be the adjacency matrix of the Cayley graph $\cay(H,S)$ and $\chi$ a character of the group $H$. As we understand $\chi$ as both, as function and vector, we identify $\chi(v)$ with $\chi_v$.
\begin{align*}
(A\chi)_v &= \sum_{u\in G} A_{u,v} \chi_v\\
&=\sum_{u: uv^{-1}\in S} \chi(u)\\
&=\sum_{s \in S} \chi(sv)\\
&= \left(\sum_{s \in S} \chi(s)\right) \chi_v\text{.}
\end{align*}
Thus $\chi$ is an eigenvector of the adjacency matrix corresponding to the eigenvalue $\sum_{s \in S} \chi(s)$. As $\cay(H,S)$ is ${\vert S \vert}$-regular, $\chi$ is also an eigenvector of the normalized Laplacian to the eigenvalue $1-\frac{1}{\vert S \vert} \sum_{s \in S} \chi(s)$. 
\end{proof}

With the knowledge from Lemmata \ref{lem:charLin} and \ref{lem:nChar} that the group $H$ has ${\vert H \vert}$ linear independent characters, we can conclude that they form an eigenbasis of the normalized Laplacian.

Interestingly the eigenvectors of a Cayley graph are fully defined by the group $H$, whereas the specific choice of a subset $S$ only affects the eigenvalues.

We will finish this section by giving a formula for the eigenvectors and -values of a certain family of Cayley graphs, namely those with the group $\bigotimes_{i=1}^k(\mathbb{Z}/2\mathbb{Z},+) =  \big( \{0,1\}^k, \oplus \big)$. We will understand the elements of the group as bit strings of length $k$ and hence the group operation, which we will denote by $\oplus$, is just bitwise addition or bitwise XOR.

\begin{thm}\label{thm:cayleyBitwise}
Let $H = \big( \{0,1\}^k, \oplus \big)$ be the group with bitwise addition on the set of bit strings with length $k$ and $S \subseteq H$ a symmetric subset. Then the eigenvectors $(\chi_x)_{x\in H}$ of the normalized Laplacian of $\cay(H,S)$ have the form
\begin{align*}
\chi_{x,y}= (-1)^{\langle x, y\rangle}
\end{align*}
with the corresponding eigenvalues
\begin{align*}
1 - \frac{1}{\vert S \vert} \sum_{s\in S} (-1)^{\langle x, s\rangle}
\end{align*}
where $\langle \cdot,\cdot\rangle$ denotes the scalar product of two bit vectors modulo $2$.
\end{thm}
\begin{proof}
For all $x\in H$ the above defined $\chi_x$ is a character of $H$ since for all $y,z \in H$ the following holds:
\begin{align*}
\chi_x(y \oplus z) &= (-1)^{\langle x, y \oplus z\rangle}
=(-1)^{{\langle x, y \rangle} \oplus {\langle x, z\rangle}}
=(-1)^{\langle x, y \rangle} (-1)^{\langle x, z\rangle}
= \chi_x(y) \chi_x(z)\text{.}
\end{align*}
It is left to check that we have really defined ${\vert H \vert}$ distinct characters. Assume $x,y \in H, \, x\ne y$, hence w.l.o.g. $x_i=0$ and $y_i=1$ for some bit position $i$. Let $e_i$ denote the bit string with all zeroes and and the only $1$ in position $i$. Then
\begin{align*}
\chi_x(e_i) = 1 \ne -1 = \chi_y(e_i)\text{.}
\end{align*}
So $(\chi_x)_{x\in H}$ are indeed ${\vert H \vert}$ distinct characters and hence linearly independant according to Lemma \ref{lem:charLin}. With Theorem \ref{thm:eigenChar} it follows that $(\chi_x)_{x\in H}$ is an eigenbasis of $\mathcal{L}(H)$ with the corresponding eigenvalues
\begin{align*}
1 - \frac{1}{\vert S \vert} &\sum_{s\in S} (-1)^{\langle x, s\rangle}\text{.} \qedhere
\end{align*}
\end{proof}


\section{Contractions and coverings}\label{sec:contraction}

In this section we return to weighted graphs and introduce covering and contraction, two useful tools that allow us to transform a weighted graph into another one with a related spectrum.

\begin{definition}\label{def:contraction}
Let $G=(V',E',w')$ and $H=(V,E,w)$ be weighted graphs. $H$ is called a \textsf{contraction} of  $G$, iff there exists a surjective \textsf{contraction function} $c: V' \rightarrow V$ such that for all $x,y \in V$
\begin{align*}
w(x,y) = \sum_{\substack{a \in c^{-1}(x) \\b \in c^{-1}(y)}} w'(a,b)\text{.}
\end{align*}
\end{definition}

If $c^{-1}(x)$ comprises more than one element, we say that the vertices of $c^{-1}(x)$ are contracted to the vertex $x$.

\begin{lem}\label{lem:contractionDegree}
Let $H=(V,E,w)$ be a contraction of $G=(V',E',w')$ via the contraction function $c$. The degree of a vertex $x$ in $H$ be denoted by $d_x$ and the degree of a vertex $a$ in $G$ by $d'_a$. Then it holds for all $x \in V$:
\begin{align*}
d_x = \sum_{a \in c^{-1}(x)} d'_a
\end{align*}
\end{lem}
\begin{proof}
\begin{align*}
d_x &= \sum_{y \in V} w(x,y)
= \sum_{y \in V} \sum_{b \in c^{-1}(y)}\sum_{a \in c^{-1}(x)}  w'(a,b)
= \sum_{a \in c^{-1}(x)} \sum_{b \in V'} w'(a,b)
=\sum_{a \in c^{-1}(x)} d'_a\text{.} \qedhere
\end{align*}
\end{proof}

\begin{thm}\label{thm:contractionGap}
Let $G=(V',E',w')$ and $H=(V,E,w)$  be connected weighted graphs with more than one vertex and $H$ a contraction of $G$. Then the spectral gaps fulfill the following inequality:
\begin{align*}
\lambda(G) \le \lambda(H)\text{.}
\end{align*}
\end{thm}
\begin{proof}
We denote the contraction function by $c$, the degree of a vertex $x$ in $H$ by $d_x$, the degree of a vertex $a$ in $G$ by $d'_a$ and the degree matrices by $T$ and $T'$, respectively. Let $T^{1/2}\ket{f}$ be the eigenvector of $\mathcal{L}(H)$ with the eigenvalue $\lambda(H)$. Define $\ket{f'}: V' \rightarrow \mathbb{C}$ by
\begin{align*}
f'_a = f_{c(a)}\text{.}
\end{align*}
The fact that $\ket{f} \perp T\mathbf{1}$ implies directly $\ket{f'} \perp T'\mathbf{1}$ since
\begin{align*}
\sum_{a \in V'} f'_a d'_a
= \sum_{x \in V} \sum_{a \in c^{-1}(x)}  f_{c(a)} d'_a
=\sum_{x \in V} f_x d_x\text{.}
\end{align*}
Hence according to equation \eqref{cor:infRayleigh} the spectral gap of $G$ can be bounded by: 
\begin{align*}
 \lambda(G)
&= \frac{\sum\limits_{x,y \in V}\sum\limits_{\substack{a \in c^{-1}(x)\\ b \in c{-1}(y)}}  \big(f'_a - f'_b\big)^2 w'(a,b)}{\sum\limits_{x \in V}\sum\limits_{a \in c^{-1}(x)} (f_a')^2 d'_a}\\
&= \frac{\sum\limits_{x,y \in V} \big(f_x - f_y\big)^2 \sum\limits_{\substack{a \in c^{-1}(x)\\ b \in c{-1}(y)}} w'(a,b)}{\sum\limits_{x \in V} f_x^2\sum\limits_{a \in c^{-1}(x)} d'_a}\\
&= \frac{\sum\limits_{x,y \in V} \big(f_x - f_y\big)^2 w(x,y)}{\sum\limits_{x \in V} f_x^2 d_x}\\
&= \lambda(H)\text{.}\qedhere
\end{align*}
\end{proof}

Contraction clearly allows us to reduce the number of vertices in a graph by keeping the lower bound on the spectral gap. Covering is a similar tool which allows us to reduce the number of vertices in a graph while preserving the adjacency spectrum. Be aware that in some literature the above defined contraction is called covering. We instead follow with our definition the method presented in \cite{osborne}.

\begin{definition}\label{def:covering}
Let $G=(V',E',w')$ and $H=(V,E,w)$ be weighted graphs. $G$ is called a \textsf{covering} of $H$, iff there exists a surjective \textsf{covering function} $c: V' \rightarrow V$ that fulfills the following two properties:
\begin{enumerate}[label={\normalfont(\roman*)}]
\item For all $x,y \in V$ the following holds:
\begin{align*}
w(x,y) = \sum_{\substack{a \in c^{-1}(x) \\b \in c^{-1}(y)}} \frac{1}{\sqrt{\left\vert c^{-1}(x)\right\vert \left\vert c^{-1}(y)\right\vert}} w'(a,b)\text{.}
\end{align*}
\item For all $y \in V$ and for all $a,a^* \in V'$ with $c(a)=c(a^*)$ it holds:
\begin{align*}
\sum_{b \in c^{-1}(y)} w'(a,b) =\sum_{b \in c^{-1}(y)} w'(a^*,b)\text{.}
\end{align*}
\end{enumerate}
\end{definition}

\begin{definition}
Let $G=(V',E',w')$ be a covering of $H=(V,E,w)$ via the covering function $c$.Then the \textsf{pull-back operator} $P: V' \rightarrow V$ is defined as
\begin{align*}
P:=\sum_{\substack{x \in V\\a \in c^{-1}(x)}} \frac{1}{\sqrt{\left\vert c^{-1}(x) \right\vert}} \ket{x}\bra{a} \text{.}
\end{align*}
\end{definition}

\begin{lem}\label{lem:coveringP}
Let $G=(V',E',w')$ be a covering of $H=(V,E,w)$ with the pull-back operator $P$ and $A(G)$ and $A(H)$ the adjacency matrices of $G$ and $H$ respectively. Then
\begin{align*}
A(H) P = PA(G) \text{.}
\end{align*}
\end{lem}
\begin{proof}
\begin{align*}
A(H)P &= \sum_{x \in V} \sum_{\substack{y \in V\\b \in c^{-1}(y)}} \frac{1}{\sqrt{\left\vert c^{-1}(y) \right\vert}} w(x,y) \ket{x}\bra{b}\\
& \overset{\ref{def:covering} \text{(i)}}{=} \sum_{x \in V} \sum_{\substack{y \in V\\b \in c^{-1}(y)}} \frac{1}{\left\vert c^{-1}(y) \right \vert \sqrt{\left\vert c^{-1}(x) \right\vert}}\underbrace{\sum_{c\in c^{-1}(y)} \sum_{a\in c^{-1}(x)} w'(c,a)}_{\substack{= \left\vert c^{-1}(y) \right \vert \sum\limits_{a\in c^{-1}(x)} w'(b,a)\\ \text{according to \ref{def:covering} (ii)}}} \ket{x}\bra{b}\\
& = \sum_{\substack{x \in V \\ a\in c^{-1}(x)}} \sum_{b \in V'} \frac{1}{\sqrt{\left\vert c^{-1}(x) \right\vert}} w'(b,a) \ket{x}\bra{b}\\
&=P A(G)\text{.}\qedhere
\end{align*}
\end{proof}

\begin{thm}\label{thm:covering}
Let $G=(V',E',w')$ be a covering of $H=(V,E,w)$ and $A(G)$ and $A(H)$ the adjacency matrices of $G$ and $H$ respectively. Then $\alpha$ is an eigenvalue of A(G) if and only if $\alpha$ is an eigenvalue of $A(H)$.
\end{thm}
\begin{proof}
\quad

\begin{description}
\item[``$\Rightarrow$'']
Let $\ket{f}$ be an eigenvector of $A(G)$ with eigenvalue $\alpha$. Using Lemma \ref{lem:coveringP} one can easily see
\begin{align*}
A(H) P \ket{f} = P A(G) \ket{f} = \alpha P \ket{f}\text{,}
\end{align*}
hence $P\ket{f}$ is an eigenvector of $A(H)$ with eigenvalue $\alpha$.
\item[``$\Leftarrow$'']
Now let $\ket{f}$ be an eigenvector of $A(H)$ with eigenvalue $\alpha$. The adjoint of the expression in Lemma \ref{lem:coveringP} gives $P^\dagger A(H) = A(G) P^\dagger$ and therefore
\begin{align*}
A(G) P^\dagger \ket{f} = P^\dagger A(H) \ket{f} = \alpha P^\dagger \ket{f}\text{.}
\end{align*}
Thus $P^\dagger\ket{f}$ is an eigenvector of $A(G)$ with eigenvalue $\alpha$.\qedhere
\end{description}
\end{proof}

We have seen now that the spectra of the adjacency matrices of a graph and its covering graph are the same except for the multiplicities. As $\dim(A(H)) \le \dim(A(G))$ the multiplicity of an eigenvalue $\lambda$ of $A(H)$ is less than or equal  to the multiplicity of $\lambda$ in $A(G)$. Therefore a graph can only be the nontrivial cover of another graph, if its spectrum is degenerated.

In the case of contraction we could conclude from the spectral gap of a graph onto the spectral gap of its contracted graph. In the case of covering Theorem \ref{thm:covering} seems at first sight like an even stronger statement since it gives us the equivalence between eigenvalues. But unfortunately this equivalence refers to the eigenvalues of the adjacency matrix and not to the eigenvalues of the normalized Laplacian. But looking at specific graphs the relations between the different spectra stated in Theorem \ref{thm:relationship} often help to conclude from the adjacency gap bounds for the spectral gap of the normalized Laplacian.


\section{Parallel transport networks}

\begin{definition}\label{def:PTN}
A \textsf{parallel transport network} $\mathcal{G}=(V,E,w,\mathcal{T},\mathcal{U},\pi)$ consists of
\begin{itemize}
\item a weighted graph $(V,E,w)$, the so-called \textsf{underlying graph},
\item a \textsf{time map} $\mathcal{T}: V \rightarrow \{0,\dots,L\}$ such that $\forall (u,v) \in E: \left|\mathcal{T}(u)-\mathcal{T}(v)\right|\le 1$,
\item a map $\mathcal{U}: \{1,\dots,L\} \rightarrow U(k)$ that maps every \textsf{time step} $t \in  \{1,\dots,L\}$ to an element of the unitary group of degree $k$
\item and a representation $\pi: U(k) \rightarrow \mathcal{B}(\mathcal{H})$ of the unitary group in the space of bounded linear operators on a $k$-dimensional Hilbert space $\mathcal{H}$.
\end{itemize}
\end{definition}

\begin{rem}
We will use the short-hand $U_t:=\pi(\mathcal{U}(t))$. For the purpose of generality in the above definition $\mathcal{U}$ does not map $t$ directly to the Hilbert space operator $U_t$ but to an element of the unitary group $U(k)$ with an arbitrary representation in $\mathcal{B}(\mathcal{H})$. However, in this thesis we always assume the fundamental representation since $\pi(\mathcal{U}(t))$ has to equal arbitrary unitary gates of quantum circuits in our later applications.
\end{rem}

\begin{definition}
Let $\mathcal{G}=(V,E,w,\mathcal{T},\mathcal{U},\pi)$ be a parallel transport network, $\mathcal{T}:V \rightarrow \{0, \dots L\}$. We define for all $0 \le t \le L$:
\begin{align*}
V_t := \mathcal{T}^{-1}(t)\text{.}
\end{align*}
\end{definition}

\begin{definition}\label{def:unitaryPath}
A \textsf{path} $p = (v_0, v_1, \dots, v_n) $ of a parallel transport network is a path of its underlying graph. The \textsf{associated unitary of a path} $(v_0,v_1)$ of length $1$ (also called the \textsf{associated unitary of the edge} $(v_0,v_1)$)  is defined as
\begin{align*}
U((v_0,v_1))=
\begin{cases}
\mathbb{I} \quad\quad &\text{if } \mathcal{T}(v_0) = \mathcal{T}(v_1)\\
U_{\mathcal{T}(v_1)} \quad\quad &\text{if } \mathcal{T}(v_0) +1 = \mathcal{T}(v_1)\\
U_{\mathcal{T}(v_0)}^\dagger \quad\quad &\text{if } \mathcal{T}(v_0) -1= \mathcal{T}(v_1)\text{.}
\end{cases}
\end{align*}
The \textsf{associated unitary} $U(p)$ of a path $p= (v_0, v_1, \dots, v_n) $ with $\len(p) \ge 2$ is given by:
\begin{align*}
U(p) := U(v_{n-1},v_n) \dots U(v_1, v_2) U(v_0,v_1)\text{.}
\end{align*}
\end{definition}

From the above definition of the associated unitary of an edge follows directly:
\begin{prop}
For every edge $(v,u)\in E$ it holds that $U(v,u) = U(u,v)^\dagger$.\hspace{10em}\qed
\end{prop}

By induction over the length of a path one can easily verify even the more general statement:

\begin{prop}
Given a path $p = (v_0, v_1, \dots, v_n)$ with $\mathcal{T}(v_0)=t_0$ and $\mathcal{T}(v_n)=t_n$ its associated unitary equals
\begin{align*}
U(p)=
\begin{cases}
U_{t_n} U_{t_n -1} \dots U_{t_0+1} \quad\quad& \text{if } t_0 < t_n\\
U_{t_n +1}^\dagger \dots U_{t_0 -1}^\dagger U_{t_0}^\dagger \quad\quad& \text{if } t_0 > t_n\\
\mathbb{I} \quad\quad& \text{if } t_0 = t_n\text{.}\hspace{18.3em}\qed
\end{cases}
\end{align*}
\end{prop}

We see that every path from a vertex $v_0$ to a vertex $v_n$ has the same associated unitary. To be more precise the associated unitary of a path does not even depend on the specific inital vertex $v_0$ and final vertex $v_n$, but just on the values $\mathcal{T}(v_0)$ and $\mathcal{T}(t_n)$. We will consider $\mathcal{T}$ in this thesis as a function that maps vertices to certain \textsf{time steps}. Hence the associated unitary of a path just depends on the time steps of its initial and final vertex.

We will define now the adjacency and Laplacian matrices of a parallel transport network and using the properties presented above, we show that they are closely related to the adjacency and Laplacian matrix of the underlying graph:

\begin{definition}
The \textsf{adjacency matrix} $A(\mathcal{G})$ of a parallel transport network $\mathcal{G}=(V,E,w,\mathcal{T}, \mathcal{U}, \pi)$, $\mathcal{U}: U(k) \rightarrow \mathcal{B}(\mathcal{H})$, is a linear operator on the tensor product of a $\vert V \vert$-dimensional \textsf{graph space} and a $k$-dimensional \textsf{computation space} and defined as
\begin{align*}
A(\mathcal{G})= \sum_{(v,u)\in E} w(v,u)\ket{u}\bra{v} \otimes U((v,u))\text{.}
\end{align*}
\end{definition}

The adjacency matrix of a parallel transport network is the sum over the edge operators $w(v,u)\linebreak[3]\ket{u}\bra{v}$ tensored by the associated unitary. That is why we represent the network pictorially as a graph with the associated unitaries as edge labels (multiplied by the weights if the underlying graph is weighted). Because the edge $(u,v)$ has label $U^\dagger$ if $(v,u)$ has label $U$, it is no longer sufficient to draw just one undirected line between two vertices. Instead we either draw two directed edges with the labels $U$ and $U^\dagger$ respectively or only one directed edge with its associated unitary and know that the opposite edge is implied by definition.

\begin{ex}
$\mathcal{G}=(G, \mathcal{T}, \mathcal{U},\pi)$ with $G=(V,E)$ an unweighted graph and
\begin{align*}
V&=\{0,1,2,a,b\}\\
E&=\{(0,1), (1,0), (1,2), (2,1), (a,b), (b,a), (0,a), (a,0), (b,2), (2,b)\}\\
\mathcal{T}(0)&=\mathcal{T}(a)=0, \, \mathcal{T}(1)=\mathcal{T}(b)=1,\, \mathcal{T}(2)=2
\end{align*}
is graphically represented as

\vspace{6pt}
\begin{minipage}{\textwidth}
\begin{tikzpicture}[scale=0.9]
\node (0) at (0,0) [circle,shade,draw,minimum size=8mm] {$0$};
\node (1) at (3,0) [circle,shade,draw,minimum size=8mm] {$1$};
\node (2) at (6,0) [circle,shade,draw,minimum size=8mm] {$2$};
\node (a) at (0,-3) [circle,shade,draw,minimum size=8mm] {$a$};
\node (b) at (3,-3) [circle,shade,draw,minimum size=8mm] {$b$};
\draw[->, line width=1pt] (0) to[out=10, in=170] (1);
\draw[->, line width=1pt] (1) to[out=190, in=350] (0);
\draw[->, line width=1pt] (1) to[out=10, in=170] (2);
\draw[->, line width=1pt] (2) to[out=190, in=350] (1);
\draw[->, line width=1pt] (a) to[out=10, in=170] (b);
\draw[->, line width=1pt] (b) to[out=190, in=350] (a);
\draw[->, line width=1pt] (0) to[out=280, in=80] (a);
\draw[->, line width=1pt] (a) to[out=100, in=260] (0);
\draw[->, line width=1pt] (b) to[out=55, in=215] (2);
\draw[->, line width=1pt] (2) to[out=235, in=35] (b);
\node at (1.5,0.5) {$U_1$};
\node at (1.5,-0.5) {$U_1^\dagger$};
\node at (4.5,0.5) {$U_2$};
\node at (4.5,-0.5) {$U_2^\dagger$};
\node at (0.5,-1.5) {$\mathbb{I}$};
\node at (-0.5,-1.5) {$\mathbb{I}$};
\node at (1.5,-2.5) {$U_1$};
\node at (1.5,-3.5) {$U_1^\dagger$};
\node at (3.5,-1.5) {$U_2$};
\node at (5,-2) {$U_2^\dagger$};
\end{tikzpicture}
\raisebox{5.5em}{\quad or\quad}
\raisebox{0.1\height}{
\begin{tikzpicture}[scale=0.9]
\node (0) at (0,0) [circle,shade,draw,minimum size=8mm] {$0$};
\node (1) at (3,0) [circle,shade,draw,minimum size=8mm] {$1$};
\node (2) at (6,0) [circle,shade,draw,minimum size=8mm] {$2$};
\node (a) at (0,-3) [circle,shade,draw,minimum size=8mm] {$a$};
\node (b) at (3,-3) [circle,shade,draw,minimum size=8mm] {$b$};
\draw[->, line width=1pt] (0) to (1);;
\draw[->, line width=1pt] (1) to (2);
\draw[->, line width=1pt] (a) to (b);
\draw[->, line width=1pt] (0) to (a);
\draw[->, line width=1pt] (b) to (2);
\node at (1.5,0.2) {$U_1$};
\node at (4.5,0.2) {$U_2$};
\node at (0.2,-1.5) {$\mathbb{I}$};
\node at (1.5,-2.8) {$U_1$};
\node at (4,-1.5) {$U_2$};
\end{tikzpicture}}
\end{minipage}

\end{ex}

\begin{definition}
The \textsf{degree matrix} $T(\mathcal{G})$ of a parallel transport network $\mathcal{G}=(G, \mathcal{T},\mathcal{U},\pi)$ is defined as
\begin{align*}
T(\mathcal{G}) = T(G) \otimes \mathbb{I}_k
\end{align*}
with $k$ the degree of the unitary group $\mathcal{U}$ maps to.
\end{definition}

\begin{definition}
The \textsf{standard Laplacian} $L(\mathcal{G})$ of a parallel transport network $\mathcal{G}=(V,E,w,\mathcal{T},\mathcal{U},\pi)$ is defined as $L(\mathcal{G})= T(\mathcal{G})-A(\mathcal{G})$. The \textsf{normalized Laplacian} $\mathcal{L}(\mathcal{G})$ of the network $\mathcal{G}$ is defined as $\mathcal{L}(\mathcal{G})=T(\mathcal{G})^{-1/2} L(\mathcal{G}) T(\mathcal{G})^{-1/2}$.
\end{definition}

\begin{lem}\label{lem:laplExplPTN}
The normalized Laplacian of a connected parallel transport network $\mathcal{G}=(V,E,w,\mathcal{T},\linebreak[3]\mathcal{U},\pi)$, $ \mathcal{T}:V\rightarrow \{0, \dots L\}$, with more than one vertex has the explicit form:
\begin{align*}
\mathcal{L}(\mathcal{G}) =& \sum_{t=0}^L \Bigg(
\sum_{v\in V_t} \ket{v}\bra{v} \otimes \mathbb{I}
 - \sum_{v\in V_t} \sum_{u\in V_t} \frac{w(v,u)}{\sqrt{d_v d_u}} \ket{u}\bra{v} \otimes \mathbb{I} \Bigg)\\
 &- \sum_{t=1}^L \sum_{v\in V_t} \sum_{u\in V_{t-1}} \frac{w(v,u)}{\sqrt{d_v d_u}} \left(\ket{v}\bra{u} \otimes U_t + \ket{u}\bra{v} \otimes U_t^\dagger\right)\text{.}\hspace{11.15em}\qed
\end{align*}
\end{lem}

Similar to the case of a graph we call the eigenvectors and eigenvalues of $\mathcal{L}(\mathcal{G})$ just the eigenvectors and eigenvalues of $\mathcal{G}$ and $\lambda(\mathcal{G}):=\lambda_1(\mathcal{G})$ the \textsf{spectral gap} of $\mathcal{G}$. In the rest of this section we will show that the eigensystems of a parallel transport network and its underlying graph are closely related.

\begin{thm}\label{underlying}
Let $M(\mathcal{G})$ be the adjacency, standard Laplacian or normalized Laplacian matrix of a parallel transport network $\mathcal{G}=(G,\mathcal{T},\mathcal{U},\pi)$ and $M(G)$ the respective matrix of the underlying graph $G=(V,E,w)$. Then the following holds:
\begin{align*}
R^\dagger M(\mathcal{G}) R= M \otimes \mathbb{I}
\end{align*}
with the unitary matrix
\begin{align*}
R= \sum_{v \in V} \ket{v}\bra{v} \otimes U_{\mathcal{T}(v)} U_{\mathcal{T}(v)-1} \dots U_1\text{.}
\end{align*}
\end{thm}
\begin{proof}
$R$ is clearly unitary since:
\begin{align*}
R^\dagger R &= \sum_{u\in V} \sum_{v \in V} \ket{u}\overbrace{\braket{u | v}}^{=\delta_{u,v}} \bra{v} \otimes U_1^\dagger \dots U_{\mathcal{T}(u)}^\dagger U_{\mathcal{T}(v)} \dots U_1\\
&=\sum_{v \in V} \ket{v} \bra{v} \otimes \overbrace{U_1^\dagger \dots U_{\mathcal{T}(v)}^\dagger U_{\mathcal{T}(v)} \dots U_1}^{= \mathbb{I}}
= \mathbb{I}\text{.}
\end{align*}
So lets compute the desired expression for the adjacency matrices:
\begin{align*}
R^\dagger A(\mathcal{G}) R
=& R^\dagger \Bigg(\sum_{\substack{(v,u)\in E\\ \mathcal{T}(v)=\mathcal{T}(u)}}w(v,u) \ket{u}\bra{v} \otimes \mathbb{I}
+ \sum_{\substack{(v,u) \in E,\\ \mathcal{T}(u) =\mathcal{T}(v)+1}} w(v,u)\ket{u}\bra{v} \otimes U_{\mathcal{T}(u)}\\
&\hspace{2em}+ \sum_{\substack{(v,u) \in E,\\ \mathcal{T}(u) = \mathcal{T}(v)-1}} w(v,u)\ket{u}\bra{v} \otimes U_{\mathcal{T}(v)}^\dagger\Bigg) R
\end{align*}
\begin{align*}
\phantom{R^\dagger A(\mathcal{G}) R}
=& \sum_{\substack{(v,u)\in E\\ \mathcal{T}(v)=\mathcal{T}(u)}} w(v,u)\ket{u}\bra{v} \otimes \overbrace{\Big(U_1^\dagger \dots U_{\mathcal{T}(u)}^\dagger\Big) \cdot \mathbb{I} \cdot \Big(U_{\mathcal{T}(v)} \dots U_1\Big)}^{=\mathbb{I}}\\
&\hspace{2em}+ \sum_{\substack{(v,u) \in E,\\ \mathcal{T}(u) = \mathcal{T}(v)+1}} w(v,u)\ket{u}\bra{v} \otimes \overbrace{\Big(U_1^\dagger \dots U_{\mathcal{T}(u)}^\dagger\Big) \cdot U_{\mathcal{T}(u)} \cdot \Big(U_{\mathcal{T}(u)-1} \dots U_1\Big)}^{=\mathbb{I}}\\
&\hspace{2em}+ \sum_{\substack{(v,u) \in E,\\ \mathcal{T}(u) = \mathcal{T}(v)-1}} w(v,u)\ket{u}\bra{v} \otimes \overbrace{\Big(U_1^\dagger \dots U_{\mathcal{T}(v)-1}^\dagger\Big) \cdot U_{\mathcal{T}(v)}^\dagger \cdot \Big(U_{\mathcal{T}(v)} \dots U_1\Big)}^{=\mathbb{I}}\\
=& \sum_{(v,u) \in E} w(v,u)\ket{u}\bra{v} \otimes \mathbb{I}\\
=&A(G) \otimes \mathbb{I}\text{.}
\end{align*}
Since $R$ commutes with $T(\mathcal{G})= T(G) \otimes \mathbb{I}$ it follows directly:
\begin{align*}
R^\dagger L(\mathcal{G}) R &= R^\dagger \big(T(\mathcal{G}) - A(\mathcal{G})\big) R\\
&=T(\mathcal{G}) - R^\dagger A(\mathcal{G}) R\\
&=T(G) \otimes \mathbb{I} - A(G) \otimes \mathbb{I} \\
&= L(G) \otimes \mathbb{I}
\end{align*}
and
\begin{align*}
R^\dagger \mathcal{L}(\mathcal{G}) R &= R^\dagger\Big( T(\mathcal{G})^{-1/2}  L(\mathcal{G}) T(\mathcal{G})^{-1/2}\Big) R\\
&= T(\mathcal{G})^{-1/2} \Big(R^\dagger L(\mathcal{G}) R\Big) T(\mathcal{G})^{-1/2}\\
&= \Big(T(G)^{-1/2} \otimes \mathbb{I}\Big) \Big(L(G) \otimes \mathbb{I}\Big) \Big(T(G)^{-1/2} \otimes \mathbb{I}\Big)\\
&=\mathcal{L}(G) \otimes \mathbb{I}\text{.}\qedhere
\end{align*}
\end{proof}

Knowing the above relation it follows directly that the spectrum of a parallel transport network equals the $k$-times degenerated spectrum of the underlying graph:

\begin{cor}\label{cor:spectrumPTN}
Let $M(\mathcal{G})$ be the adjacency, standard Laplacian or normalized Laplacian matrix of a parallel transport network $\mathcal{G}=(G,\mathcal{T}, \mathcal{U},\pi)$, $\pi: U(k) \rightarrow \mathcal{B}(\mathcal{H})$, and $M(G)$ the respective matrix of the underlying graph $G=(V,E,w)$. Then the following two statements are equivalent:
\begin{enumerate}[label={\normalfont(\roman*)}]
\item 
$\ket{g}$ is an eigenvector of $M(G)$ with the eigenvalue $\lambda$.
\item
$R \left(\ket{g} \otimes \ket{e}\right)$ is for every $\ket{e}\in \mathbb{C}^k$ an eigenvector of $M(\mathcal{G})$ with the eigenvalue $\lambda$. \hspace{7em}\qed
\end{enumerate}
\end{cor}

\begin{cor}\label{cor:history}
The null-space of the normalized Laplacian of a connected parallel transport network $\mathcal{G}=(V,E,w,\mathcal{T},\mathcal{U},\pi)$, $\mathcal{T}:V\rightarrow \{0,\dots L\}$, $\pi: U(k) \rightarrow \mathcal{B}(\mathcal{H})$, is spanned by the orthonormal \textsf{history states}
\begin{align*}
\ket{\eta}= \frac{1}{\sqrt{\vol(V)}} \sum_{t=0}^L \sum_{v \in V_t} \sqrt{d_v} \ket{v} \otimes U_t \dots U_1 \ket{x}\text{,}\quad x\in \{0,1\}^k\text{.}
\end{align*}
\end{cor}
\begin{proof}
According to the previous Corollary \ref{cor:spectrumPTN} and Theorem \ref{thm:spectra} the null-space is spanned by the orthogonal vectors
\begin{align*}
\ket{\eta}=\frac{R \left(T^{1/2} \mathbf{1} \otimes \ket{x}\right)}{\left\Vert R \left(T^{1/2} \mathbf{1} \otimes \ket{x}\right)\right\Vert}\text{,}\quad x\in \{0,1\}^k\text{,}
\end{align*}
which turns out to equal the above expression.
\end{proof}

 \chapter{Quantum Circuits and their Complexity}

\section{The quantum circuit model}\label{sec:circuits}

\begin{definition}\label{def:circuits}
\quad

\begin{itemize}
\item A \textsf{quantum circuit} $C=\big(n, (U_1,U_2, \dots U_L)\big)$ is a series of unitary operations $(U_1, U_2, \dots U_L)$, called \textsf{gates}, acting on a Hilbert space of $n$ qubits. A quantum circuit $C$ can be understood as a function mapping every $n$-qubit state $\ket{\phi}$ to the state $U_L \dots U_1 \ket{\phi}$.
\item $L$ is called the \textsf{length} of the circuit.
\item Let $\ket{0}$ and $\ket{1}$ be an orthonormal basis of the 1-qubit Hilbert space. Then the vectors $\ket{x} = \ket{x_1} \otimes \ket{x_2} \otimes \dots \otimes \ket{x_n}$ with $x \in \{0,1\}^n$ form the so-called \textsf{computational basis} of the $n$-qubit Hilbert space. We write $\ket{\mathbf{0}}:= \ket{0^n}$.
\end{itemize}
\end{definition}

A quantum circuit that depends on a computation input is a standard tool to realize quantum algorithms. Normally quantum circuits are used such that the computation input determines the number of circuit qubits and the gates which act then successively on the state $\ket{\mathbf{0}}$ to produce the computation output $U_L \dots U_2 U_1 \ket{\mathbf{0}}$. In contrast to this, here we allow a whole set of \textsf{valid circuit inputs} including the $\ket{\mathbf{0}}$ vector. The reasoning is as follows: One can imagine cases where more than one circuit output can be regarded as correct computation output, for example when the computation is a yes / no problem and only the first qubit of the circuit output represents the result. This implies that more states than only $\ket{\mathbf{0}}$ get mapped to a correct computation output, these are the states we will call the valid circuit inputs. The motivation beyond this idea is that a larger set of valid circuit inputs might give an efficiency improvement of the computation model that we will introduce in Chapter \ref{chap:model}.

Regarding the complexity of a computation model there is an important point to be made with respect to the pre- and postprocessing of the circuit and its output. One can argue that the initial mapping of the computation input to the circuit and the final decoding of the circuit output to the actual computation output allow us to shift the major part of the computation into these mappings and leave the quantum circuit rather trivial. This issue can only be resolved if one agrees on a specific encoding and decoding scheme and compares the complexity of quantum circuits under this premise. A typical setting for example would only allow a classical preprocessing and require for an efficient computation that the computation input is mapped in polynomial time onto the quantum circuit. In most concepts the computation output is also required to be classical and obtained (perhaps only with a certain probability) from the circuit output after measuring in the computational basis and / or tracing out some subsystem.

Indeed regarding measurements, one might find a deviation between ours and the usual definition of a quantum circuit. Normally quantum circuits are defined to include measurements at any point of the circuit. Measurements introduce true randomness which is an essential feature of quantum mechanics and should not be excluded from the tools of a quantum algorithms. The reason that we omitted measurements in Definition \ref{def:circuits} of quantum circuits is that adiabatic quantum computation, a computation model that we will introduce in the next chapter, can only simulate quantum circuits without measurements. But fortunately this restriction is not a problem since the adiabatic quantum computation ends with a projective measurement which gives a state arbitrarily close to the output state of the circuit with a certain probability. If a quantum circuit now consisted orginally of unitary gates and a final projective measurement, we can just simulate the circuit without the final measurement and combine  it instead into the final projective measurement of the adiabatic quantum computation. If a quantum circuit contains a measurement anywhere in between it can moreover be replaced by a projective measurement at the end of the circuit on an auxillary qubit. This proof can be found in \cite[sec 2.2.8]{nielson} and is how we justify the restriction to our definition of quantum circuits.

\section{Complexity of quantum circuits}

In classical computer science, memory space and running time constitute the central complexity quantities which characterise algorithms. Looking at a quantum algorithm defined by a family of quantum circuits there seems to be an intuitive correspondance to the number of input qubits and the circuit length. But if we allow any quantum circuit according to Definition \ref{def:circuits} the circuit length is clearly problematic since any series of unitaries gates could just be substituted by its product as a single gate and hence any quantum circuit could be reduced to length $1$.

Therefore there is a need to define a set of elementary gates that are only allowed to occur in quantum circuits. This is also motivated for practical reasons, since under this premise only some simple unitary transformation have to be implemented whose combinations can realize arbitrary unitary transformations. Because of practical issues we also wish that our elementary gates are $k$-local with $k$ some small integer since only those seem to be elementary transformations in nature.

\begin{definition}
A unitary operator $U$ on a $n$-qubit space is called \textsf{$k$-local} iff it has the form
\begin{align*}
U= \tilde{U} \otimes \mathbb{I}_{n-k}
\end{align*}
with $\tilde{U}$ a unitary acting on $k$ arbitrary qubits and the identity $\mathbb{I}_{n-k}$ acting on the rest of the qubits (note that the tensor product can be taken in any order). We say that $U$ acts on the rest $n-k$ qubits \textsf{trivially} and call it a \textsf{trivial extension} of $\tilde{U}$.
\end{definition}

The main requirement of our elementary gates is of course that they can be combined to approximate any unitary arbitrarily closely. This leads us to the definition of a universal gate set:

\begin{definition}
A \textsf{universal gate set} is a set of gates such that any unitary operation on qubits can be approximated to arbitrary accuracy by a product of trivial extensions of these gates.
\end{definition}

It is not obvious that a universal gate set exists and not at all if its elements are $k$-local or the set even finite. Fortunately all these properties do hold and the proof of the following theorem can be found in \cite[sec. 4.5]{nielson}:

\begin{thm}
The $1$-qubit Hadamard gate

\begin{tikzpicture}[scale=0.9]
\draw[line width=1pt] (-0.6,0) to (0,0);
\draw[line width=1pt] (0.8,0) to (1.4,0);
\draw[line width=0.7pt] (0,-0.4) -- (0,0.4) -- (0.8,0.4) -- (0.8,-0.4) -- (0,-0.4);
\node at (0.4,0) {$H$};
\node at (-8.57,0) {$\displaystyle H = \ket{+}\bra{0} + \ket{-}\bra{1}\,\,\text{ with }\,\, \ket{+}:= \frac{\ket{0} + \ket{1}}{\sqrt{2}} \,\,\text{ and }\,\, \ket{-}:= \frac{\ket{0} - \ket{1}}{\sqrt{2}}$,};
\end{tikzpicture}\\
the $1$-qubit $\pi/8$ gate

\begin{tikzpicture}[scale=0.9]
\draw[line width=1pt] (-0.6,0) to (0,0);
\draw[line width=1pt] (0.8,0) to (1.4,0);
\draw[line width=0.7pt] (0,-0.4) -- (0,0.4) -- (0.8,0.4) -- (0.8,-0.4) -- (0,-0.4);
\node at (0.4,0) {$T$};
\node at (-12.3,0) {$\displaystyle T = \ket{0}\bra{0} + e^{i \pi/4} \ket{1}\bra{1}$};
\end{tikzpicture}\\
and the $2$-qubit CNOT gate

\begin{tikzpicture}[scale=0.9]
\draw[line width=1pt] (0,0) to (2,0);
\draw[line width=1pt] (0,-1) to (2,-1);
\draw[line width=0.7pt] (1,0) to (1,-1.25);
\draw[line width=0.7pt] (1,-1) circle (0.25);
\fill (1,0) circle (0.15);
\node at (-9.62,-0.5) {$\displaystyle CNOT = \ket{0}\bra{0} \otimes \mathbb{I} + \ket{1}\bra{1} \otimes \big(\ket{0}\bra{1} + \ket{1}\bra{0} \big)$};
\end{tikzpicture}\\
form a universal gate set.
\end{thm}

The symbols on the right represent the gates in a graphical drawing of the quantum circuit. In such a drawing each qubit is represented by a line and the above symbols for the gates touch the lines of the qubits they act on. Another commonly used gate which can be built from the above gates is the $3$-qubit Toffoli gate

\begin{tikzpicture}[scale=0.9]
\draw[line width=1pt] (0,1) to (2,1);
\draw[line width=1pt] (0,0) to (2,0);
\draw[line width=1pt] (0,-1) to (2,-1);
\draw[line width=0.7pt] (1,1) to (1,-1.25);
\draw[line width=0.7pt] (1,-1) circle (0.25);
\fill (1,0) circle (0.15);
\fill (1,1) circle (0.15);
\node at (-7.2,0) {$\displaystyle
TOF = \big(\ket{00}\bra{00} + \ket{01}\bra{01} + \ket{10}\bra{10}\big) \otimes \mathbb{I} + \ket{11}\bra{11} \otimes  \big(\ket{0}\bra{1} + \ket{1}\bra{0} \big)
$.};
\end{tikzpicture}\\

The Hadamard, $\pi/8$ and CNOT gate form the most widely used universal gate set. We do not have to assume neccessarily from now on that all quantum circuits consist of these specific gates, but we do assume that their gates are elements from \textit{some} finite universal 2-local gate set. Consequently from now on the length of a quantum circuit is an important complexity quantity.

In complexity theory one is rarely interested in the actual value of a quantity but rather in its growing behaviour in terms of the input length $n$. To compare two functions regarding this behaviour we will use the so-called $\mathcal{O}$-notation:

\begin{definition}
For functions $f, g: \mathbb{N} \rightarrow \mathbb{N}$ we define
\begin{align*}
f(n) \in \mathcal{O}(g(n)) \quad &\Longleftrightarrow \quad \lim_{n\rightarrow \infty} \frac{f(n)}{g(n)} < \infty\\
f(n) \in \Omega(g(n)) \quad &\Longleftrightarrow \quad \lim_{n\rightarrow \infty} \frac{g(n)}{f(n)} < \infty\\
f(n) \in \Theta(g(n)) \quad &\Longleftrightarrow \quad \lim_{n\rightarrow \infty} \frac{f(n)}{g(n)} = \text{const.}
\quad \Longleftrightarrow\quad f(n) \in \mathcal{O}(g(n)) \text{ and } f(n) \in \Omega(g(n))\text{.}
\end{align*}
\end{definition}

An algorithm with input length $n$ is called \textsf{efficient}, if the complexity quantities of interest (e.g. running time, memory space) are elements of $\mathcal{O}(\poly(n))$, which means that they are upper bounded by a polynomial in the input length.

\section{Time-dependant circuits and identity extensions}

In the next chapters we will set up a different computation model that will simulate a quantum circuit. Because it might offer some features we will implement an identity extension of the original circuit:
\begin{definition}\label{def:identityExtension}
For a quantum circuit $C=\big(n, (\tilde{U}_1, \dots \tilde{U}_L)\big)$, the circuit $C'=\big(n, (U_1, \dots U_{L'})\big)$, $L' = L_0 + L + L_f$, with
\begin{align*}
U_t =
\begin{cases}
\mathbb{I} \quad\quad &\text{for } 0 \le t \le L_0 \text{ or } L'-L_f \le t \le L'\\
\tilde{U}_{t-L_0} &\text{otherwise}
\end{cases}
\end{align*}
is called an \textsf{identity extension} of the circuit $C$.
\end{definition}
As the quantum ciruit $C'$ is none other than the original circuit $C$ with $L_0$ initial and $L_f$ final identity gates appended, it obviously carries out the same computation. 

For our computation model we will also have to introduce a normalized time parameter $s$ into the circuit gates that lets them transform smoothly from the identity to the actual operator. The next proposition will allow us to find an easy solution for this objective:
\begin{prop}\label{prop:unitaryHermitean}
Any unitary operator $U$ can be written as
\begin{align*}
U = e^{i h}
\end{align*}
with $h$ a hermitean operator and $\Vert h \Vert \le 2\pi$.
\end{prop}
\begin{proof}
Since every unitary operator $U$ is normal, it is diagonaliziable by a unitary operator $A$:
\begin{align*}
U = A D A^\dagger
\end{align*}
with $D$ a diagonal matrix. Its diagonal elements $d_i$ are all roots of unity since $U$ is unitary. Consequently $D$ can be written as $D = e^{iD'}$ with $\Vert D'_i \Vert \le 2 \pi$ and hence
\begin{align*}
U &= A e^{i D'} A^\dagger\\
&=e^{i A D' A^\dagger}\\
&= e^{i h}
\end{align*}
with $h := A D' A^\dagger$ obviously hermitean and $\Vert h \Vert = \Vert D' \Vert \le 2 \pi$.
\end{proof}

Given a quantum circuit $C=\big(n, (U_1, \dots U_L)\big)$ with $U_t = \exp(ih_t)$, the operators $U_t(s) := \exp(ish_t)$ are obviously also unitary. This allows us to introduce our desired time-dependency into quanutm circuits:
\begin{definition}\label{def:correspondingCircuit}
For a quantum circuit $C=\big(n, (U_1, \dots U_L)\big)$ with $U_t = \exp(ih_t)$, $\Vert h_t \Vert \le 2\pi$ hermitean, the \textsf{corresponding time-depending circuit} is defined as
\begin{align*}
C(s) := \big(n, (U_1(s), \dots U_L(s)\big)
\end{align*}
with $U_t(s) = \exp(ish_t)$.
\end{definition}

\chapter{Adiabatic Quantum Computation}

\section{Requirements for an efficient adiabatic quantum computation}\label{sec:AQCrequirements}

The concept of quantum circuits introduced in the previous chapter is the most widely used tool to formulate quantum algorithms. In this chapter we further introduce another kind of quantum computing, the so-called adiabatic quantum computation (AQC). Its centerpiece is the quantum adiabatic theorem, whose proof we postpone to Sections \ref{sec:AQTstates} and \ref{sec:AQTproj}. Notice that we use the convention $\hbar=c=1$:

\begin{thm*}[Quantum Adiabatic Theorem for states]
Let $H(s)$ be a finite-dimensional Hamiltonian such that for all $0 \le s \le 1$ it is twice differentiable and let $\ket{\eta(s)}$ be a differentiable, nondegenerate ground-state separated by an energy gap $\gamma(s)$ from the first excited state. Choose $\epsilon > 0$ arbitrarily and define
\begin{align*}
T :=
 \frac{1}{\epsilon} \left( \frac{1}{\gamma(0)^2} \left\Vert\frac{\mathrm{d} H}{\mathrm{d} s} \right\Vert_{s=0}
+ \frac{1}{\gamma(1)^2} \left\Vert\frac{\mathrm{d} H}{\mathrm{d} s} \right\Vert_{s=1}
+ \int_0^1 \mathrm{d}s \left( \frac{5}{\gamma(s)^3} \left\Vert\frac{\mathrm{d} H}{\mathrm{d} s} \right\Vert^2
+ \frac{1}{\gamma(s)^2} \left\Vert\frac{\mathrm{d}^2 H}{\mathrm{d} s^2} \right\Vert \right)
\right)\text{.}
\end{align*}
Consider a system evolving according to the Hamiltonian $H(t/T)$. Denote by $U\left(\frac{t_1}{T}, \frac{t_0}{T}\right)$ its time evolution operator. Then it holds:
\begin{align*}
\left\Vert U(1,0)\ket{\eta(0)} - e^{i\beta(s)}\ket{\eta(1)}\right\Vert \le \epsilon
\quad\quad\text{ with }\quad\quad
\beta(s):= \int_0^s i \braket{\eta(u) | \eta'(u)} du\text{.}
\end{align*}
\end{thm*}

The adiabatic theorem states that having a system initially prepared in its ground-state $\ket{\eta(0)}$ and letting it evolve with a long enough time $T$, the system will stay arbitrarily close to its ground-state $\ket{\eta(1)}$ up to some phase $\beta(1)$. This version of the quantum adiabatic theorem is based on Jeffrey Goldstone \cite[Appendix F]{jordan} and requires that the ground-state is nondegenerate. In Section \ref{sec:AQTproj} we will derive a similar result for degenerated ground-spaces:

\begin{thm*}[Quantum Adiabatic Theorem for projection operators]
Let $H(s)$ be a finite-dimensional Hamiltonian such that for all $0 \le s \le 1$ it is twice differentiable and let $P(s)$ be the projection operator onto the ground-space which is separated by an energy gap $\gamma(s)$ from the first excited state. Choose $\epsilon > 0$ arbitrarily and define
\begin{align*}
T :=
 \frac{2}{\epsilon} \left( \frac{1}{\gamma(0)^2} \left\Vert\frac{\mathrm{d} H}{\mathrm{d} s} \right\Vert_{s=0}
+ \frac{1}{\gamma(1)^2} \left\Vert\frac{\mathrm{d} H}{\mathrm{d} s} \right\Vert_{s=1}
+ \int_0^1 \mathrm{d}s \left( \frac{6}{\gamma(s)^3} \left\Vert\frac{\mathrm{d} H}{\mathrm{d} s} \right\Vert^2
+ \frac{1}{\gamma(s)^2} \left\Vert\frac{\mathrm{d}^2 H}{\mathrm{d} s^2} \right\Vert \right)
\right)\text{.}
\end{align*}
Consider a system evolving according to the Hamiltonian $H(t/T)$. Denote by $U\left(\frac{t_1}{T}, \frac{t_0}{T}\right)$ its time evolution operator. Then it holds:
\begin{align*}
\left\Vert U(1,0) P(0) U(0,1) - P(1) \right\Vert \le \epsilon\text{.}
\end{align*}
\end{thm*}

The arbitrary closeness of the projection operators implies of course also that an initially prepared ground-state $\ket{\eta(0)}$ evolves to a state that is arbitrarily close to a final ground-state. The decomposition ansatz $U(1,0)\ket{\eta(0)} = \sqrt{1-\delta^2} \ket{\eta(1)} + \delta \ket{\eta^\perp(1)}$ into a final ground-state $\ket{\eta(1)}$ and a vector $\ket{\eta^\perp(1)}$ orthogonal to the ground-space leads directly to
\begin{align*}
\left\Vert U(1,0) \ket{\eta(0)} - \ket{\eta(1)} \right\Vert
= \sqrt{2\left(1-\sqrt{1-\epsilon^2}\right)}
\le \sqrt{2} \epsilon\text{.}
\end{align*}

The important difference compared to the original adiabatic theorem for states is that it is normally not known which degenerate ground-state $\ket{\eta(1)}$ the evolved state is close to.

The idea of adiabatic quantum computation is to find for any computational task a time-dependant Hamiltonian $H(s)$ that has an easy to prepare initial ground-state and whose final ground-state(s) encode the desired output of the computation. Then we simply prepare the system in the initial ground-state and let it evolve according to the Hamiltonian $H\left(\frac{t}{T}\right)$. If we choose the evolution time $T$ long enough, we finally obtatin a state that is arbitrarily close to the encoded output of the computation. The ``encoded output'' does not have to equal exactly the output of the computation, but we expect to obtain from it the desired computation output with a non-vanishing probabilty after a final projective measurement and perhaps the discarding of an auxillary system.

A nice property of adiabatic quantum computation is that it is close to physical implementation. But of course this is only true if the Hamiltonian is physically realizable. This need along with the desirability of an efficient adiabatic quantum computation leads us to a list of requirements for the desired Hamiltonian. Recall that we consider an algorithm to be efficient if the crucial complexity quantities are bounded from above by a polynomial in the input size. The Hamiltonian is of course a function of the computation input and so are its properties. Therefore we use the $\mathcal{O}$-notation we introduced in the last chapter to express the bounds that guarantee an efficient computation and physical realization:
\tcbset{colback=black!5!white,colframe=black!50!white}
\vspace{2mm}
\begin{tcolorbox}[toptitle=2mm,fonttitle=\sffamily\bfseries\large,title=Requirements for an efficient adiabatic quantum computation]
to realize a computation with input of length $n$.
\tcblower
\begin{enumerate}
\item The Hamiltonian $H(s)$ is normalized and defined on a space of $\mathcal{O}(\poly(n))$ qubits.
\item The Hamiltonian is a sum of $\mathcal{O}(\poly(n))$ many, local interaction terms.
\item For every interaction term $H_1(s)$ it holds $\Vert H_1(s) \Vert \in \Omega\left(\frac{1}{\poly(n)}\right)$ for all $0\le s \le 1$.
\item There is a classical algorithm that computes the interaction terms in time $\mathcal{O}(\poly(n))$.
\item The Hamiltonian $H(0)$ has a ground-state independant of the computation that can be prepared from the all zero state $\ket{\mathbf{0}}$ by a quantum circuit of length $\mathcal{O}\left(\poly(n)\right)$.
\item After a projective measurement on any ground-state of $H(1)$ and a possible discarding of a subsystem the computation output is obtained with probability $p \in \Omega\left(\frac{1}{\poly(n)}\right)$.
\item For the gap $\gamma(s)$ of the Hamiltonian it holds that $\gamma(s) \in \Omega\left(\frac{1}{\poly(n)}\right)$ for all $0 \le s \le 1$.
\item The norm of the first and second derivative of $H(s)$ is for all $0 \le s \le 1$ polynomially bounded: $\left\Vert  \frac{dH(s)}{ds} \right\Vert \in \mathcal{O}(\poly(n))$ and $\left\Vert  \frac{d^2H(s)}{ds^2} \right\Vert \in \mathcal{O}(\poly(n))$.
\end{enumerate}
\end{tcolorbox}
\vspace{2mm}

In Section \ref{sec:circuitSimReq} we will rewrite this list under the assumption of a very specific Hamiltonian structure that includes normalized Laplacians of parallel transport networks.

Requirements 1--4 guarantee an efficient setup of a physical Hamiltonian, 3 is also motivated for reasons of accuracy. The classical algorithm mentioned in 4 computes of course a certain classical representation of the interaction terms. The Hamiltonians we use for adiabatic quantum computation should hence have a certain structure such that computing a classical encoding in polynomial time is possible. The idea behind requirement 5 is that the initial state is ``easy'' to prepare. Requirement 6 guarantees an acceptable rate of computing the correct output. Finally, if requirements 7 and 8 are fulfilled and an error $\epsilon$ fixed, the quantum adiabatic theorem implies that the neccessary evolution time $T$ is polynomial in the input size $n$ and hence efficient.

Actually we are not only concerned that our computation \textit{is} efficient but especially \textit{how} efficient it is. In particular we are interested in an optimization of the running time $T$ of the computation. While we simply demand that requirements 1--6 are fulfilled, we want particular low polynomials for the inverse gap and the Hamiltonian derivatives in requirements 7--8. In the next section we will discuss that a comparision of these polynomials makes sense only for normalized Hamiltonians, which is the reason why we included this in requirement 1.

\section{The requirement of a normalized Hamiltonian}\label{sec:norm}

In this thesis, when we refer to ``normalized'' Hamiltonians such as that of requirement 1, we assume that the norm of the Hamiltonian is upper bounded by a constant. This constant does not neccessarily have to equal 1 since the actual value does not affect the complexity classes written in $\mathcal{O}$-notation.

Why do we demand that the Hamiltonian in requirement 1 is normalized? The reason is that if we multiply the Hamiltonian by a factor $z$ (and hence consequently also its derivatives and the energy gap), the neccessary evolution time $T$ will be scaled by $\frac{1}{z}$ according to the adiabatic theorem. First of all this is consistent with the fact that the unit of the time $T$ has to equal the inverse energy under the convention $\hbar=c=1$. Physically this scaling also makes sense, since two initially identical states $\ket{\Psi_1(0)}=\ket{\Psi_2(0)}$, the first evolving under the Hamiltonian $H(t/T)$, the second under the Hamiltonian $z H(zt/T)$, fulfill $\ket{\Psi_2(t)} = \ket{\Psi_1(zt)}$ according to Schr\"odinger's equation. So indeed, in scaling the Hamiltonian arbitrarily large, we can make the evolution time $T$ as small as we want by keeping the final state the same. 

Hence not paying attention to the norm of the Hamiltonian would make every complexity discussion about the evolution time meaningless. Considering the norm of the Hamiltonian also as an complexity issue makes sense since, in practice, an arbitrary high energy scaling requires a significant effort. So either we look at $\Vert H \Vert \cdot T$ as relevant complexity quantity or we normalize the Hamiltonians to be able to compare the evolution time $T$. As requirement 1 reflects we pursue the second approach in this thesis.

Another convincing reason for considering $\Vert H \Vert T$ as the relevant complexity quantity or just $T$ requiring a normalized Hamiltonian occurs if we try to formulate the complexity of the evolution of the system. A standard way to formulate the evolution complexity is to count the numbers of gates and oracle queries for the matrix entries a quantum circuit needs to simulate this evolution. We don't want to go into detail about this method here, but we want to mention that Berry, Cleve and Somma \cite{berry} showed that this number grows like $\Vert H \Vert T \log^3 \left(\Vert H \Vert T +\Vert \frac{dH}{dt} \Vert T\right)$.

\section{The quantum adiabatic theorem for states}\label{sec:AQTstates}

In this section we will give the proof for the quantum adiabatic theorem for states as it has already been stated at the beginning of this chapter. But first we show that for every differentiable unit vector $\ket{\eta(s)}$ of some Hilbert space there exists another unit vector $\ket{\hat{\eta}(s)}$ that equals $\ket{\eta(s)}$ up to a phase and is initially even identical while the scalar product with its derivative vanishes:

\begin{lem}[adapted from \cite{ambainis}]\label{lem:phase}
For every differentiable unit vector $\ket{\eta(s)}$ of some Hilbert space the vector
\begin{align*}
\ket{\hat{\eta}(s)}:= e^{i\beta(s)} \ket{\eta(s)} \quad\quad\text{ with }\quad\quad\beta(s):= \int_0^s i \braket{\eta(u) | \eta'(u)} du
\end{align*}
fulfills
\begin{align*}
\ket{\hat{\eta}(0)}=\ket{\eta(0)}
\quad\quad and \quad\quad
\braket{\hat{\eta}(s) | \hat{\eta}'(s)}=0\quad\text{ for all }0\le s\le 1\text{.}
\end{align*}
\end{lem}
\begin{proof}
$\ket{\hat{\eta}(0)}=\ket{\eta(0)}$ holds by definition. The derivative of $\ket{\hat{\eta}(s)}=e^{i\beta(s)}\ket{\eta(s)}$ can be written as
\begin{align*}
\ket{\hat{\eta}'(s)}= e^{i\beta(s)} \ket{\eta'(s)} + e^{i \beta(s)} i \beta'(s) \ket{\eta(s)}\text{.}
\end{align*}
The scalar product with $\ket{\hat{\eta}(s)}$ gives with the above defined $\beta(s)$:
\begin{align*}
\braket{\hat{\eta}(s) | \hat{\eta}'(s)}&= \braket{\eta(s) | \eta'(s)} + i\beta'(s) \underbrace{\braket{\eta(s) | \eta(s)}}_{=1}=0\text{.} \qedhere
\end{align*}
\end{proof}

As a ground-state multiplied by a phase is still a ground-state, we can assume w.l.og. in the next two technical lemmata that the scalar product of the ground-state with its derivative vanishes.

\begin{lem}\label{lem:derivPhi}
Let $H(s)$ be a finite dimensional, differentiable Hamiltonian with nondegenerate ground-state corresponding to the eigenvalue $0$. Denote by $\ket{\phi_j(s)}$, $0 \le j\le n-1$, its eigenstates with eigenvalues $E_j(s)$ and by $\ket{\phi_0(s)}$ the ground-state such that the scalar product with its derivative vanishes. Then it holds
\begin{align*}
\ket{\phi'_0(s)} &= - G(s) H'(s) \ket{\phi_0(s)}
\end{align*}
with
\begin{align*}
G(s):=\sum\limits_{j=1}^{n-1} \frac{\ket{\phi_j(s)}\bra{\phi_j(s)}}{E_j(s)}\text{.}
\end{align*}
\end{lem}
\begin{proof}
\begin{align*}
\frac{d}{ds}(H(s) \ket{\phi_0(s)})&=0\\
H(s) \ket{\phi'_0(s)} &= - H'(s) \ket{\phi_0(s)} \\
\mathbb{I} \ket{\phi'_0(s)} - \ket{\phi_0(s)} \underbrace{\braket{\phi_0(s) | \phi'_0(s)}}_{=0} &= - G(s) H'(s) \ket{\phi_0(s)}\qedhere 
\end{align*}
\end{proof}

\begin{lem}\label{lem:integral}
The quantities defined in the previous lemma fulfill
\begin{align*}
\left\Vert \frac{d}{ds} \left( G(s)^2 \frac{dH(s)}{ds} \ket{\phi_0(s)} \right) \right\Vert
\le 5 \left\Vert G(s)\right\Vert^3 \left\Vert \frac{dH(s)}{ds}\right\Vert^2 + \left\Vert G(s) \right\Vert^2 \left\Vert \frac{d^2H(s)}{ds^2}\right\Vert\text{.}
\end{align*}
\end{lem}
\begin{proof}
The operator
\begin{align*}
\tilde{H}(s) := H(s) + \ket{\phi_0(s)}\bra{\phi_0(s)}
\end{align*}
is invertible and fulfills like every invertible differentiable operator
\begin{align}
\frac{d\tilde{H}(s)^{-1}}{ds} = \left( \frac{d\tilde{H}(s)^{-1}}{ds} \tilde{H}(s) \right) \tilde{H}(s)^{-1} 
= - \tilde{H}(s)^{-1} \frac{d\tilde{H}(s)}{ds} \tilde{H}(s)^{-1}\text{.}\label{eq:invertH}
\end{align}
For its derivative we find
\begin{align*}
\frac{d\tilde{H}(s)}{ds} = \frac{dH(s)}{ds} + \ket{\phi'_0(s)} \bra{\phi_0(s)} + \ket{\phi_0(s)}\bra{\phi'_0(s)}
\end{align*}
and hence
\begin{align}
G(s)\frac{d\tilde{H}(s)}{ds}G(s) = G(s)\frac{dH(s)}{ds}G(s)\text{.}\label{eq:derivHtilde}
\end{align}
Define
\begin{align*}
Q(s):= \mathbb{I} - \ket{\phi_0(s)}\bra{\phi_0(s)}\text{.}
\end{align*}
Then
\begin{align*}
G(s) = Q(s) \tilde{H}(s)^{-1} = \tilde{H}(s)^{-1} Q(s) = Q(s) \tilde{H}(s)^{-1} Q(s)
\end{align*}
and we can derive the following expression for the derivative of $G(s)$:
\begin{align*}
\frac{dG(s)}{ds} &= \frac{d}{ds} \left( Q(s) \tilde{H}(s)^{-1} Q(s) \right)\\
&= \frac{dQ(s)}{ds} \tilde{H}(s)^{-1} Q(s) + Q(s) \frac{d\tilde{H}(s)^{-1}}{ds} Q(s) + Q(s) \tilde{H}(s)^{-1} \frac{Q(s)}{ds}\\
&\stackrel{\eqref{eq:invertH}}{=}\frac{dQ(s)}{ds} G(s) - Q(s) \tilde{H}(s)^{-1} \frac{d\tilde{H}(s)}{ds} \tilde{H}(s)^{-1} Q(s) + G(s) \frac{dQ(s)}{ds}\\
&\stackrel{\eqref{eq:derivHtilde}}{=}\frac{dQ(s)}{ds} G(s) - G(s) \frac{dH(s)}{ds}G(s) + G(s) \frac{dQ(s)}{ds}\text{.}
\end{align*}
With the help of Lemma \ref{lem:derivPhi} we can substitute
\begin{align*}
\frac{dQ(s)}{ds} &= - \ket{\phi'_0(s)}\bra{\phi_0(s)} - \ket{\phi_0(s)}\bra{\phi'_0(s)}\\
&= G(s)\frac{dH(s)}{ds} \ket{\phi_0(s)} \bra{\phi_0(s)} + \ket{\phi_0(s)} \bra{\phi_0(s)} \frac{dH(s)}{ds} G(s)
\end{align*}
to get
\begin{align*}
\frac{dG(s)}{ds} &=\ket{\phi_0(s)} \bra{\phi_0(s)} \frac{dH(s)}{ds} G(s)^2
  - G(s) \frac{dH(s)}{ds}G(s)
+ G(s)^2 \frac{dH(s)}{ds} \ket{\phi_0(s)} \bra{\phi_0(s)}\text{.}
\end{align*}
With the above expression for $\frac{dG(s)}{ds}$ and the help of the triangle inequality we can derive our desired result:
\begin{alignat*}{3}
\bigg\Vert \frac{d}{ds} \bigg( &G(s)&&^2 \frac{dH(s)}{ds} \ket{\phi_0(s)} \bigg) \bigg\Vert\\
&\le \bigg\Vert&& \frac{dG(s)}{ds} G(s) \frac{dH(s)}{ds} \ket{\phi_0(s)} + G(s)\frac{dG(s)}{ds} \frac{dH(s)}{ds} \ket{\phi_0(s)}
 + G(s)^2 \frac{dH(s)}{ds} \ket{\phi'_0(s)} \bigg\Vert\\
&&&+ \left\Vert G(s) \right\Vert^2 \left\Vert \frac{d^2H(s)}{ds^2} \right\Vert\\
&= \bigg\Vert&& \left(\ket{\phi_0(s)} \bra{\phi_0(s)} \frac{dH(s)}{ds} G(s)^3 \frac{dH(s)}{ds} \ket{\phi_0(s)} - G(s) \frac{dH(s)}{ds}G(s)^2\frac{dH(s)}{ds} \ket{\phi_0(s)} \right)\\
&&&+ \left(- G(s)^2 \frac{dH(s)}{ds}G(s) \frac{dH(s)}{ds} \ket{\phi_0(s)} +  G(s)^3 \frac{dH(s)}{ds} \ket{\phi_0(s)} \bra{\phi_0(s)}\frac{dH(s)}{ds} \ket{\phi_0(s)}\right)\\
&&& +\left(- G(s)^2 \frac{dH(s)}{ds} G(s) \frac{dH(s)}{ds}\ket{\phi_0(s)} \right)\bigg\Vert 
+ \left\Vert G(s) \right\Vert^2 \left\Vert \frac{d^2H(s)}{ds^2} \right\Vert\\
&\le 5&& \left\Vert G(s) \right\Vert^3 \left\Vert\frac{dH(s)}{ds}\right\Vert^2 
+ \left\Vert G(s) \right\Vert^2 \left\Vert \frac{d^2H(s)}{ds^2} \right\Vert\text{.}
\hspace{16.2em}\qedhere
\end{alignat*}
\end{proof}

Now we can finally prove the complete quantum adiabatic theorem for states. As mentioned earlier the theorem is adapted from \cite[Appendix F] {jordan} based on a proof by Jeffrey Goldstone:

\begin{thm}[Quantum Adiabatic Theorem for states]\label{thm:AQTstates}
Let $H(s)$ be a finite-dimensional Hamiltonian such that for all $0 \le s \le 1$ it is twice differentiable and let $\ket{\eta(s)}$ be a differentiable, nondegenerate ground-state separated by an energy gap $\gamma(s)$ from the first excited state. Choose $\epsilon > 0$ arbitrarily and define
\begin{align*}
T :=
 \frac{1}{\epsilon} \left( \frac{1}{\gamma(0)^2} \left\Vert\frac{\mathrm{d} H}{\mathrm{d} s} \right\Vert_{s=0}
+ \frac{1}{\gamma(1)^2} \left\Vert\frac{\mathrm{d} H}{\mathrm{d} s} \right\Vert_{s=1}
+ \int_0^1 \mathrm{d}s \left( \frac{5}{\gamma(s)^3} \left\Vert\frac{\mathrm{d} H}{\mathrm{d} s} \right\Vert^2
+ \frac{1}{\gamma(s)^2} \left\Vert\frac{\mathrm{d}^2 H}{\mathrm{d} s^2} \right\Vert \right)
\right)\text{.}
\end{align*}
Consider a system evolving according to the Hamiltonian $H(t/T)$. Denote by $U\left(\frac{t_1}{T}, \frac{t_0}{T}\right)$ its time evolution operator. Then it holds:
\begin{align*}
\left\Vert U(1,0)\ket{\eta(0)} - e^{i\beta(s)}\ket{\eta(1)}\right\Vert \le \epsilon
\quad\quad\text{ with }\quad\quad
\beta(s):= \int_0^s i \braket{\eta(u) | \eta'(u)} du\text{.}
\end{align*}
\end{thm}
\begin{proof}
Let $\ket{\phi_j(s)}$, $0 \le j\le n-1$, be the eigenstates of the Hamiltonian with eigenvalues $E_j(s)$ and $\ket{\phi_0(s)}=e^{i\beta(s)}\ket{\eta(s)}$. Lemma \ref{lem:phase} tells us that the scalar product of the ground-state $\ket{\phi_0(s)}$ with its derivative vanishes.

First of all we want to argue why it is sufficient to only prove the case $E_0(s)=0$ for all $0\le s\le 1$. Consider the time evolution operator $U_0(s,s')$ of the Hamiltonian $H_0(s):=H(s)-E_0(s)$ with time $t:=sT$. Of course a state evolving according to $H(s)$ just differs in a phase in comparision to an evolution with $H_0(s)$:
\begin{align*}
U(s,0)\ket{\eta(0)}&=e^{i\alpha(s)} U_0(s,0)\ket{\eta(0)} \quad\quad\text{ with }\alpha(s)= T\int_0^s E_0(u) du \text{.}
\end{align*}
Assume the adiabatic theorem holds for $H_0(s)$ and define the new ground-state $\ket{\eta_0(s)} =  e^{i\alpha(s)}\ket{\eta(s)}$. Obviously $\ket{\eta_0(0)}=\ket{\eta(0)}$ and $\beta_0(s):= \int_0^s i \braket{\eta_0(s) | \eta'_0(s)} du = \beta(s)$. It follows 
\begin{align*}
\left\Vert U(1,0)\ket{\eta(0)} - e^{i\beta(s)}\ket{\eta(1)}\right\Vert
&= \left\Vert e^{i\alpha(s)}U_0(1,0)\ket{\eta_0(0)} - e^{i\beta_0(s)} e^{i\alpha(s)}\ket{\eta_0(1)}\right\Vert
\le \epsilon\text{,}
\end{align*}
so the adiabatic theorem also holds for $H(s)$. Hence, it is sufficient to prove the adiabatic theorem just for the case of zero ground-energy. So from now on assume $E_0(s)=0$ for all $0\le s \le 1$.

As before define
\begin{align*}
G(s):=&\sum\limits_{j=1}^{n-1} \frac{\ket{\phi_j(s)}\bra{\phi_j(s)}}{E_j(s)}\\
Q(s):=& \mathbb{I} - \ket{\phi_0(s)}\bra{\phi_0(s)}
\end{align*}
and keep in mind that $G(s)$ almost ``inverts'' $H(s)$:
\begin{align*}
G(s)H(s) = H(s)G(s)=\mathbb{I} - \ket{\phi_0(s)}\bra{\phi_0(s)}=Q(s)\text{.}
\end{align*}

For any state $\ket{\psi(s)}$ Schr\"odinger's equation gives
\begin{align*}
i \frac{d}{dt}\ket{\psi(t/T)} &= H(t/T)\ket{\psi(t/T}\quad\quad\quad | \, s=\frac{t}{T}, \, \frac{ds}{dt} = \frac{1}{T}\\
\frac{i}{T} \frac{d}{ds}\ket{\psi(s)} &= H(s)\ket{\psi(s)}\text{.}
\end{align*}
Consequently the unitary evolution operator $U(s,s')$ is the solution of
\begin{align*}
\frac{i}{T} \frac{d}{ds}U(s,s') &= H(s)U(s,s'), \quad U^\dagger(s,s')=U(s',s), \quad U(s,s) = \mathbb{I}\text{.}
\end{align*}
Taking the adjoint of the first equation leads to
\begin{align}
-\frac{i}{T} \frac{d}{ds}U(s',s) &= U(s',s)H(s)\text{.}\label{eq:adjointU}
\end{align}
Later we will need an expression for $U(s',s)$, so we multiply $G(s)$ from the right and get:
\begin{align}
-\frac{i}{T} \frac{d}{ds}U(s',s) G(s)&= U(s',s) \left(\mathbb{I} - \ket{\phi_0(s)}\bra{\phi_0(s)}\right)\nonumber\\
U(s',s) &= -\frac{i}{T} \frac{d}{ds}U(s',s)G(s) + U(s',s) \ket{\phi_0(s)}\bra{\phi_0(s)}\text{.}\label{eq:explicitU}
\end{align}
We want to bound the following expression:
\begin{align*}
\left\Vert U(1,0) \ket{\eta(0)} - e^{i\beta(s)}\ket{\eta(1)}\right\Vert
&= \Vert U(1,0)\ket{\phi_0(0)} - \ket{\phi_0(1)}\Vert\\
&= \bigg\Vert \int_0^1 \frac{d}{ds} \big(U(1,s)\ket{\phi_0(s)}\big) ds\bigg\Vert\\
&= \bigg\Vert \int_0^1 \underbrace{\frac{d}{ds}U(1,s)}_\eqref{eq:adjointU} \ket{\phi_0(s)} + U(1,s) \underbrace{\ket{\phi'_0(s)}}_{\substack{\text{lemma}\\\ref{lem:derivPhi}}} ds\bigg\Vert\\
&= \Bigg\Vert \int_0^1 iT U(1,s) \underbrace{H(s) \ket{\phi_0(s)}}_{=0} - \underbrace{U(1,s)}_\eqref{eq:explicitU} G(s) H'(s) \ket{\phi_0(s)} ds\Bigg\Vert\\
&= \Bigg\Vert \int_0^1 \Big(\frac{i}{T} \frac{d}{ds}U(1,s)\Big) \Big(G^2(s) H'(s) \ket{\phi_0(s)}\Big)\\
&\hspace{3.5em} - U(1,s) \ket{\phi_0(s)}\underbrace{\bra{\phi_0(s)} G(s)}_{=0} H'(s) \ket{\phi_0(s)} ds\Bigg\Vert\\
&= \Bigg\Vert \bigg[\frac{i}{T} U(1,s)G^2(s) H'(s) \ket{\phi_0(s)}\bigg]_{s=0}^{s=1}\\
&\hspace{3.5em}- \int_0^1 \Big(\frac{i}{T} U(1,s)\Big) \frac{d}{ds}\Big(G(s)^2 H'(s) \ket{\phi_0(s)}\Big) ds\Bigg\Vert\\
&\le \frac{1}{T}\bigg( \big\Vert  G(s)^2 H'(s)\big\Vert_{s=1} + \big\Vert  G(s)^2 H'(s)\big\Vert_{s=0}\\
&\hspace{3.5em}+ \int_0^1 \bigg\Vert \frac{d}{ds}\Big(G(s)^2 H'(s) \ket{\phi_0(s)}\Big)\bigg\Vert ds \bigg)\text{.}
\end{align*}
Using Lemma \ref{lem:integral} for the integral expression and substituting $\left\Vert G(s)\right\Vert = \frac{1}{\gamma(s)}$ leads to
\begin{align*}
\left\Vert U(1,0) \ket{\eta(0)} - e^{i\beta(s)}\ket{\eta(1)}\right\Vert
\le \frac{1}{T}\bigg(& \frac{1}{\gamma(1)^2} \left\Vert H'(s)\right\Vert_{s=1} + \frac{1}{\gamma(0)^2} \left\Vert H'(s)\right\Vert_{s=0}\\
&+ \int_0^1 \frac{5}{\gamma(s)^3} \left\Vert H'(s) \right\Vert^2 +\frac{1}{\gamma(s)^2} \left\Vert H''(s) \right\Vert ds \bigg)\text{.}
\end{align*}
Replacing $T$ by its definition finally results in
\begin{align*}
\left\Vert U(1,0) \ket{\eta(0)} - e^{i\beta(s)}\ket{\eta(1)}\right\Vert &\le \epsilon\text{,}
\end{align*}
as desired.
\end{proof}

\section{The quantum adiabatic theorem for projections}\label{sec:AQTproj}

By analogue to the last section we will present in this section the proof for the quantum adiabatic theorem for projection operators. The proof idea and most calculations will look very familar. The first two lemmata are the analogues of \ref{lem:derivPhi} and \ref{lem:integral}:

\begin{lem}\label{lem:derivP}
Let $H(s)$ be a finite dimensional, differentiable Hamiltonian with $\ket{\phi_j(s)}$, $0 \le j\le n-1$, its eigenstates corresponding to the eigenvalues $E_j(s)$ and $P(s)$ the projection operator onto its ground-states $\ket{\phi_i(s)}$, $0 \le i \le k-1$, with energy $E_i(s)=0$. Then it holds
\begin{align*}
P'(s) &= - G(s) H'(s) P(s) - P(s) H'(s) G(s)
\end{align*}
with
\begin{align*}
G(s):=\sum\limits_{j=k}^{n-1} \frac{\ket{\phi_j(s)}\bra{\phi_j(s)}}{E_j(s)}\text{.}
\end{align*}
\end{lem}
\begin{proof}
First of all we observe that
\begin{align*}
\frac{d}{ds}P(s) = \frac{d}{ds} \big(P^2(s)\big) = P'(s) P(s) + P(s) P(s')
\end{align*}
and hence
\begin{align*}
P(s) P'(s) P(s) = 2 P(s) P'(s) P(s) = 0\text{.}
\end{align*}
We will use this equality in the next calculation:
\begin{alignat}{3}
&\frac{d}{ds} (H(s) P(s)) = 0 &&\nonumber\\
&H(s) P'(s) = - H'(s) P(s) \hspace{12em} &&|\, \text{ multiply }G(s)\text{ from the left}\nonumber\\
&(\mathbb{I}-P(s))P'(s) = - G(s) H'(s) P(s)  &&|\,\text{ multiply }P(s)\text{ from the right}\nonumber\\
&P'(s) P(s) - \underbrace{P(s) P'(s) P(s)}_{=0} = - G(s) H'(s) P(s)\text{.}\label{eq:P1}&&
\end{alignat}
Analogously
\begin{alignat}{3}
&\frac{d}{ds} (P(s) H(s)) = 0 &&\nonumber\\
&P'(s) H(s)= - P(s) H'(s) \hspace{12em} &&|\, \text{ multiply }G(s)\text{ from the right} \nonumber\\
&P'(s) - P'(s) P(s) = - P(s) H'(s) G(s)\text{.}\label{eq:P2}&&
\end{alignat}
Adding equations \eqref{eq:P1} and \eqref{eq:P2} results in
\begin{align*}
P'(s) &= - G(s) H'(s) P(s) - P(s) H'(s) G(s)\text{.}\qedhere
\end{align*}
\end{proof}

\begin{lem}\label{lem:integralProj}
The quantities defined in the previous lemma fulfill
\begin{align*}
\left\Vert \frac{d}{ds} \left( G(s)^2 \frac{dH(s)}{ds} P(s) U(s,1) \right) \right\Vert
\le 6 \left\Vert G(s)\right\Vert^3 \left\Vert \frac{dH(s)}{ds}\right\Vert^2 + \left\Vert G(s) \right\Vert^2 \left\Vert \frac{d^2H(s)}{ds^2}\right\Vert
\end{align*}
with $U(s,s')$ the time evolution operator of the Hamiltonian $H(s)$ with time $t=sT$.
\end{lem}
\begin{proof}
According to Schr\"odinger's equation
\begin{align*}
i \frac{d}{dt}\ket{\psi(t/T)} &= H(t/T)\ket{\psi(t/T}\quad\quad\quad | \, s=\frac{t}{T}, \, \frac{ds}{dt} = \frac{1}{T}\\
\frac{i}{T} \frac{d}{ds}\ket{\psi(s)} &= H(s)\ket{\psi(s)}
\end{align*}
the time evolution operator $U(s,s')$ is the solution of
\begin{align}
\frac{i}{T} \frac{d}{ds}U(s,s') &= H(s)U(s,s'), \quad U^\dagger(s,s')=U(s',s), \quad U(s,s) = \mathbb{I}\text{.}\label{eq:UProj}
\end{align}

We keep this in mind for later and look at the operator
\begin{align*}
\tilde{H}(s) := H(s) + P(s)\text{,}
\end{align*}
which is invertible and fulfills like every invertible differentiable operator
\begin{align}
\frac{d\tilde{H}(s)^{-1}}{ds} = \left( \frac{d\tilde{H}(s)^{-1}}{ds} \tilde{H}(s) \right) \tilde{H}(s)^{-1} 
= - \tilde{H}(s)^{-1} \frac{d\tilde{H}(s)}{ds} \tilde{H}(s)^{-1}\text{.}\label{eq:invertHProj}
\end{align}
With the help of Lemma \ref{lem:derivP} we find for the derivative of $\tilde{H}(s)$
\begin{align*}
\tilde{H}'(s) = H'(s) + P'(s) = H'(s) - G(s) H'(s) P(s) - P(s) H'(s) G(s)\text{.}
\end{align*}
Since $G(s)P(s) = P(s)G(s)=0$ it follows
\begin{align}
G(s) \tilde{H}'(s) G(s) = G(s) H'(s) G(s)\text{.}\label{eq:derivHtildeProj}
\end{align}
With
\begin{align*}
Q(s):= \mathbb{I} - P(s)
\end{align*}
it follows that
\begin{align*}
G(s) = Q(s) \tilde{H}(s)^{-1} = \tilde{H}(s)^{-1} Q(s) = Q(s) \tilde{H}(s)^{-1} Q(s)
\end{align*}
and we can derive the following expression for the derivative of $G(s)$:
\begin{alignat*}{3}
\frac{dG(s)}{ds} &= \frac{d}{ds} \left( Q(s) \tilde{H}(s)^{-1} Q(s) \right) && \\
&= \frac{dQ(s)}{ds} \tilde{H}(s)^{-1} Q(s) + Q(s) \frac{d\tilde{H}(s)^{-1}}{ds} Q(s) + Q(s) \tilde{H}(s)^{-1} \frac{Q(s)}{ds} \hspace{3em} &&| \, \eqref{eq:invertHProj}\\
&=\frac{dQ(s)}{ds} G(s) - Q(s) \tilde{H}(s)^{-1} \frac{d\tilde{H}(s)}{ds} \tilde{H}(s)^{-1} Q(s) + G(s) \frac{dQ(s)}{ds} &&| \, \eqref{eq:derivHtildeProj}\\
&=-P'(s) G(s) - G(s) H'(s) G(s) - G(s) P'(s) &&|\, \text{Lemma } \ref{lem:derivP}\\
&= P(s) H'(s) G(s)^2 - G(s) H'(s) G(s) - G(s)^2 H'(s) P(s)\text{.} &&
\end{alignat*}

With the above expression for $G'(s)$ and the help of the triangle inequality we can derive our desired result:
\begin{alignat*}{3}
\bigg\Vert \frac{d}{ds} &\bigg( G(&&s)^2 \frac{dH(s)}{ds} P(s) U(s,1) \bigg) \bigg\Vert\\
&\le \bigg\Vert && \frac{dG(s)}{ds} G(s) \frac{dH(s)}{ds} P(s) U(s,1) + G(s)\frac{dG(s)}{ds} \frac{dH(s)}{ds} P(s) U(s,1)\\
&&& + G(s)^2 \frac{dH(s)}{ds} \underbrace{P'(s)}_{\substack{\text{Lemma }\\ \ref{lem:derivP}}} U(s,1) + G(s)^2 \frac{dH(s)}{ds} \underbrace{P(s) \frac{dU(s,1)}{ds}}_{\substack{\overset{\ref{eq:UProj}}{=} P(s) H(s) U(s,1)\frac{T}{i}\\ = 0}} \bigg\Vert
+\Vert G(s) \Vert^2 \left\Vert \frac{d^2H(s)}{ds^2} \right\Vert\\
&= \bigg\Vert&& \left(P(s) \frac{dH(s)}{ds} G(s)^3 \frac{dH(s)}{ds} P(s) U(s,1) - G(s) \frac{dH(s)}{ds}G(s)^2\frac{dH(s)}{ds} P(s) U(s,1) \right)\\
&&&+ \left(- G(s)^2 \frac{dH(s)}{ds}G(s) \frac{dH(s)}{ds} P(s)U(s,1) +  G(s)^3 \frac{dH(s)}{ds} P(s)\frac{dH(s)}{ds} P(s)U(s,1)\right)\\
&&& +\left(- G(s)^2 \frac{dH(s)}{ds}G(s) \frac{dH(s)}{ds} P(s)U(s,1) - G(s)^2 \frac{dH(s)}{ds}P(s) \frac{dH(s)}{ds} G(s)U(s,1) \right)\bigg\Vert \\
&&&+ \left\Vert G(s) \right\Vert^2 \left\Vert \frac{d^2H(s)}{ds^2} \right\Vert\\
&\le 6&& \left\Vert G(s) \right\Vert^3 \left\Vert\frac{dH(s)}{ds}\right\Vert^2 
+ \left\Vert G(s) \right\Vert^2 \left\Vert \frac{d^2H(s)}{ds^2} \right\Vert\text{.}
\hspace{16.2em}\qedhere
\end{alignat*}
\end{proof}

Now we combine the above results to finally prove the quantum adiabatic theorem for projection operators:

\begin{thm}[Quantum Adiabatic Theorem for projection operators]\label{thm:AQTproj}
Let $H(s)$ be a finite-dimensional Hamiltonian such that for all $0 \le s \le 1$ it is twice differentiable and let $P(s)$ be the projection operator onto the ground-space which is separated by an energy gap $\gamma(s)$ from the first excited state. Choose $\epsilon > 0$ arbitrarily and define
\begin{align*}
T :=
 \frac{2}{\epsilon} \left( \frac{1}{\gamma(0)^2} \left\Vert\frac{\mathrm{d} H}{\mathrm{d} s} \right\Vert_{s=0}
+ \frac{1}{\gamma(1)^2} \left\Vert\frac{\mathrm{d} H}{\mathrm{d} s} \right\Vert_{s=1}
+ \int_0^1 \mathrm{d}s \left( \frac{6}{\gamma(s)^3} \left\Vert\frac{\mathrm{d} H}{\mathrm{d} s} \right\Vert^2
+ \frac{1}{\gamma(s)^2} \left\Vert\frac{\mathrm{d}^2 H}{\mathrm{d} s^2} \right\Vert \right)
\right)\text{.}
\end{align*}
Consider a system evolving according to the Hamiltonian $H(t/T)$. Denote by $U\left(\frac{t_1}{T}, \frac{t_0}{T}\right)$ its time evolution operator. Then it holds:
\begin{align}
\left\Vert U(1,0) P(0) U(0,1) - P(1) \right\Vert \le \epsilon\text{.}\label{eq:PineqProj}
\end{align}
\end{thm}
\begin{proof}
Let $\ket{\phi_j(s)}$, $0 \le j\le n-1$, be the eigenstates of the Hamiltonian corresponding to the eigenvalues $E_j(s)$ and $\ket{\phi_i(s)}$, $0\le i \le k-1$ its ground-states with eigenvalue $E_i=0$.

The evolution operator $U_0(s,s')$ of the Hamiltonian $H_0(s):=H(s)-E_0(s)$ equals the evolution operatur $U(s,s')$ up to a phase:
\begin{align*}
U(s,0)&=e^{iT\int_0^s E_0(u) du} U_0(s,0)
\end{align*}
which cancels in the unitary transformation of an operator such as $P(0)$. Hence
\begin{align*}
\left\Vert U(1,0) P(0) U(0,1) - P(1)\right\Vert = \left\Vert U_0(1,0) P(0) U_0(0,1) - P(1)\right\Vert\text{.}
\end{align*}
Therefore it is sufficient to prove the adiabatic theorem just for the case of zero ground-energy. So let's assume from now on $E_i(s)=0$ for all $0\le i \le k-1$.

As before define
\begin{align*}
G(s):=&\sum\limits_{j=k}^{n-1} \frac{\ket{\phi_j(s)}\bra{\phi_j(s)}}{E_j(s)}\\
Q(s):=& \mathbb{I} - P(s)
\end{align*}
and keep in mind that $G(s)$ ``inverts'' $H(s)$ up to the ground-space:
\begin{align*}
G(s)H(s) = H(s)G(s)=\mathbb{I} - P(s)=Q(s)\text{.}
\end{align*}

Taking the adjoint of the first equation in \eqref{eq:UProj} leads to
\begin{align}
-\frac{i}{T} \frac{d}{ds}U(s',s) &= U(s',s)H(s)\text{.}\label{eq:adjointUProj}
\end{align}
Later we will need an expression for $U(s',s)$, so we multiply $G(s)$ from the right and get:
\begin{align}
-\frac{i}{T} \frac{d}{ds}U(s',s) G(s)&= U(s',s) \left(\mathbb{I} - P(s)\right)\nonumber\\
U(s',s) &= -\frac{i}{T} \frac{d}{ds}U(s',s)G(s) + U(s',s) P(s)\label{eq:explicitUProj}
\end{align}
We look for a bound on
\begin{align*}
\Vert U(&1,0) P(0) U(0,1) - P(1)\Vert\\
&= \bigg\Vert \int_0^1 \frac{d}{ds} \big(U(1,s)P(s) U(s,1)\big) ds\bigg\Vert\\
&= \bigg\Vert \int_0^1 \bigg(\underbrace{\frac{d}{ds}U(1,s)}_\eqref{eq:adjointUProj}  P(s) U(s,1) +h.c.\bigg) + U(1,s) \underbrace{P'(s)}_{\substack{\text{Lemma}\\\ref{lem:derivP}}} U(s,1) ds\bigg\Vert\\
&= \Bigg\Vert \int_0^1 \Big(iT U(1,s) \underbrace{H(s) P(s)}_{=0} U(s,1) + h.c.\Big)
- \Big( \underbrace{U(1,s)}_{\eqref{eq:explicitUProj}} G(s) H'(s) P(s) U(s,1) + h.c. \Big) ds\Bigg\Vert\\
&= 2\Bigg\Vert \int_0^1 \Big(\frac{i}{T} \frac{d}{ds}U(1,s)\Big) \Big(G^2(s) H'(s) P(s) U(s,1)\Big)\\
& \hspace{4em}- U(1,s) \underbrace{P(s) G(s)}_{=0} H'(s) P(s)U(s,1) ds\Bigg\Vert\\
&= 2\Bigg\Vert \bigg[\frac{i}{T} U(1,s)G^2(s) H'(s) P(s) U(s,1)\bigg]_{s=0}^{s=1}\\
&\hspace{4em}- \int_0^1 \Big(\frac{i}{T} U(1,s)\Big) \frac{d}{ds}\Big(G(s)^2 H'(s) P(s) U(s,1)\Big) ds\Bigg\Vert\\
&\le \frac{2}{T}\bigg(\big\Vert  G(s)^2 H'(s)\big\Vert_{s=1} + \big\Vert  G(s)^2 H'(s)\big\Vert_{s=0}\\
&\hspace{4em}+ \int_0^1 \bigg\Vert \frac{d}{ds}\Big(G(s)^2 H'(s) P(s)U(s,1)\Big)\bigg\Vert ds \bigg)\text{.}
\end{align*}
Using Lemma \ref{lem:integralProj} for the integral expression and substituting $\left\Vert G(s)\right\Vert = \frac{1}{\gamma(s)}$ leads to
\begin{align*}
\Vert U(1,0) P(0) U(0,1) - P(1)\Vert
\le \frac{2}{T}\bigg(& \frac{1}{\gamma(1)^2} \left\Vert H'(s)\right\Vert_{s=1} + \frac{1}{\gamma(0)^2} \left\Vert H'(s)\right\Vert_{s=0}\\
&+ \int_0^1 \frac{6}{\gamma(s)^3} \left\Vert H'(s) \right\Vert^2 +\frac{1}{\gamma(s)^2} \left\Vert H''(s) \right\Vert ds \bigg)
\end{align*}
and replacing $T$ by its definition finally results in
\begin{align*}
\Vert U(1,0) P(0) U(0,1) - P(1)\Vert &\le \epsilon\text{,}
\end{align*}
as desired.
\end{proof}

\chapter{Quantum Circuit Simulation via Adiabatic Quantum Computation}\label{chap:model}
\ihead{CHAPTER 5.\,\, QUANTUM CIRCUIT SIMULATION VIA AQC}

\section{Standard graph Hamiltonians for circuit simulations}

In Section \ref{sec:AQCrequirements} we gave a list of requirements that a Hamiltonian has to fulfill in order to be suitable for an efficient adiabatic quantum computation. We formulated the objective of optimizing the evolution time $T$, which is determined by the derivatives of the Hamiltonian and the gap between the ground-space and the first excited state. Unfortunately it can sometimes be a difficult task to compute the gap of a Hamiltonian. But for a large family of Hamiltonians we can easily derive the gap by taking advantage of known results from spectral graph theory: Normalized Laplacian of parallel transport networks are hermitean according to their definition and their gap, according to Corollary \ref{cor:spectrumPTN}, exactly equals the gap of the underlying graph family. As a further advantage the gap is even time-independant as long as only the edge unitaries are time-dependant.

The question is now how the normalized Laplacian of a parallel transport network can encode a computation. Let's assume that the computation is already formalized in the most natural way, that is, in terms of a quantum circuit depending on the computation input. For a circuit simulation via adiabatic quantum computation we require that after a measurement on any final ground-state of the Hamiltonian and the possible discarding of a subsystem we obtain a correct computation output with nonvanishing probability (remember that in Section \ref{sec:circuits} we allowed a whole set of valid circuit inputs that lead to correct computation outputs). Therefore we require that our normalized Laplacian implements the corresponding time-dependant circuit to an identity extension of the original circuit (recall Definitions \ref{def:correspondingCircuit} and \ref{def:identityExtension}). Notice that we distinguish between ``simulating'' and ``implementing'' a circuit:

\begin{definition}\label{def:implement}
A connected parallel transport network $\mathcal{G}=(G,\mathcal{T},\mathcal{U},\pi)$ with more than one vertex \textsf{implements} a quantum circuit $C'=\big(n, (U_1, \dots U_{L'})\big)$, iff $\mathcal{T}: V \rightarrow \{0, \dots, L'\}$ is surjective and $\pi(\mathcal{U}(t))= U_t$ for all $1 \le t \le L'$.

A parallel transport network \textsf{simulates} a quantum circuit $C$, iff it implements the corresponding time-dependant circuit to an identity extension of $C$.
\end{definition}

Requiring the time map $\mathcal{L}$ of the parallel transport network to be surjective implies $V_t\ne \emptyset$ for all $0\le t \le L'$. We know from Corollary \ref{cor:history} that the ground-space of our parallel transport network is spanned by the history states
\begin{align*}
\ket{\eta(s)}= \frac{1}{\sqrt{\vol(V)}} \sum_{t=0}^{L'} \sum_{v \in V_t} \sqrt{d_v} \ket{v} \otimes U_t(s) \dots U_1(s) \ket{x}\text{,}\quad x\in \{0,1\}^n\text{.}
\end{align*}
So if we measure on the graph subsystem of any final ground-state in the vertex basis and trace out this subsystem afterwards we obtain a linear combination of the states $U_{L'} \dots U_1 \ket{x}$, $x\in \{0,1\}^n$, with a nonvanishing probability. Since there is no restriction on $\ket{x}$ this can be any $n$-qubit state, but we only want to obtain a correct computation output, which equals $U_{L'} \dots U_1 \ket{x}$ with $x \in S$ a valid circuit input. So we add a so-called penalty term to the normalized Laplacian of the parallel transport network such that only history states with $x\in S$ remain as ground-states.

In Section \ref{sec:AQCrequirements} we required the Hamiltonian to be defined on a qubit space. To avoid the restriction to graphs whose vertex set size is a power of $2$, we consider only a subspace of the whole Hilbert space as actual space of the parallel transport network and introduce a further penalty term that ensures that the gap of the Hamiltonian is determined by the gap of this subspace.

All these ideas together lead now to the definition of the type of Hamiltonian for adiabatic quantum computation we want to concentrate on for the rest of this thesis:
\begin{definition}\label{def:standard}
A \textsf{standard graph Hamiltonian} for simulating a quantum circuit $C=\big(n, (\tilde{U}_1, \dots \tilde{U}_{L})\big)$ via adiabatic quantum computation is defined on the tensor product of a $m$-qubit \textsf{extended graph space} and a $n$ qubit \textsf{computation space} and has the form
\begin{align*}
H(s) = H_{prop}(s) + H_{in} + H_{graph}\text{.}
\end{align*}

The Hamiltonian $H(s)$ leaves the subspace $D$ of \textsf{proper network states} and the orthogonal subspace $D^\perp$ invariant. $H_{prop\vert D}(s)$ is the normalized Laplacian of a parallel transport network $\mathcal{G}=(V,E,w,\mathcal{T},\mathcal{U},\pi)$ implementing the quantum circuit $C(s)=\big(n, (U_1(s),\dots U_{L'}(s))\big)$, the corresponding time-dependant circuit to an identity extension  of $C$ with $L_i$ initial and $L_f$ final identity gates.

The term $H_{in}$ is positive semi-definite and the null-space of $H_{in\vert D}$ is spanned by 
\begin{align*}
\big\{ \ket{v} \otimes \ket{x} \,\big\vert\, \mathcal{T}(v) \le L_i,\, x \in S\big\}
\cup \big\{ \ket{v} \otimes \ket{x} \,\big\vert\, \mathcal{T}(v)>L_i,\, x \in \{0,1\}^n\big\}
\end{align*}
with $S$ the set of valid circuit inputs.

$H_{graph}$ vanishes on $D$ and has a positive expectation value for all $\ket{\phi} \in D^\perp$.
\end{definition}

\begin{definition}
The graph $G$ of the previous definition is called the \textsf{underlying graph} of the Hamiltonian $H(s)$, $\mathcal{G}$ the \textsf{underlying parallel transport network}. The vertex set $\bigcup_{t=0}^{L_i} V_t$ is called the \textsf{set of initial vertices}, the set $\bigcup_{t=L-L_f}^{L} V_t$ the \textsf{set of final vertices}.
\end{definition}

The orthonormal vertex states $\ket{v}$, $v\in V$ can no longer be assumed to span the whole $m$-qubit extended graph space, hence we define a new extended basis:
\begin{definition}
Given a standard graph Hamiltonian with underlying graph $G=(V,E,w)$ the extension of the orthonormal vertex states $\ket{v}$, $v\in V$, to a orthonormal basis of the $m$-qubit extended graph space is called the \textsf{extended vertex basis}.
\end{definition}

The null-space of the penalty term $H_{in}$ is defined such that the null-space of the whole Hamiltonian $H(s)$ is spanned by
\begin{align}
\ket{\eta(s)}= \frac{1}{\sqrt{\vol(V)}} \sum_{t=0}^{L'} \sum_{v \in V_t} \sqrt{d_v} \ket{v} \otimes U_t(s) \dots U_1(s) \ket{x}\text{,}\quad x\in S\text{,}\label{eq:finalGroundstates}
\end{align}
as wanted. Moreover the null-space of $H_{in}$ has the above definition since it can be realized by a term of the form
\begin{align*}
H_{in}= \sum_{t=0}^{L_i} \sum_{v \in V_t}\sum_{x\notin S} a_{v,x} \ket{v}\bra{v} \otimes \ket{x}\bra{x}\text{,}
\end{align*}
which is often the easiest expression for $H_{in}$ that can be written in local terms. As we will see later the gap of the Hamiltonian $H(s)$ will grow with the number of vertex projections $\ket{v}\bra{v}$ in the above sum. But the idea to simply enlarge the sum over relevantly many vertex projections $\ket{v}\bra{v}$ with $\mathcal{T}(v)> L_i$ causes locality problems, since the vertex projections would have to be tensored by $U_t(s) \dots U_1(s) \ket{x}\bra{x}U_t^\dagger(s) \dots U_1^\dagger(s)$, $x\in \overline{S}$, for arbitrary circuits in order to obtain the same null-space of the Hamiltonian $H(s)$. The next idea that naturally comes to mind is to change the circuit by inserting the inverse gates $U_t^\dagger(s), \dots U_1^\dagger(s)$ after the gate $U_t(s)$ and to include the terms $\ket{v}\bra{v}\otimes\ket{x}\bra{x}$ with $v$ a vertex belonging to the time step $2t$ into the above sum. But instead of executing and then inverting again a part of the computation one could also substitute these gates by identity gates and raise the value of $L_i$, which recovers our orignal construction of $H_{in}$. As long as we don't have any further information about the specific circuit this construction seems to be the simplest and best one possible.

\section{Specification of an efficient adiabatic quantum computation with standard graph Hamiltonians}\label{sec:circuitSimReq}

In this section we want to discuss the requirements for standard graph Hamiltonians used for an efficient quantum circuit simulation via adiabatic quantum computation and derive some useful expressions for the crucial complexity quantities. The first complexity quantity we want to look at is the probability to obtain a correct computation output at the end of the adiabatic quantum computation. From expression \eqref{eq:finalGroundstates} for the final ground-states it follows directly:
\begin{prop}\label{prop:outputProb}
Let $H(s)$ be a standard graph Hamiltonian according to Definition \ref{def:standard}. Measuring in the extended vertex basis on the extended graph subsystem of a final ground-state and discarding the graph subsystem afterwards, a correct computation output is obtained with probability
\begin{align*}
p &= \frac{\sum\limits_{t=L'-L_f}^{L'} \vol(V_t)}{\vol(V)}\text{.}\hspace{29.1em}\qed
\end{align*}
\end{prop}

The next important efficiency quantity is the gap of the Hamiltonian. The contruction of a standard graph Hamiltonian allows us first of all to break down the gap to the gap on the subspace $D$ of proper network states:
\begin{lem}\label{lem:gaps}
Denote by $\gamma(H)$ the energy gap of a Hamiltonian $H$ between its ground- and its first excited eigenstate. Let $H(s)$ be a standard graph Hamiltonian according to Definition \ref{def:standard}. Then it holds
\begin{align*}
\gamma\big(H(s)\big)=\min\left\{\gamma(H_{graph}),\gamma\left(H_{prop\vert D}(s) + H_{in \vert D}\right) \right\}\text{.}
\end{align*}
\end{lem}
\begin{proof}
Firstly, we know that $H(s)$ has ground-energy $0$, hence the gap $\gamma(H(s))$ equals its first nonzero eigenvalue. Furthermore as $D$ is the null-space of $H_{graph}$ the gap $\gamma(H_{graph})$ is the lowest expectation value of $H_{graph}$ regarding any vector from $D^\perp$.

According to Definition \ref{def:standard} $H(s)$ leaves the subspace of proper network states $D$ and its orthogonal subspace $D^\perp$ invariant, hence they are both spanned by eigenvectors of $H(s)$. If the corresponding eigenvector $\ket{n}$ to $\gamma(H(s))$ is from $D^\perp$, the following clearly holds:
\begin{align*}
\gamma(H(s)) \ge \braket{n |H_{graph} | n} \ge \gamma(H_{graph})\text{.}
\end{align*}
If on the other hand $\ket{n}\in D$, then obviously
\begin{align*}
\gamma(H(s)) =\gamma\left(H_{prop\vert D}(s) + H_{in\vert D}\right)\text{,}
\end{align*}
leading to the desired bound.
\end{proof}

The idea is to define $H_{graph}$ for all $\ket{\Psi} \in D^\perp$ with expectation values higher than the gap of $H_{prop\vert D}(s) + H_{in\vert D}$. Then this gap determines the gap of $H(s)$. We know that the gap of $H_{prop \vert D}(s)$ equals the spectral gap of the parallel transport network, but to calculate the gap of the sum $H_{prop\vert D}(s) + H_{in\vert D}$ we need the next lemma which is a consequence from a lemma by Kitaev presented in \cite{npsurvey}.

\begin{lem}\label{lem:angle}
Let $H_1$ and $H_2$ be two Hermitean positive semidefinite matrices with nonempty null-spaces $N_1$ and $N_2$ and $\lambda_1$ and $\lambda_2$ the lowest nonzero eigenvalues, respectively. Define $\mu:= \min(\lambda_1, \lambda_2)$ and denote by $\theta$ the angle between $N_1$ and $N_2 \setminus N_1\ne \emptyset$. Then, the lowest non-zero eigenvalue of $H_1 + H_2$, denoted by $\lambda(H_1 + H_2)$, can be bounded in the following way:
\begin{align*}
\mu \sin^2 \left(\frac{\theta}{2}\right) \le \lambda(H_1+H_2) \le \Vert H_1\Vert \sin^2(\theta)\text{.}
\end{align*}
\end{lem}
\begin{proof}
Define $\hat{H}_2:= H_2 + \Vert H_1+ H_2 \Vert P_N$ with $P_N$ the projection operator onto $N:= N_1\cap N_2$. The elements of $N$ are the eigenvectors of $H_1 + H_2$ corresponding to the eigenvalue $0$ and eigenvectors of $H_1+ \hat{H}_2$ corresponding to the highest eigenvalue. Since $H_1+H_2$ and $H_1 + \hat{H}_2$ are diagonalizable in the same basis, the rest of the spectrum of $H_1 + \hat{H}_2$ compared to $H_1 + H_2$ is unchanged. Hence the lowest eigenvalue of $H_1 + \hat{H}_2$ equals the lowest nonzero eigenvalue of $H_1 + H_2$.

Consequently it is sufficient to show that
\begin{alignat*}{3}
1) \quad&\mu \sin^2 \left(\frac{\theta}{2}\right) \le \braket{\Psi | H_1+\hat{H}_2 | \Psi}
\quad &&\text{ for all unit vectors } \ket{\Psi}\\
2) \quad &\braket{\phi | H_1+\hat{H}_2 | \phi} \le \Vert H_1\Vert \sin^2(\theta)
\quad&&\text{ for some unit vector } \ket{\phi}\text{.}
\end{alignat*}

We denote the null-space of $\hat{H}_2$ by $\hat{N}_2 = N_2 \setminus N_1$. The angle between the null-spaces of $H_1$ and $\hat{H}_2$ is per definition $\theta$ and $\lambda_2$ is not only the lowest nonzero eigenvalue of $H_2$ but also the lowest nonzero eigenvalue of $\hat{H}_2$.

We first show the lower bound. Let $\ket{\Psi}$ be an arbitrary unit vector. The angle between $\ket{\Psi}$ and one of the null-spaces $N_1$ or $\hat{N}_2$ is at least $\frac{\theta}{2}$. W.l.o.g. assume that $N_1$ is this space. We write $\ket{\Psi} = \ket{n_1} + \ket{n_1^\perp}$ with $\ket{n_1}$ the projection of $\ket{\Psi}$ onto $N_1$ and $\ket{n_1^\perp}$ the projection onto the orthogonal subspace.

Now we can approximate
\begin{align*}
\braket{\Psi | H_1 + \hat{H}_2 | \Psi} &\ge \braket{\Psi | H_1 | \Psi}\\
&= \braket{n_1^\perp | H_1 | n_1^\perp}\\
&\ge \mu \left\Vert \ket{n_1^\perp} \right\Vert^2\\
&\ge \mu \sin^2\left( \frac{\theta}{2}\right)
\end{align*}
where $\left\Vert \ket{n_1^\perp} \right\Vert \ge \sin\left( \frac{\theta}{2}\right)$ because the angle between $\ket{\Psi}$ and $N_1$ is at least $\frac{\theta}{2}$.

It remains to show the upper bound. Let $\ket{\phi} \in \hat{N}_2$ be the unit vector with angle $\theta$ to $N_1$. We again write $\ket{\phi} = \ket{n_1} + \ket{n_1^\perp}$ with $\ket{n_1}$ the projection of $\ket{\phi}$ onto $N_1$ and $\ket{n_1^\perp}$ the projection onto the orthogonal subspace. We can now bound the expectation value:
\begin{align*}
\braket{\phi | H_1 + \hat{H}_2 | \phi} &= \braket{\phi | H_1 | \phi}\\
&= \braket{n_1^\perp | H_1 | n_1^\perp}\\
&\le \Vert H_1 \Vert \left\Vert \ket{n_1^\perp} \right\Vert ^2\\
&= \Vert H_1 \Vert \sin^2(\theta)\text{.}\qedhere
\end{align*}
\end{proof}

In most cases we will define the penalty term $H_{in}$ such that its gap is larger than the spectral gap $\lambda$ of the parallel transport network. Then the gap $\gamma$ of the whole Hamiltonian $H(s)$ which determines the running time of the adiabatic quantum computation is lower bounded by
\begin{align*}
\gamma \ge \lambda \sin^2\left(\frac{\theta}{2}\right)
\end{align*}
with $\theta$ the angle between the null-space $N_{in}$ of $H_{in\vert D}$ and $N_{prop}(s)\setminus N_{in}$ with $N_{prop}(s)$ the null-space of $H_{prop\vert D}(s)$. Also the upper bound of $\gamma$ scales with $\sin^2(\theta)$, hence if the angle $\theta$ is too small an efficient adiabatic quantum computation is not possible. Because $\theta$ is such an important complexity quantity we will give it an own name:
\begin{definition}
Given a standard graph Hamiltonian according to Definition \ref{def:standard} let $N_{in}$ be the null-space of $H_{in\vert D}$ and $N_{prop}(s)$ the null-space of $H_{prop\vert D}(s)$. The angle $\theta$ between $N_{in}$ and $N_{prop}(s)\setminus N_{in}$ is called the \textsf{gap angle}.
\end{definition}

The optimization of the energy gap $\gamma$ for an efficient evolution time in the adiabatic quantum computation is now reduced to the optimization of the spectral gap of the underlying graph and the gap angle. But notice that if the spectral gap and the gap angle get too large then the gaps of the penalty terms $H_{in}$ and $H_{graph}$ might determine $\gamma$. As $H_{in}$ and $H_{graph}$ are supposed to be normalized, they often obey gaps that scale like $\Theta\left(\frac{1}{L}\right)$. Hence if one optimizes the spectral gap of the underlying graph beyond $\Theta\left(\frac{1}{L}\right)$, one has to consider again the gaps of the penalty terms.

It is, by the way, an interesting fact that the normalization of the penalty terms $H_{in}$ and $H_{graph}$ do not affect the  evolution time in the Adiabatic Theorems \ref{thm:AQTstates} and \ref{thm:AQTproj}. Only the norms of the derivatives appear in the expression and these vanish for the time-independant penalty terms. Still we stick to our normalization requirement of the whole Hamiltonian because of the reasons discussed in Section \ref{sec:norm}. Moreover it might be that the independance from the norms of the penalty terms is only a feature of the specific bound for the evolution time that is derived in the adiabatic theorems.

Let's turn our focus back to the improvement of the spectral gap and the gap angle. Our definition of a standard graph Hamiltonian allows us to derive a very simple expression for the gap angle:
\begin{prop}\label{prop:angleExpl}
Let $H(s)$ be a standard graph Hamiltonian according to Definition \ref{def:standard}. Then its gap angle $\theta$ fulfills
\begin{align*}
\sin^2(\theta) = \frac{\sum\limits_{t=0}^{L_i}\vol(V_t)}{\vol(V)}\text{.}
\end{align*}
\end{prop}
\begin{proof}
Let $N_{prop\vert D}(s)$ be the null-space of $H_{prop\vert D}(s)$ and $N_{in}(s)$ the null-space of $H_{in\vert D}(s)$. According to Definition \ref{def:standard} and since a unitary just causes a basis transformation we can write
\begin{align*}
N_{in}(s)
&=\spanN \big(\big\{ \ket{v} \otimes \ket{x} \,\big\vert\, \mathcal{T}(v) \le L_i,\, x \in S\big\} \\
&\hspace{3.8em} \cup\big\{ \ket{v} \otimes U_t(s) \dots U_1(s) \ket{x} \,\big\vert\, \mathcal{T}(v)>L_i,\, x \in \{0,1\}^n\big\}\big)
\end{align*}
and hence
\begin{align*}
N_{prop}(s) \setminus N_{in\vert D}(s) = \spanN \left\{ \frac{1}{\sqrt{\vol(V)}} \sum_{t=0}^{L'} \sum_{v\in V_t} \sqrt{d_v} \ket{v} \otimes U_t(s)\dots U_1(s) \ket{x} \,\big\vert\, x\in \overline{S}\right\}\text{.}
\end{align*}
Let $\ket{n_{in}(s)} \in N_{in}(s)$ and $\ket{n_{prop}(s)} \in N_{prop}(s) \setminus N_{in}(s)$ be arbitrary unit vectors. Then they can be expressed in the following form:
\begin{align}
\ket{n_{in}(s)}=& \sum_{t=0}^{L'} \sum_{v \in V_t} p_v \ket{v} \otimes \sum_{x \in \{0,1\}^n} q_{t,x} U_t(s) \dots U_1(s) \ket{x}\nonumber\\
&p_v, q_{t,x} \in \mathbb{C},\quad \sum_{v\in V} |p_v|^2 = 1\text{,}\nonumber\\
& \sum_{x\in \{0,1\}^n} |q_{t,x}|^2 = 1\,\text{ for all } 0\le t\le L', \quad\quad q_{t,x}= 0 \,\text{ for } t\le L_i \text{ and } x\in \overline{S}\text{.}\label{eq:i}\\
\ket{n_{prop}}=& \frac{1}{\sqrt{\vol(V)}} \sum_{t=0}^{L'} \sum_{v \in V_t} \sqrt{d_v} \ket{v} \otimes U_t(s) \dots U_1(s) \sum_{x \in \overline{S}} r_x \ket{x}\nonumber\\
&r_x\in \mathbb{C}, \quad\sum_{x\in \overline{S}} |r_x|^2 = 1\text{.}\label{eq:ii}
\end{align}
Now we can derive the expression for the angle between $N_{in}(s)$ and $N_{prop}(s) \setminus N_{in}(s)$:
\begin{align*}
\cos^2(\theta) &= \max_{\substack{\ket{n_{prop}(s)} \in N_{prop}(s)\\\ket{n_{in}(s)}\in N_{in}(s)}} \left\vert\braket{n_{prop}(s)| n_{in}(s)}\right\vert^2\\
&= \max_{\substack{p_v,q_{t,x},r_x\\\text{fulfilling \eqref{eq:i}, \eqref{eq:ii}}}} \frac{1}{\vol(V)} \Bigg\vert \sum_{t=0}^{L'} \sum_{t'=0}^{L'} \sum_{v\in V_t} \sum_{v\in V_{t'}}\sqrt{d_v} p_{v'} \overbrace{\braket{v | v'}}^{=\delta_{v,v'} \delta_{t,t'}}\\
&\hspace{12em}\otimes \sum_{x\in \overline{S}} \sum_{x' \in \{0,1\}^n} r_x^* q_{t,x} \braket{x | U_1^\dagger(s) \dots U_t^\dagger(s) U_{t'}(s) \dots U_1(s) | x'} \Bigg\vert^2\\
&= \max_{\substack{p_v,q_{t,x},r_x\\\text{fulfilling \eqref{eq:i}, \eqref{eq:ii}}}} \frac{1}{\vol(V)} \Bigg\vert \sum_{t=0}^{L'} \sum_{v\in V_t} \sqrt{d_v} p_v \sum_{x\in \overline{S}} r_x^* q_{t,x} \Bigg\vert^2\text{.}
\end{align*}
Since $q_{t,x}=0$ for all $t\le L_i$ and all $x \in \overline{S}$ we can start summing the time steps from $t=L_i+1$. Moreover we can omit the modulus function as the maximum will be reached by real, non-negative coefficients $p_v$, $r_x$ and $q_{t,x}$ ( $|p_v|$, $|r_x|$, $|q_{t,x}|$ fulfill the upper bound of the triangle inequality):
\begin{align*}
\cos^2(\theta) &=\max_{\substack{p_v,q_{t,x},r_x\ge 0\\\text{fulfilling \eqref{eq:i}, \eqref{eq:ii}}}} \frac{1}{\vol(V)} \Bigg( \sum_{t=L_i+1}^{L'} \sum_{v\in V_t} \sqrt{d_v} p_v \sum_{x\in \overline{S}} r_x q_{t,x} \Bigg)^2\text{.}
\end{align*}
According to the Cauchy-Schwarz inequality $\sum_{x\in \overline{S}} r_x q_{t,x} \le 1$ and as the maximizing set includes for all $t>L_i$ cases with $q_{t,x}=r_x$ for all $x\in \{0,1\}^n$ which fulfill the upper bound it holds:
\begin{align*}
\cos^2(\theta) &=\max_{\substack{p_v\ge 0\\\text{fulfilling \eqref{eq:ii}}}} \Bigg( \sum_{t=L_i+1}^{L'} \sum_{v\in V_t}  \sqrt{\frac{d_v}{\vol(V)}} p_v \Bigg)^2\text{.}
\end{align*}
The argument runs similar now for $p_v$. According to Cauchy-Schwarz
\begin{align}
\Bigg(  \sum_{t=L_i+1}^{L'} \sum_{v\in V_t}   \sqrt{\frac{d_v}{\vol(V)}} p_v \Bigg)^2
\le  \sum_{t=L_i+1}^{L'} \sum_{v\in V_t}   \frac{d_v}{\vol(V)}
= 1-\frac{\sum_{t=0}^{L_i} \vol(V_t)}{\vol(V)}
=: c\label{eq:CS}
\end{align}
and the case
\begin{align*}
p_v=
\begin{cases}
\frac{1}{\sqrt{c}}\sqrt{\frac{d_v}{\vol(V)}} \quad &\text{if } v\in \bigcup_{t=L_i+1}^{L'} V_t\\
0 &\text{else}
\end{cases}
\end{align*}
is included in the maximizing set and leads to equality in \eqref{eq:CS}. Hence
\begin{align*}
\sin^2(\theta) &= \frac{ \sum_{t=0}^{L_i}\vol(V_t)}{\vol(V)}\text{.}\qedhere
\end{align*}
\end{proof}

The expression for the gap angle is surprisingly simple and given a standard graph Hamiltonian just determined by the graph. The gap angle is also independant of the set $S$ of correct circuit inputs. So the idea that multiple ground-states which correspond to multiple valid circuit inputs might improve the gap is wrong for the standard graph Hamiltonian model. This idea was the original motivation for deriving the Quantum Adiabatic Theorem for Projections \ref{thm:AQTproj}. But in the model of standard graph Hamiltonians we can also restrict simply to $\ket{\mathbf{0}}$ as only valid circuit input and thus to a unique ground-state in the adiabatic quantum computation. 
 If one wants to take an advantage out of multiple valid circuit input one has to go beyond the model of standard graph Hamiltonians and especially look at models where the circuit is not just a black box but instead some additional information about the specific computation is known. 

We can now specify and comment on our efficiency list from Section \ref{sec:AQCrequirements} for an efficient adiabatic quantum computation (AQC) under the assumption of a standard graph Hamiltonian. The initial complexity quantities are now the circuit length $L$ and its number of qubits $n$ instead of the length of the original computation input. Notice that for a full efficient computation the initial mapping from the computation input to the circuit should, of course, also be efficient.

For the rest of this thesis we make the realistic assumption that $n \in \mathcal{O}(L)$, otherwise the circuit would act trivially on some qubits for large enough $L$, which would allows us to define another circuit with less qubits for the same computation purpose. This assumption will allow us to write the bounds just in terms of $L$.

\tcbset{colback=black!5!white,colframe=black!50!white}
\vspace{2mm}
\begin{tcolorbox}[toptitle=2mm,fonttitle=\sffamily\bfseries\large,title=Requirements for an efficient AQC with standard graph Hamiltonians]
We assume the notation of a standard graph Hamiltonian $H(s)$ from Definition \ref{def:standard} simulating the quantum circuit $C=(n, (U_1, U_2, \dots U_{L}))$.
\tcblower
\begin{enumerate}
\item The Hamiltonian $H(s)$ is normalized and defined on a space of $\mathcal{O}(\poly(L))$ qubits.\\
\textit{\small Notice that the underlying graph can still have a vertex set of exponential size.}
\item The Hamiltonian is a sum of $\mathcal{O}(\poly(L))$ many, local interaction terms.
\item  For every interaction term $H_1(s)$ it holds $\Vert H_1(s) \Vert \in \Omega\left(\frac{1}{\poly(L)}\right)$ for all $0\le s \le 1$.\\
\textit{\small Notice that this implies a restriction onto the weights of the underlying graph $G$, since its normalized Laplacian contains interactions terms with factors $\frac{w(v,u)}{\sqrt{d_v d_u}}$.}
\item There is a classical algorithm that computes the interaction terms in time $\mathcal{O}(\poly(L))$.\\
\textit{\small Notice that given two vertices $u,v$ from the underlying graph, their weight $w(u,v)$ has to be computable in polynomial time.}
\item The state $\mathbf{0} \otimes \ket{\mathbf{0}}$ has to be transformed into the initial history state
\begin{align*}
\ket{\eta(0)}= \frac{1}{\sqrt{\vol(V)}} \sum_{v \in V} \sqrt{d_v} \ket{v} \otimes \ket{\mathbf{0}}
\end{align*}
 by a quantum circuit with length $\mathcal{O}(\poly(L))$.\\
\textit{\small As we introduced time-dependant unitary gates equaling initially $\mathbb{I}$ the initial history state is automatically independant of the computation.}
\item After a projective measurement on the extended graph subsystem of any ground-state of $H(1)$ in the extended vertex basis and the discarding of the subsystem, the correct computation output is obtained with probability $p\in \Omega\left(\frac{1}{\poly(L)}\right)$.\\
\textit{\small Proposition \ref{prop:outputProb}:
$p = \frac{\sum\limits_{t=L'-L_f}^{L'} \vol(V_t)}{\vol(V)}$.}
\item For the gap $\gamma$ of the Hamiltonian $H(s)$ it holds that $\gamma \in \Omega\left(\frac{1}{\poly(L)}\right)$.\\
\textit{\small Lemmata \ref{lem:gaps} and \ref{lem:angle}: 
$\gamma \ge \min\left\{\gamma(H_{graph\vert D}), \sin^2(\theta/2)\cdot \min\{\lambda, \gamma(H_{in\vert D})  \} \right\}$.\\
Proposition \ref{prop:angleExpl}: 
$\sin^2(\theta) = \frac{\sum\limits_{t=0}^{L_i}\vol(V_t)}{\vol(V)}$.}
\item  --\\
\textit{\small Requirement 8 is no longer neccessary, since in the next section we will show that the derivatives $\left\Vert  \frac{dH(s)}{ds} \right\Vert$ and $\left\Vert  \frac{d^2H(s)}{ds^2} \right\Vert$ are always constant for standard graph Hamiltonians.}
\end{enumerate}
\end{tcolorbox}
\vspace{2mm}

As commented under point 8 we will prove in the next section that the Hamiltonian derivatives that appear in the expression for the evolution time $T$ in the quantum adiabatic theorem will be upper bounded by a constant for standard graph Hamiltonians. Hence our objective to optimize $T$ reduces to optimizing the energy gap $\gamma$ of the Hamiltonian.

We want to finish this section with a closing remark about the technique of extending the circuit by identity gates and about the final measurement of the adiabatic quantum computation. In \cite{aharonov2}, among other works, it is required that a state arbitrarily close to the desired computation output is obtained directly at the end of the adiabatic evolution without a final measurement. This can be achieved by appending linearly many final identity gates in a standard graph Hamiltonian construction with underlying path graph (the so-called Kitaev Hamiltonian which we will present as first example in Chapter \ref{chap:path}). The reason is that the states belonging to final time steps dominate the history state. But this concept would forbid the addition of any initial identity gates as this would reduce again the dominance of the final vertex set.

This concept causes two problems regarding standard graph Hamiltonians with arbitrary underlying graphs: First of all there are graphs other than the path graph for which it is not sufficient to append linearly many final identity gates to achieve a history state which is arbitrarily close to the desired output state. Secondly some graphs require appending a certain ratio of initial identity gates to achieve an efficient, inverse polynomial gap at all since we know from Proposition \ref{prop:angleExpl} that the gap angle relies on the set of initial vertices. It would be a pity not to consider these graphs for standard graph Hamiltonian constructions, especially since a final projective measurement is a well-established tool in quantum computation. Therefore in Chapter \ref{chap:path} the Kitaev Hamiltonian will also be discussed under the conditions of our concept with identity extension and final measurement in order to achieve a comparibility to standard graph Hamiltonians based on other graphs.

\section{Norm of Hamiltonian derivatives}

The upper bound for the evolution time of a Hamiltonian according to the Adiabatic Theorems \ref{thm:AQTstates} and \ref{thm:AQTproj} contains the terms $\left\Vert \frac{d H(s)}{ds} \right\Vert$ and $\left\Vert \frac{d^2 H(s)}{ds^2} \right\Vert$, that is, the norms of the first and second derivative of the Hamiltonian with respect to the normalized time parameter $s$. The next theorem shows that these derivatives of a standard graph Hamiltonian (equaling the derivatives of the normalized Laplacian) do not play any role in the complexity discussion since they are always bounded from above by a constant, no matter which underlying graph is used nor which quantum circuit is encoded by the parallel transport network. 

\begin{thm}
Given a connected parallel transport network $\mathcal{G}=(V,E,w,\mathcal{T},n,\mathcal{U},\pi)$ with more than one vertex and time dependant unitaries $U_t(s)=\exp(i s h_t)$, $h_t$ hermitean, $\Vert h_t \Vert \le 2\pi$, the norms of the first and second derivative of the normalized Laplacian are bounded from above by a constant:
\begin{align*}
\left\Vert \frac{d\mathcal{L}(\mathcal{G})(s)}{ds} \right\Vert \in \mathcal{O}(1)\quad\text{ and }\quad
\left\Vert \frac{d^2\mathcal{L}(\mathcal{G})(s)}{ds^2} \right\Vert \in \mathcal{O}(1)\text{.}
\end{align*}
\end{thm}
\begin{proof}
We will show that the first derivative is bounded from above by a constant and then argue that the proof for the second derivative works analogously.

Let $\mathcal{T}:V \rightarrow \{0, \dots L\}$. From Lemma \ref{lem:laplExplPTN} we can conclude for the first derivative of the normalized Laplacian:
\begin{align*}
\frac{d \mathcal{L}(\mathcal{G})(s)}{d s} &= -\sum_{t=1}^L  \sum_{v\in V_t} \sum_{u\in V_{t-1}} \frac{w(v,u)}{\sqrt{d_v d_u}} \left(\ket{v}\bra{u} \otimes i U_t(s) h_t+ \ket{u}\bra{v} \otimes -i h_t U_t^\dagger(s)\right)\\
&= \sum_{\substack{t=1 \\ t \text{ even}}}^L M_t(s) + \sum_{\substack{t=1 \\ t \text{ odd}}}^L M_t(s)
\end{align*}
with
\begin{align}
M_t(s) &=  \sum_{v\in V_t} \sum_{u\in V_{t-1}}
\frac{w(v,u)}{\sqrt{d_v d_u}} \left(\ket{v}\bra{u} \otimes -i U_t(s)h_t + \ket{u}\bra{v} \otimes i h_t U_t^\dagger(s)\right)\text{.}\label{eq:Mt}
\end{align}

In writing the derivative as two sums, one over the even and one over the odd time steps, we achieve that each sum has block-diagonal form. In $\sum_{t=1,t \text{ even}}^L M_t(s)$ the blocks are restricted to the interaction terms between vertices of an even time step $t$ and vertices belonging to the next lower odd time step:

\begin{tikzpicture}[scale=0.9]
\node at (0.8,3.3) {$
\hspace{0.5em}\overbrace{\hspace{2.7em}}^{\displaystyle V_1}
\overbrace{\hspace{2em}}^{\displaystyle V_2}\hspace{0.8em}
\overbrace{\hspace{1.7em}}^{\displaystyle V_3}
\overbrace{\hspace{2.9em}}^{\displaystyle V_4}\hspace{0.8em}
\overbrace{\hspace{1.6em}}^{\displaystyle V_5}
\cdots
$};
\node at (5.3,0)[rotate=-90] {$
\hspace{0.3em}\overbrace{\hspace{2.3em}}^{}
\overbrace{\hspace{1.5em}}^{}\hspace{0.6em}
\overbrace{\hspace{1.3em}}^{}
\overbrace{\hspace{2.1em}}^{}\hspace{0.6em}
\overbrace{\hspace{1.4em}}^{}
\cdots
$};
\node at (5.7,2.1){$V_1$};
\node at (5.7,1.1){$V_2$};
\node at (5.7,0.3){$V_3$};
\node at (5.7,-0.67){$V_4$};
\node at (5.7,-1.67){$V_5$};
\fill[color=gray!20] (-2.15,0.75) rectangle (0.05,2.65);
\fill[color=gray!20] (0.05,0.75) rectangle (2.35,-1.15);
\shade[shading=axis,bottom color=white,top color=gray!20,shading angle=45] (2.35,-1.15) rectangle (4.15,-2.65);
\node at (0,0) {
$\displaystyle
\sum_{\substack{t=1\\t \text{ even}}}^L M_t(s)=
\begin{pmatrix}
&&& &&&&&&\\
&* &\cdots &* &0&\cdots&0&\cdots&&\quad\\
&\vdots &\ddots &\vdots &\vdots&\ddots&\vdots&\cdots&&\\
&* &\cdots &* &0&\cdots&0&\cdots&&\\
&0&\cdots&0 &*&\cdots &*&\ddots&&\\
&\vdots&\ddots&\vdots &\vdots&\ddots &\vdots&&&\\
&0&\cdots&0 &*&\cdots &*&&&\\
&\vdots&\vdots&\vdots &\ddots&&&*&&\\
&&& &&&&&\ddots&\\
&&& &&&&&&
\end{pmatrix}$
};
\end{tikzpicture}.\\
The same holds for $\sum_{t=1,t \text{ odd}}^L M_t(s)$ with the words ``even'' and ``odd'' exchanged. 

The norm of a block diagonal matrix equals the largest norm of its blocks, hence:
\begin{align*}
\left\Vert \frac{d \mathcal{L}(\mathcal{G})(s)}{d s} \right \Vert
&\le \left\Vert  \sum_{\substack{t=1 \\ t \text{ even}}}^L M_t(s) \right \Vert + \left\Vert  \sum_{\substack{t=1 \\ t \text{ odd}}}^L M_t(s) \right\Vert
\le 2 \max_{1\le t\le L} \Vert M_t(s) \Vert\text{.}
\end{align*}

To bound $M_t(s)$ we carry out a unitary transformation with the unitary
\begin{align}
Q_t(s):&= \sum_{v \in V_t} \ket{v}\bra{v} \otimes -i U_t(s)  +
 \sum_{v \in V\setminus V_t} \ket{v}\bra{v} \otimes \mathbb{I}\label{eq:Qt}\\
Q_t^\dagger(s)&= \sum_{v \in V_t} \ket{v}\bra{v} \otimes i U^\dagger_t(s)  +
 \sum_{v \in V\setminus V_t} \ket{v}\bra{v} \otimes \mathbb{I}\text{.}\nonumber
\end{align}
to achieve a tensor product form:
\begin{alignat*}{3}
Q^\dagger_t(s) M_t(s) Q_t(s) = &\sum_{v\in V_t} \sum_{u\in V_{t-1}} \frac{w(v,u)}{\sqrt{d_v d_u}}
\Big( && \ket{v}\bra{u} \otimes \big(iU_t^\dagger(s) \big) \big(-i U_t(s) h_t\big) \\
 &&+& \ket{u}\bra{v} \otimes \big(i h_t U_t^\dagger(s) \big) \big(- i U_t(s)\big) \Big) \\
=& m_t \otimes h_t
\end{alignat*}
with
\begin{align*}
m_t= \sum_{v\in V_t} \sum_{u\in V_{t-1}} \frac{w(v,u)}{\sqrt{d_v d_u}}
\Big( \ket{v}\bra{u} + \ket{u}\bra{v}\Big)\text{.}
\end{align*}

Since the norm of the hermitean matrix $h_t$ is smaller than $2\pi$, the only interesting task left is to bound $m_t$. We consider the term $m_t$ as part of a normalized Laplacian. We can achieve that by simply adding and subtracting a diagonal term. Define the new weighted graph $G'$ (see also Figure \ref{fig:G'}) on the vertex set $V'=V_{t-1}\cup V_t$ with the weight function
\begin{align*}
w'(v,u) =
\begin{cases}
w(v,u) &\text{ if } \vert \mathcal{T}(v) - \mathcal{T}(u) \vert = 1\\
d_v - \sum\limits_{u\in V_{t-1}} w(v,u) &\text{ if } u=v\in V_t\\
d_u - \sum\limits_{v\in V_t} w(v,u) &\text{ if } u=v \in V_{t-1}\\
0 &\text{ otherwise.}
\end{cases}
\end{align*}

The weights of the loops are defined such that the original degrees of the vertices are preserved, i.e. $d'_v=d_v$ for all $v \in V'$.

\begin{figure}[htb]
\centering
\begin{tikzpicture}[scale=0.9]
\fill[gray!20] (10,2) ellipse (1.7 and 4);
\node at (10,-1.3){${\color{gray}{V_{t-1}}}$};
\fill[gray!20] (17,2) ellipse (1.7 and 3);
\node at (17,-0.3) {${\color{gray}{V_{t}}}$};

\draw[line width=1pt] (10,4.6) circle (0.4);
\draw[line width=1pt] (10,2.6) circle (0.4);
\draw[line width=1pt] (10,0.6) circle (0.4);
\draw[line width=1pt] (17,1.6) circle (0.4);
\draw[line width=1pt] (17,3.6) circle (0.4);

\node (c) at (10,0) [circle,shade,draw,minimum size=8mm] {$c$};
\node (b) at (10,2) [circle,shade,draw,minimum size=8mm] {$b$};
\node (a) at (10,4) [circle,shade,draw,minimum size=8mm] {$a$};
\node (y) at (17,1) [circle,shade,draw,minimum size=8mm] {$y$};
\node (x) at (17,3) [circle,shade,draw,minimum size=8mm] {$x$};

\draw[-, line width=0.5pt] (a) to (x);
\draw[-, line width=0.5pt] (a) to (y);
\draw[-, line width=0.5pt] (b) to (y);
\draw[-, line width=0.5pt] (c) to (x);

\draw[-, dotted,gray,line width=1pt] (a) to (6,2.5);
\draw[-, dotted,gray,line width=1pt] (b) to(6,5);
\draw[-, dotted,gray,line width=1pt] (c) to (6,1);
\draw[-, dotted,gray,line width=1pt] (b) to (6,1.5);
\draw[-, dotted,gray,line width=1pt] (x) to (21,4.5);
\draw[-, dotted,gray,line width=1pt] (x) to (21,2);
\draw[-, dotted,gray,line width=1pt] (y) to (21,2.5);
\draw[-, dotted,gray,line width=1pt] (y) to (21,0);

\node at (13.8,3.8) [rotate=-8.1] {$w(a,x)$};
\node at (12.9,1.9) [rotate=-8.1] {$w(b,y)$};
\node at (12.9,0.8) [rotate=23.2] {$w(c,x)$};
\node at (13.7,2.75) [rotate=-23.2] {$w(a,y)$};

\node at (10,5.2) {$\mathbf{a_1}$};
\node at (10,3.22) {$\mathbf{b_1+b_2}$};
\node at (10,1.2) {$\mathbf{c_1}$};
\node at (17,4.2) {$\mathbf{x_1+x_2}$};
\node at (17,2.22) {$\mathbf{y_1+y_2}$};

\node at (7,3.1) [rotate=20.6] {${\color{gray}{a_1}}$};
\node at (7,4.6) [rotate=-36.9] {${\color{gray}{b_1}}$};
\node at (7,1.9) [rotate=7.1] {${\color{gray}{b_2}}$};
\node at (7,1) [rotate=-14] {${\color{gray}{c_1}}$};
\node at (20,4.4) [rotate=20.6] {${\color{gray}{x_1}}$};
\node at (19.5,2.65) [rotate=-14] {${\color{gray}{x_2}}$};
\node at (19.5,1.6) [rotate=20.6] {${\color{gray}{y_1}}$};
\node at (20,0.53) [rotate=-14] {${\color{gray}{y_2}}$};
\end{tikzpicture}
\caption{Example contruction of the graph $G'$: The orginal gray dotted edges from $G$ are missing, therefore loops are added to the vertices with the corresponding weights.}
\label{fig:G'}
\end{figure}
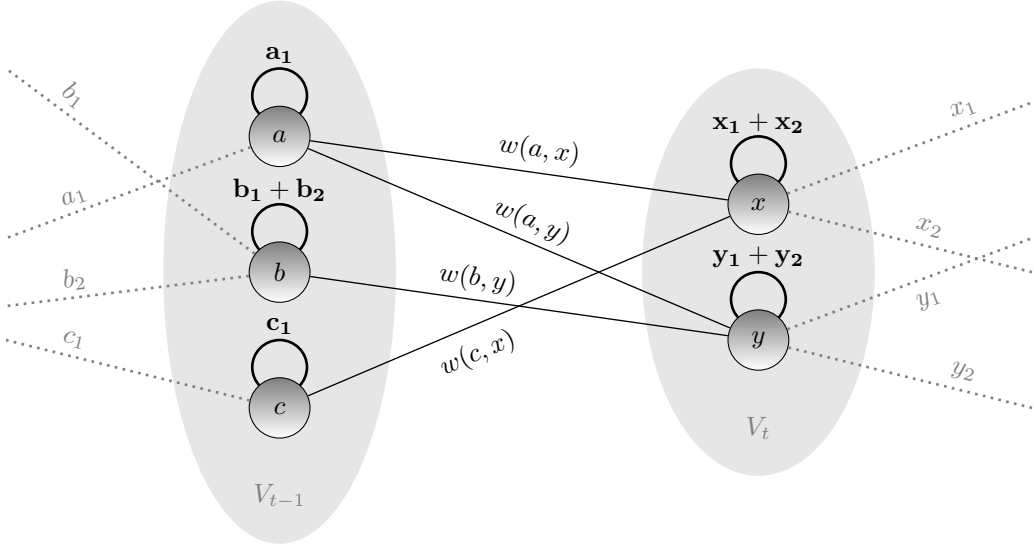

The normalized Laplacian of the graph $G'$ is then
\begin{align*}
\mathcal{L}(G') = m_t +  \sum_{v\in V'} \left(1-\frac{w'(v,v)}{d'_v}\right) \ket{v}\bra{v}\text{.}
\end{align*}
As the norm of a normalized Laplacian is bounded by $2$ and the norm of a diagonal matrix by its largest absolute diagonal element we can bound $m_t$ by
\begin{align*}
\Vert m_t \Vert \le \Vert \mathcal{L}(G') \Vert + \left\Vert \sum_{v\in V'} \left(1-\frac{w'(v,v)}{d'_v}\right) \ket{v}\bra{v} \right\Vert
\le 3\text{.}
\end{align*}

Consequently, the whole derivative $\frac{d \mathcal{L}(\mathcal{G})(s)}{d s}$ is bounded by a constant:
\begin{align*}
\left\Vert \frac{d \mathcal{L}(\mathcal{G})(s)}{d s} \right \Vert
&\le 2 \max_{1\le t\le L} \Vert M_t(s) \Vert\\
&= 2 \max_{1\le t\le L} \Vert m_t \Vert \cdot\Vert h_t \Vert\\
&\le 12\pi\text{.}
\end{align*}

For the second derivative of the normalized Laplacian the proof works analogously. The derivative is still the sum of the even and odd $M_t(s)$, we just have to change slightly the definition of the $M_t(s)$ in equation \eqref{eq:Mt}:
\begin{align*}
M_t(s) &=  \sum_{v\in V_t} \sum_{u\in V_{t-1}}
\frac{w(v,u)}{\sqrt{d_v d_u}} \left(\ket{v}\bra{u} \otimes - U_t(s)h_t^2 + \ket{u}\bra{v} \otimes - h_t^2 U_t^\dagger(s)\right)\text{.}
\end{align*}
If we alter also the definition of the unitary $Q_t(s)$ in equation \eqref{eq:Qt} to
\begin{align*}
Q_t(s):&= \sum_{v \in V_t} \ket{v}\bra{v} \otimes U_t(s)  +
 \sum_{v \in V\setminus V_t} \ket{v}\bra{v} \otimes \mathbb{I}
\end{align*}
we get
\begin{align*}
Q^\dagger_t(s) M_t(s) Q_t(s) &= m_t \otimes -h_t^2
\end{align*}
with the original $m_t$. Hence the rest of the proof is exactly the same until the end when we obtain for the norm of the second derivative a slightly different but still constant term:
\begin{align*}
\left\Vert \frac{d \mathcal{L}(\mathcal{G})(s)}{d s} \right \Vert
&\le 2 \max_{1\le t\le L} \Vert m_t \Vert \cdot\Vert h_t^2 \Vert\\
&\le 24\pi^2\text{.}\qedhere
\end{align*}
\end{proof}

\section{Limitations on graphs for an efficient adiabatic quantum computation}\label{sec:limitations}

In this section we will derive some consequences from vertex expansion for (unweighted) graphs with an attractive large gap for adiabatic quantum computation. The basic requirement for a parallel transport network implementing a quantum circuit concerns its diameter:
\begin{lem}\label{lem:diameterNeed}
The underlying graph of a parallel transport network implementing a quantum circuit with $L$ gates has at least diameter $L$.
\end{lem}
\begin{proof}
This follows directly from the fact that only vertices whose assigned time steps differ maximally by $1$ can be connected (review the Definition \ref{def:PTN} of a parallel transport network) and that the connected graph contains vertices belonging to $L+1$ different time steps (review Definition \ref{def:implement} about circuit implementation).
\end{proof}

\begin{thm}\label{thm:linearDiam}
Let $\mathcal{G}=(G,\mathcal{T},\mathcal{U},\pi)$  be a parallel transport network simulating a quantum circuit with length $L$ and $G=(V,E)$ a graph with a constant degree ratio $\frac{d_{max}}{d_{min}}\in \mathcal{O}(1)$. If
\begin{align*}
\lambda(\mathcal{G}) \in \Omega\left(\frac{1}{L^k}\right)
\end{align*}
with $k<2$, then
\begin{align*}
 \vert V\vert \notin \mathcal{O}\left(\poly(L)\right)\text{.}
\end{align*}
\end{thm}
\begin{proof}
If $\mathcal{G}$ simulates a quantum circuit with length $L$, it implements a quantum circuit with length $L'\ge L$. We remind ourselves of the relationship between the diameter of a graph and its spectral gap given by Corollary \ref{cor:diameter} in the first chapter:
\begin{align*}
\sqrt{\frac{d_{max}}{d_{min}}} \frac{ \log_2 \vert V\vert}{\sqrt{\lambda}}\in \Theta(\diam(G))\text{.}
\end{align*}
Now let $\lambda \in \Omega\left(\frac{1}{L^k}\right)\subseteq \Omega\left(\frac{1}{L'^k}\right)$, $k<2$. With the knowledge of the previous lemma that $G$ has at least diameter $L'$ it follows
\begin{align*}
 L'^{\frac{k}{2}} \log_2 \vert V\vert &\in \Omega(L')\\
\vert V \vert &\in \Omega\left(2^{L'^{1-\frac{k}{2}}}\right)\subseteq  \Omega\left(2^{L^{1-\frac{k}{2}}}\right)\text{.}
\end{align*}
As $2^{L^\epsilon}$ with a constant $\epsilon>0$ always exceeds any polynomial in $L$ for large enough $L$, $\vert V\vert$ cannot be upper bounded by a polynomial in $L$.
\end{proof}

The previous theorem tells us that the vertex set size of a graph family with constant degree ratio and desirable large spectral gap $\Omega\left(\frac{1}{L^k}\right)$, $k<2$, scales faster than polynomially in the circuit length. This does not neccessarily contradict the efficiency requirement that the vertices can be represented by polynomially many qubits, since an $m$-qubit Hilbert space offers $2^m$ distinct basis vectors. But it might be more difficult to write the normalized Laplacian of the parallel transport network in local terms.

Furthermore a vertex set scaling faster than any polynomial may cause a problem for expressions in that the volume of the vertex set appears in the denominator like for the gap angle in Proposition \ref{prop:angleExpl} or the output probability in Proposition \ref{prop:outputProb}. If the volumes of the initial and final vertices do not scale comparibly fast enough, the gap angle and the output probability decrease faster than polynomial and hence the construction is not capable of an efficient adiabatic quantum computation. We will see an example of such an behaviour in the construction of the hypercube Hamiltonian in Chapter \ref{chap:hypercube}.

The next theorem will show that for graphs with even stricter gap $\lambda \in \Omega\left(\frac{1}{L^k}\right)$, $k<1$, the property of vertex expansion forbids a large volume of the initial and the final vertices at the same time and hence directly implies that either the gap angle or the output probability is too small for an efficient adiabatic quantum computation. This is a quite strong negative result: It excludes for example all standard graph Hamiltonians based on expander graphs from an efficient adiabatic quantum computation and the energy gap $\gamma \in \Theta\left(\frac{1}{L}\right)$ seems to be the natural lower bound for the efficiency that can be reached with standard constructions (but notice that there are indeed functions with growing behaviour truely between $L$ and $L^k$, $k<1$, for example $\frac{L}{\log L}$). On the other hand we have just discussed that even reaching $\gamma \in \Theta\left(\frac{1}{L}\right)$ might be difficult since the implied nonpolynomial vertex set might cause problems for local implementation and for the sufficiently high volume fractions needed for the gap angle and the output probability.

\begin{thm}\label{thm:tradeoff}
Let $H(s)$ be a standard graph Hamiltonian according to Definition \ref{def:standard} simulating a quantum circuit with length $L$. Let $G=(V,E)$ be its  underlying graph with spectral gap $\lambda \in \Omega\left(\frac{1}{L^k}\right)$, $k <1$. Denote by $\theta$ the gap angle and by $p$ the probability to obatin a correct circuit output after measuring on the graph subsystem of any final ground-state in the extended vertex basis and discarding the subsystem. Then it holds 
\begin{align*}
p  \notin \Omega\left(\frac{1}{\poly(L)}\right) \quad\quad \text{or} \quad\quad
\sin^2\left(\theta\right) \notin \Omega\left(\frac{1}{\poly(L)}\right)\text{.}
\end{align*}
\end{thm}
\begin{proof}
Since the set of initial vertices $V_I:= \bigcup_{t=0}^{L_i} V_t$ and the set of final vertices $V_F := \bigcup_{t=L'-L_f}^{L'} V_t$ fulfill $\dist\left(V_I,V_F\right)\ge L$, it follows for large enough $L$ that $\dist_{\frac{L}{4}+1}(V_{I}) \cap \dist_{\frac{L}{4}+1}(V_{F}) = \emptyset$. Let $X$ be the set $V_I$ or $V_F$ that fulfills $\vol\left(\dist_{\frac{L}{4}+1}(X)\right) \le \frac{1}{2}\vol(V)$. According to Lemma \ref{lem:expVertexExp} it holds
\begin{align*}
\frac{\vol(X)}{\vol(V)} &\le \frac{1}{2} \left( 1+ \frac{\lambda}{2}\right)^{-\frac{L}{4}}\\
\frac{\vol(X)}{\vol(V)} &\in \mathcal{O}\left(e^{-\frac{1}{8} L^{1-k}} \right)
\end{align*}
The fraction $\frac{\vol(X)}{\vol(V)}$ equals according to Propositions \ref{prop:angleExpl} and \ref{prop:outputProb} either $\sin^2(\theta)$ or the probability $p$, depending on whether $X=V_I$ or $X=V_F$ and cannot be lower bounded by a polynomial since $e^{-\frac{1}{8}L^{1-k}}$ decreases faster than any polynomial for $k<1$.
\end{proof}

\chapter{Adiabatic Quantum Computation with the Kitaev Hamiltonian}\label{chap:path}
\ihead{CHAPTER 6.\,\, AQC WITH THE KITAEV HAMILTONIAN}

In this chapter we present the so-called Feynman or Kitaev Hamiltonian which is the standard graph Hamiltonian that  is mainly used for adiabatic quantum computation at the current state of research. Originally this Hamiltonian gained significance as center piece in the QMA-completeness proof of the local Hamiltonian problem (see \cite{npsurvey}). As a Hamiltonian for adiabatic quantum computation it offers an efficient running time thanks to the energy gap of $\Omega\left(\frac{1}{L^2}\right)$ which originates from the spectral gap of the underlying regular path graph.

\begin{definition}
The $2$-regular graph $G=(V,E)$ with $V=\{0, \dots L\}$ and adjacency matrix
\begin{align*}
A = \ket{0}\bra{0} + \ket{L}\bra{L} + \sum_{t=1}^{L} \ket{t}\bra{t-1} + \ket{t-1}\bra{t}
\end{align*}
is called the \textsf{regular path graph} of length $L+1$.
\end{definition}

For a reference of the spectral gap see for example \cite[Theorem 4]{yueh}:
\begin{prop}
The normalized Laplacian spectrum of the regular path graph of length $L+1$ is
\begin{align*}
\lambda_k = 1 - \cos\left(\frac{\pi k}{L+1}\right), \quad\quad 0\le k\le L
\end{align*}
and hence the spectral gap obeys
\begin{align*}
\lambda \in \Theta\left(\frac{1}{L^2}\right)\text{.}
\end{align*}
\end{prop}

The Kitaev Hamiltonian $H(s)=H_{prop}(s) + H_{in} + H_{graph}$ is a standard graph Hamiltonian with underlying regular path graph of length $L'+1=L_i+L+L_f+1$ simulating a quantum circuit $C=\big(n, (U_1, \dots U_L)\big)$. The Hamiltonian has an $L'$-qubit extended graph space and is defined by
\begin{align*}
H_{prop}(s):=&  \frac{1}{2}\Big( \overline{\ket{0}\bra{0}}\otimes \mathbb{I} + \overline{\ket{L'}\bra{L'}}\otimes \mathbb{I} \Big) + \sum_{t=1}^{L'-1} \overline{\ket{t}\bra{t}} \otimes \mathbb{I}\\
&- \frac{1}{2} \sum_{t=1}^{L'} \Big( \overline{\ket{t}\bra{t-1}} \otimes U_t(s) + \overline{\ket{t-1}\bra{t}} \otimes U_t^\dagger(s)\Big)\\
H_{in} :=& \frac{1}{n} \sum_{t=0}^{L_i}\sum_{i=1}^n \overline{\ket{t}\bra{t}} \otimes \ket{1}\bra{1}_i\\
H_{graph}:=& \frac{1}{L'} \sum_{i = 2}^{L'} \ket{01}\bra{01}_{i-1,i} \otimes \mathbb{I}\text{.}
\end{align*}
A lower index of an operator indicates the qubit(s) on which the operator acts non-trivially. The notation of overlined operators will also be used in the following chapter and is defined as
\begin{definition}\label{def:overlined}
Given a $L'$-qubit space we write
\begin{align*}
\overline{\ket{t}\bra{t}}:=& \ket{110}\bra{110}_{t-1, t, t+1}\\
\overline{\ket{t-1}\bra{t-1}}:=& \ket{100}\bra{100}_{t-1, t, t+1}\\
\overline{\ket{t}\bra{t-1}}:=& \ket{110}\bra{100}_{t-1, t, t+1}\\
\overline{\ket{t-1}\bra{t}}:=& \ket{100}\bra{110}_{t-1, t, t+1}
\end{align*}
for $1 < t < L'$ and
\begin{align*}
\overline{\ket{1}\bra{1}}:=& \ket{10}\bra{10}_{1,2}& \hspace{1em}
\overline{\ket{L'}\bra{L'}}:=& \ket{11}\bra{11}_{L'-1, L'}\\
\overline{\ket{0}\bra{0}}:=& \ket{00}\bra{00}_{1,2}& 
\overline{\ket{L'-1}\bra{L'-1}}:=& \ket{10}\bra{10}_{L'-1, L'}\\
\overline{\ket{1}\bra{0}}:=& \ket{10}\bra{00}_{1,2}& 
\overline{\ket{L'}\bra{L'-1}}:=& \ket{11}\bra{10}_{L'-1, L'}\\
\overline{\ket{0}\bra{1}}:=& \ket{00}\bra{10}_{1,2}&
\overline{\ket{L'-1}\bra{L'}}:=& \ket{10}\bra{11}_{L'-1, L'}\text{.}
\end{align*}
\end{definition}

We use the notation of the overlined operators since they reduce to the operators between the corresponding vertex states if we restrict the Hamiltonian to the subspace of proper network states
\begin{align*}
D:= \big\{\ket{t} \otimes \ket{x} \,\vert\, 0 \le t \le L',\, x \in \{0,1\}^n\big\}
\end{align*}
with vertex states
\begin{align*}
\ket{t} := &\ket{1\dots \mathrel{\mathop{1}_{\substack{\uparrow\\ \mathclap{\text{t-th qubit}}}}} 0 \dots 0}\text{.}
\end{align*}
Explicitly the Hamiltonian $H(s)$ restricted to $D$ becomes
\begin{align*}
H(s)_{\vert D}:=& H_{prop \vert D}(s) + H_{in\, \vert D}\\
H_{prop\vert D}(s):=& \frac{1}{2}\Big( \ket{0}\bra{0}\otimes \mathbb{I} + \ket{L'}\bra{L'}\otimes \mathbb{I} \Big) + \sum_{t=1}^{L'-1} \ket{t}\bra{t} \otimes \mathbb{I}\\
&- \frac{1}{2} \sum_{t=1}^{L'} \Big( \ket{t}\bra{t-1} \otimes U_t(s) + \ket{t-1}\bra{t} \otimes U_t^\dagger(s)\Big)\\
H_{in\vert D}:=& \frac{1}{n} \sum_{t=0}^{L_i} \sum_{t=1}^n \ket{t}\bra{t} \otimes \ket{1}\bra{1}_i\text{.}
\end{align*}

We can now see that $H(s)$ indeed fulfills Definition \ref{def:standard} of a standard graph Hamiltonian: It leaves the subspace $D$ and the orthogonal subspace $D^\perp$ invariant, $H_{prop\vert D}(s)$ equals the normalized Laplacian of our desired parallel transport network with $\mathcal{T}(\ket{t})=t$, the penalty term $H_{in}$ is positive semi-definite, the null-space of $H_{in\vert D}$ is spanned by
\begin{align*}
\big\{ \ket{t} \otimes \ket{\mathbf{0}} \,\big\vert\, 0\le t \le L_i \big\}
\cup \big\{ \ket{t} \otimes \ket{x} \,\big\vert\, t>L_i,\, x\in \{0,1\}^n\big\}
\end{align*}
and $H_{graph}$ vanishes on $D$ and has an expectation value $\ge \frac{1}{L'}$ for all states from $D^\perp$.
\begin{figure}[bht]
\centering
\begin{tikzpicture}[scale=0.8]
\draw[line width=1pt] (-0.8,0) circle (0.4);
\draw[line width=1pt] (15.8,0) circle (0.4);
\node (0) at (0,0) [circle,shade,draw,minimum size=10mm] {$0$};
\node (1) at (3,0) [circle,shade,draw,minimum size=10mm] {$1$};
\node (2) at (6,0) [circle,shade,draw,minimum size=10mm] {$2$};
\node (P) at (9,0) [minimum size=10mm] {$\cdots$};
\node (L1) at (12,0) [circle,shade,draw,minimum size=10mm] {${\scriptstyle L'-1}$};
\node (L) at (15,0) [circle,shade,draw,minimum size=10mm] {$L'$};
\draw[->, line width=1pt] (0) to[out=10, in=170] (1);
\draw[->, line width=1pt] (1) to[out=190, in=350] (0);
\draw[->, line width=1pt] (1) to[out=10, in=170] (2);
\draw[->, line width=1pt] (2) to[out=190, in=350] (1);
\draw[->, line width=1pt] (2) to[out=10, in=170] (P);
\draw[->, line width=1pt] (P) to[out=190, in=350] (2);
\draw[->, line width=1pt] (P) to[out=10, in=170] (L1);
\draw[->, line width=1pt] (L1) to[out=190, in=350] (P);
\draw[->, line width=1pt] (L1) to[out=10, in=170] (L);
\draw[->, line width=1pt] (L) to[out=190, in=350] (L1);
\node at (1.5,0.6) {$U_1$};
\node at (1.5,-0.6) {$U_1^\dagger$};
\node at (4.5,0.6) {$U_2$};
\node at (4.5,-0.6) {$U_2^\dagger$};
\node at (13.5,0.6) {$U_{L'}$};
\node at (13.5,-0.6) {$U_{L'}^\dagger$};
\node at (-1.5,0) {$\mathbb{I}$};
\node at (16.5,0) {$\mathbb{I}$};
\end{tikzpicture}
\caption{The parallel transport network of the Kitaev Hamiltonian}
\end{figure}
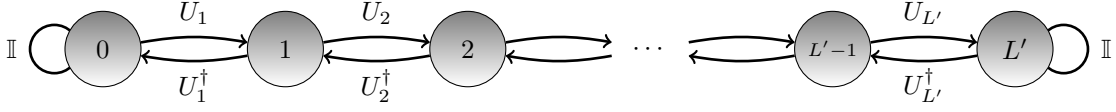

$H_{in}$ is defined such that $\ket{\mathbf{0}}$ is the only valid circuit input. We can hence use the Kitaev Hamiltonian for an adiabatic quantum computation with unique ground-state
\begin{align*}
\ket{\eta(s)}:=& \frac{1}{\sqrt{L'+1}} \sum_{t=0}^{L'} \ket{t} \otimes U_t(s) U_{t-1}(s) \dots U_1(s) \ket{\mathbf{0}}\text{.}
\end{align*}

One can easily check that the efficiency requirements 1--5 from the list in Section \ref{sec:circuitSimReq} are fulfilled. It remains to determine the gap of the Hamiltonian and the probability to measure the correct computation output at the end. Since the gap $\frac{1}{L'}$ of $H_{graph}$ and the gap $\frac{1}{n}$ of $H_{in\vert D}$ are larger than the spectral gap $\lambda\in \theta\left(\frac{1}{L'^2}\right)$ of $H_{prop\vert D}(s)$, the gap $\gamma$ of the Hamiltonian is determined according to Lemmata \ref{lem:gaps} and \ref{lem:angle} by
\begin{align*}
\gamma \in \Omega\left( \lambda \sin^2\left(\frac{\theta}{2}\right) \right)
\end{align*}
with $\theta$ denoting the gap angle.

According to Proposition \ref{prop:angleExpl} the gap angle is given by
\begin{align*}
\sin^2(\theta)&= \frac{\sum_{t=0}^{L_i}\vol(V_t)}{\vol(V)}= \frac{L_i+1}{L'+1}\\
\sin^2\left(\frac{\theta}{2}\right) &= \frac{1}{4} \frac{\sin^2\theta}{\cos^2\frac{\theta}{2}}
\ge\frac{1}{4}\frac{L_i+1}{L'+1}
\end{align*}
and according to Proposition \ref{prop:outputProb} we obtain the correct computation output $U_{L'} \dots U_1 \ket{\mathbf{0}}$ after measuring on the extended graph subsystem of the final ground-state $\ket{\eta(1)}$ in the extended vertex basis and discarding the subsystem with probability
\begin{align*}
p &= \frac{\sum\limits_{t=L'-L_f}^{L'} \vol(V_t)}{\vol(V)} = \frac{L_f+1}{L'+1}\text{.}
\end{align*}

If we had not added any initial and final identity gates to the circuit, then the gap angle would be lower bounded by $\Omega\left(\frac{1}{L}\right)$ and we would obtain
\begin{align*}
\gamma \in \Omega\left(\frac{1}{L^3}\right) \quad\quad\text{ and }\quad\quad p \in \Theta\left(\frac{1}{L}\right)\text{.}
\end{align*}
At the end of Section \ref{sec:circuitSimReq} we noted that in the literature adiabatic quantum computation is usually presented without final measurement and consequently forbids the addition of identity gates to improve the gap angle. Hence the scaling behaviour $\Omega\left(\frac{1}{L^3}\right)$ is usually given as result for the gap of the Kitaev Hamiltonian. But we also stated that our concept with a final measurement and the tool of adding identity gates makes more sense for comparisions with the performance of other standard graph Hamiltonians. Hence if we choose $L_i=L$ initial and $L_f=L$ final identity gates, the gap angle is lower bounded by a constant and we can improve our bounds to
\begin{align*}
\gamma \in \Omega\left(\frac{1}{L^2}\right) \quad\quad\text{ and }\quad\quad p \in \Theta\left(1\right)\text{.}
\end{align*}

In the next chapters we will try to improve the gap by looking at standard graph Hamiltonians with different underlying graphs. Unfortunately we will not succeed in lifting the gap above $\Omega\left(\frac{1}{L^2}\right)$. So it may be the case that the Kitaev contruction, based on the most simple and natural path graph, already offers the optimal adiabatic quantum computation by a standard graph Hamiltonian.

\chapter{Adiabatic Quantum Computation with a Hypercube Hamiltonian}\label{chap:hypercube}
\ihead{CHAPTER 7.\,\, AQC WITH A HYPERCUBE HAMILTONIAN}

\section{Motivation and definition of the Hamiltonian}\label{sec:hypercube}

The Kitaev Hamiltonian of the previous chapter offered us an effective adiabatic quantum computation with an energy gap that equaled the spectral gap $\lambda \in \Omega\left(\frac{1}{L^2}\right)$ of the underlying path graph. We already know from Theorem \ref{thm:linearDiam} that a standard graph Hamiltonian construction based on a graph with constant degree ratio and a better gap $\lambda \in \Omega\left(\frac{1}{L^k}\right)$, $k<2$, implies that the vertex set of the graph grows faster than polynomial. On the other hand we are required to write the Hamiltonian in local terms which often turns out to be very difficult for nonpolynomial sized graphs. 

It is probably the simplest solution to investigate graphs whose vertices are identified by bitstrings and have edges only between vertices that maximally differ by a constant number of bits. These graphs might offer a good chance for local implementation. This idea leads us to Cayley graphs $\cay(H,S)$ based on the group $H=\big(\{0,1\}^N, \oplus\big)$ of $N$-bitstrings with bitwise addition $\oplus$ and the symmetric set $S\subseteq\{s \,\vert\, s\in \{0,1\}^{N},\, \vert s \vert_1 \le k\}$ with $k$ a locality constant and $\vert s\vert_1$ denoting the Hamming weight of the bit string $s$. As a further advantage we already know from Theorem \ref{thm:cayleyBitwise} an easy expression for the eigenvalues of such graphs:
\begin{align}
\lambda_x = 1 - \frac{1}{\vert S\vert} \sum_{s\in S} (-1)^{\langle x, s\rangle}\text{,} \quad\quad x\in \{0,1\}^k\text{.}\label{eq:eigCay}
\end{align}

The most simple graph of this kind is certainly the hypercube whose vertices are adjacent iff they differ in exactly one bit:
\begin{definition}\label{def:hypercube}
The Cayley graph $\cay(H,S)$ with
\begin{align*}
H:=&\big(\{0,1\}^{L},\oplus\big)\\
S:=&\{e_i \,\vert\, 1\le i \le L\}\\
e_i:=& 0\dots 0  \mathrel{\mathop{1}_{\substack{\uparrow\\ \mathclap{\text{i-th bit}}}}} 0 \dots 0
\end{align*}
is called the \textsf{hypercube graph} of degree $L$.
\end{definition}

\begin{lem}\label{lem:hypercubeGap}
For the hypercube graph of degree $L$ the spectral gap obeys
\begin{align*}
\lambda = \frac{2}{L}\text{.}
\end{align*}
\end{lem}
\begin{proof}
The all zero bitstring $x=0^N$ leads in equation \eqref{eq:eigCay} to the nondegenerated eigenvalue $0$, the next lowest eigenvalue i.e. the spectral gap is consequently given by
\begin{align*}
\lambda&=\min_{\substack{x\in \{0,1\}^{N}\\x \ne 0^N}} \left(1-\frac{1}{\vert S \vert} \sum_{s\in S} (-1)^{\langle x, s\rangle} \right)
=\frac{2}{L} \min_{\substack{x\in \{0,1\}^{N}\\x \ne 0^N}} \left\vert \big\{ s\in S \,\big\vert\, {\langle x, s\rangle}=1\big\}\right\vert
\stackrel{x=e_i}{=} \frac{2}{L}\text{.}\qedhere
\end{align*}
\end{proof}

A hypercube of degree $L$ has obviously diameter $L$ and can hence simulate a quantum circuit of that length. Since we proved in Theorem \ref{thm:tradeoff} that no standard graph Hamiltonian based on a graph with spectral gap $\lambda \in \Omega\left(\frac{1}{L^k}\right)$, $k<1$, can be used for an efficient adiabatic quantum computation, the gap $\lambda\in \Theta\left(\frac{1}{L}\right)$ of the hypercube might ultimately be the optimal possible improvement. So at first glance the hypercube looks like a very promising candidate for an efficient simulation of a quantum circuit $C$ via adiabatic quantum computation. Our first goal is hence to find a local formulation of the normalized Laplacian of a parallel transport network based on the hypercube graph that implements $C'(s)=\big(n, (U_1(s), \dots U_{L'}(s)) \big)$, the corresponding time-dependent circut to an identitiy extension of $C$.

It is very natural to locally implement the normalized Laplacian of a hypercube \textit{graph} on a $L'$-qubit space by identifying each computational basis state $\ket{z}$ with the vertex $z$ according to the above Definition \ref{def:hypercube}. Each interaction term representing an edge is then local since it exactly inverts one qubit. The local implementation gets unfortunately a little more complicated for the normalized Laplacian of the \textit{parallel transport network} with underlying hypercube graph, since the interaction terms have to include the unitary belonging to this edge. Our idea is to consider the Hamming weight $\vert z\vert_1$ of a vertex $z$ as corresponding time step $t$. Notice that this is possible because the Hamming weight of adjacent vertices differs exactly by $1$. Unfortunately an interaction term that inverts a certain qubit affects every vertex $z$, independently of its Hamming weight and hence of its time step. That's why we will identify each vertex $z$ in Definition \ref{def:hypercube} by the computational basis vector $\ket{z} \otimes \ket{t}$ with $t$ the Hamming weight of $z$. The extended graph space of our standard graph Hamiltonian $H(s)$ according to Definition \ref{def:standard} is hence already a tensor product of two $L'$-qubit spaces, the first is supposed to carry the label of a vertex and the second the information about its time step (its Hamming weight). For the latter one we use again an encoding with the overlined operators from Definition \ref{def:overlined}:
\begin{align*}
H(s):=&H_{prop}(s) + H_{in} + H_{graph}\\
H_{prop}(s):=& \mathbb{I}\otimes  \mathbb{I}\otimes \mathbb{I}
- \frac{1}{L'}\sum_{j=1}^{L'} \sum_{t=1}^{L'} 
\Big(\ket{1}\bra{0}_j \otimes \overline{\ket{t}\bra{t-1}} \otimes U_t(s)
+\ket{0}\bra{1}_j \otimes \overline{\ket{t-1}\bra{t}} \otimes U_t^\dagger(s) \Big)\\
H_{in}:=& \frac{1}{n}\mathbb{I} \otimes \sum_{t=0}^{L_i}\overline{\ket{t}\bra{t}} \otimes \sum_{i=1}^n \ket{1}\bra{1}_i\\
H_{graph}:=& H_{label}+H_{weight}\\
H_{label}:=& \frac{1}{L'}\sum_{i=1}^{L'} \ket{1}\bra{1}_i  \otimes  \overline{\ket{0}\bra{0}} \otimes \mathbb{I}\\
H_{weight}:=& \frac{1}{L'}\mathbb{I} \otimes \sum_{i=2}^{L'} \ket{01}\bra{01}_{i-1,i} \otimes \mathbb{I}\text{.}
\end{align*}

We can recognize the structure of the parallel transport network better, if we again restrict the Hamiltonian to the subspace of proper network states
\begin{align*}
D:=\big\{ \ket{z} \otimes \ket{t} \otimes \ket{x} \quad\big\vert\, 0 \le t \le L',\, z\in\{0,1\}^{L'},\, \vert z\vert_1 = t,\, x \in \{0,1\}^n \big\}
\end{align*}
with
\begin{align*}
\ket{t} := &\ket{1\dots \mathrel{\mathop{1}_{\substack{\uparrow\\ \mathclap{\text{t-th qubit}}}}} 0 \dots 0}\text{.}
\end{align*}

Our Hamiltonian $H(s)$ restricted to this subspace becomes
\begin{align*}
H_{\vert D}(s):=&H_{prop\vert D}(s) +  H_{in\vert D}\\
H_{prop\vert D}(s):=& \sum_{t=0}^{L'} \sum_{\substack{z\in \{0,1\}^{L'} \\ \vert z \vert_1 = t}} \ket{z}\bra{z} \otimes \ket{t}\bra{t} \otimes \mathbb{I}\\
&-\frac{1}{L'}\sum_{t=1}^{L'} \sum_{\substack{z\in\{0,1\}^{L'}\\ \vert z \vert_1 = t}} \sum_{\substack{j\\ \vert z \vert_j = 1}}\Bigg(\ket{z}\bra{z \oplus e_j} \otimes \ket{t}\bra{t-1} \otimes U_t(s)\\
&\hspace{11.3em}+ \ket{z \oplus e_j}\bra{z} \otimes \ket{t-1}\bra{t} \otimes U_t^\dagger(s) \Bigg)\\
H_{in\vert D}:=& \frac{1}{n}\mathbb{I} \otimes  \sum_{t=0}^{L_i}\ket{t}\bra{t} \otimes \sum_{i=1}^n \ket{1}\bra{1}_i\text{.}
\end{align*}

Now we can see that $H_{prop\vert D}(s)$ is indeed the normalized Laplacian of a parallel transport network $\mathcal{G}=(G,\mathcal{T},\mathcal{U},\pi)$ implementing the circuit $C(s)=\big(n, (U_1(s), \dots U_{L'}(s))\big)$ whose underlying graph $G$ is a hypercube of degree $L'$ defined on the vertex set $\big\{ \ket{z} \otimes \ket{t} \,\big\vert\, 0 \le t \le L',\, z\in\{0,1\}^{L'},\, \vert z\vert_1 = t\big\}$ with $\mathcal{T}(\ket{z} \otimes \ket{t})=t$.

\begin{figure}[hb]
\centering
\begin{tikzpicture}[scale=0.8]
\node (0000) at (0,0) [circle,shade,draw,minimum size=8mm] {$0000$};
\node (0001) at (4,3) [circle,shade,draw,minimum size=8mm] {$0001$};
\node (0010) at (4,1) [circle,shade,draw,minimum size=8mm] {$0010$};
\node (0100) at (4,-1) [circle,shade,draw,minimum size=8mm] {$0100$};
\node (1000) at (4,-3) [circle,shade,draw,minimum size=8mm] {$1000$};
\node (0011) at (8,5) [circle,shade,draw,minimum size=8mm] {$0011$};
\node (0101) at (8,3) [circle,shade,draw,minimum size=8mm] {$0101$};
\node (0110) at (8,1) [circle,shade,draw,minimum size=8mm] {$0110$};
\node (1001) at (8,-1) [circle,shade,draw,minimum size=8mm] {$1001$};
\node (1010) at (8,-3) [circle,shade,draw,minimum size=8mm] {$1010$};
\node (1100) at (8,-5) [circle,shade,draw,minimum size=8mm] {$1100$};
\node (0111) at (12,3) [circle,shade,draw,minimum size=8mm] {$0111$};
\node (1011) at (12,1) [circle,shade,draw,minimum size=8mm] {$1011$};
\node (1101) at (12,-1) [circle,shade,draw,minimum size=8mm] {$1101$};
\node (1110) at (12,-3) [circle,shade,draw,minimum size=8mm] {$1110$};
\node (1111) at (16,0) [circle,shade,draw,minimum size=8mm] {$1111$};

\node at (1.9,1.9) {$U_1$};
\node at (5.9,4.35) {$U_2$};
\node at (10.1,4.35) {$U_3$};
\node at (14.1,1.9) {$U_4$};

\node at (0,-1.1){${\color{gray}{\ket{0000}\otimes \ket{0}}}$};
\node at (4,-4.1){${\color{gray}{\ket{1000}\otimes \ket{1}}}$};
\node at (6.2,-5.3){${\color{gray}{\ket{1100}\otimes \ket{2}}}$};
\node at (12,-4.1){${\color{gray}{\ket{1110}\otimes \ket{3}}}$};
\node at (16.1,-1.1){${\color{gray}{\ket{1111}\otimes \ket{4}}}$};

\node at (4,4){${\color{gray}{\ket{0001}\otimes \ket{1}}}$};
\node at (6.2,5.3){${\color{gray}{\ket{0011}\otimes \ket{2}}}$};
\node at (12,4.05){${\color{gray}{\ket{0111}\otimes \ket{3}}}$};

\draw[->, line width=1pt] (0000) to (0001);
\draw[->, line width=1pt] (0000) to (0010);
\draw[->, line width=1pt] (0000) to (0100);
\draw[->, line width=1pt] (0000) to (1000);

\draw[->, line width=1pt] (0001) to (0011);
\draw[->, line width=1pt] (0001) to (0101);
\draw[->, line width=1pt] (0001) to (1001);

\draw[->, line width=1pt] (0010) to (0011);
\draw[->, line width=1pt] (0010) to (0110);
\draw[->, line width=1pt] (0010) to (1010);

\draw[->, line width=1pt] (0100) to (0101);
\draw[->, line width=1pt] (0100) to (0110);
\draw[->, line width=1pt] (0100) to (1100);

\draw[->, line width=1pt] (1000) to (1001);
\draw[->, line width=1pt] (1000) to (1010);
\draw[->, line width=1pt] (1000) to (1100);

\draw[->, line width=1pt] (0111) to (1111);
\draw[->, line width=1pt] (1011) to (1111);
\draw[->, line width=1pt] (1101) to (1111);
\draw[->, line width=1pt] (1110) to (1111);

\draw[->, line width=1pt] (0011) to (0111);
\draw[->, line width=1pt] (0101) to (0111);
\draw[->, line width=1pt] (0110) to (0111);

\draw[->, line width=1pt] (0011) to (1011);
\draw[->, line width=1pt] (1010) to (1011);
\draw[->, line width=1pt] (1001) to (1011);

\draw[->, line width=1pt] (0101) to (1101);
\draw[->, line width=1pt] (1001) to (1101);
\draw[->, line width=1pt] (1100) to (1101);

\draw[->, line width=1pt] (0110) to (1110);
\draw[->, line width=1pt] (1010) to (1110);
\draw[->, line width=1pt] (1100) to (1110);
\end{tikzpicture}
\caption{Parallel transport network with underlying hypercube graph of degreee $4$. 
For the sake of clarity only the upper path is labelled by unitaries. Edges vertically below carry the same unitaries. The actual vertex states are attached by way of example to the upper and lower vertices.}
\end{figure}
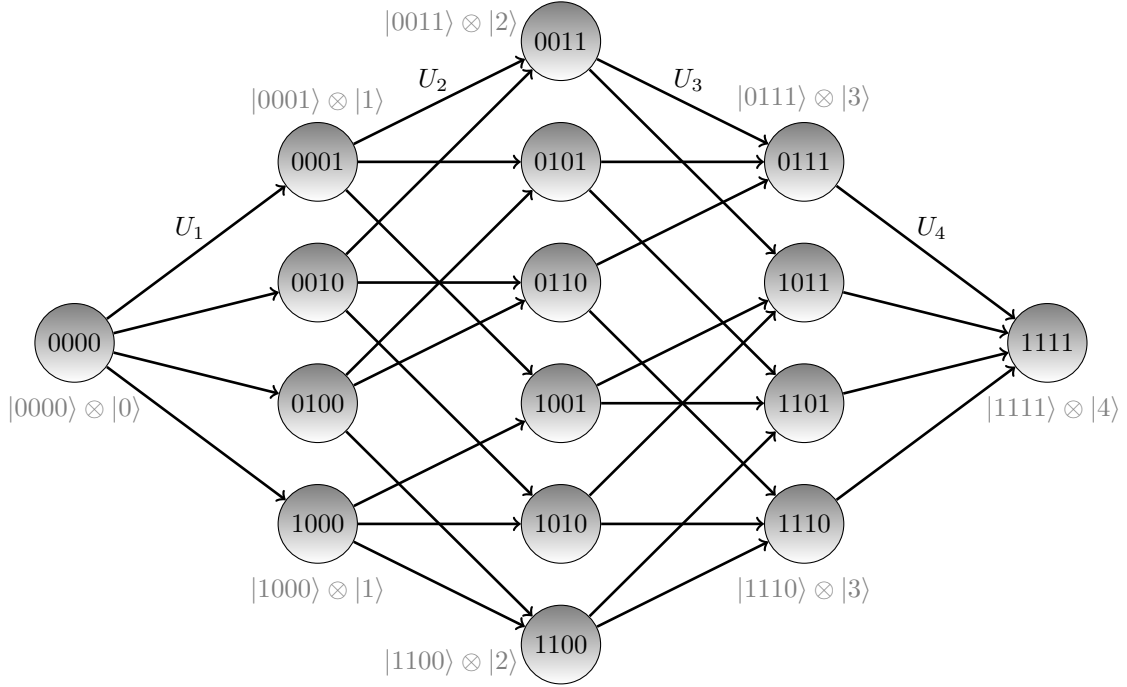

Also the other properties of a standard graph Hamiltonian from Definition \ref{def:standard} are fulfilled: $H(s)$ leaves the subspace $D$ of proper network states and its orthogonal space $D^\perp$ invariant, the penalty term $H_{in}$ is positive semi-definite, the null-space of $H_{in\vert D}$ is spanned by
\begin{align*}
&\big\{ \ket{z} \otimes \ket{t} \otimes \ket{\mathbf{0}} \,\big\vert\, z\in \{0,1\}^{L'},\, \vert z \vert_1 = t\le L_i\big\}\\
\cup &\big\{ \ket{z} \otimes \ket{t} \otimes \ket{x} \,\big\vert\,  z\in \{0,1\}^{L'},\, \vert z \vert_1=  t>L_i,\, x \in \{0,1\}^n\big\}
\end{align*}
and $H_{graph}$ vanishes on D and has an expectation value $\ge \frac{1}{L'} $ for all states from $D^\perp$.

$H_{in}$ is defined such that again $\ket{\mathbf{0}}$ is the only valid circuit input. We can hence use the constructed hypercube Hamiltonian for an adiabatic quantum computation with unique ground-state
\begin{align*}
\ket{\eta(s)}=\frac{1}{\sqrt{2^{L'}}}\sum_{t=0}^{L'} \sum_{\substack{z\in \{0,1\}^{L'}\\ \vert z\vert_1=t}} \ket{z} \otimes \ket{t} \otimes U_t(s) U_{t-1}(s) \dots U_1(s) \ket{\mathbf{0}}\text{.}
\end{align*}

The efficiency requirements 1--4 from the list in Section \ref{sec:circuitSimReq} are obviously fulfilled. Moreover the initial ground-state can be prepared from the the state $\ket{\mathbf{0}}\otimes \ket{\mathbf{0}}\otimes \ket{\mathbf{0}}$ by an efficient quantum circuit as required in point 5: Let's call the first $L'$ qubits the ``vertex qubits'' since they actually label the vertices and call the second $L'$ qubits the ``Hamming weight qubits''. First we apply Hadamard gates to each vertex qubit. Afterwards we apply for each vertex qubit the following circuit, which acts in addition to the vertex qubit on all Hamming weight qubits and increases the encoded Hamming weight by $1$ iff the vertex qubit equals $\ket{1}$:

\begin{figure}[\textwidth]
\centering
\begin{tikzpicture}[scale=0.9]
\draw[line width=1pt] (0,0) to (10,0);
\draw[line width=1pt] (0,-1.5) to (10,-1.5);
\draw[line width=1pt] (0,-2) to (10,-2);
\draw[line width=1pt] (0,-3) to (10,-3);
\draw[line width=1pt] (0,-3.5) to (10,-3.5);
\draw[line width=1pt] (0,-4) to (10,-4);

\draw[line width=0.7pt] (1,0) to (1,-4.25);
\draw[line width=0.7pt] (1.6,0) to (1.6,-4);
\draw[line width=0.7pt] (1,-4) circle (0.25);
\draw[line width=0.7pt] (1.6,-3.5) circle (0.25);
\fill (1,0) circle (0.15);
\fill (1.6,0) circle (0.15);
\fill (1,-3.5) circle (0.15);
\fill (1.6,-4) circle (0.15);

\draw[line width=0.7pt] (3,0) to (3,-3.75);
\draw[line width=0.7pt] (3.6,0) to (3.6,-3.5);
\draw[line width=0.7pt] (3,-3.5) circle (0.25);
\draw[line width=0.7pt] (3.6,-3) circle (0.25);
\fill (3,0) circle (0.15);
\fill (3.6,0) circle (0.15);
\fill (3,-3) circle (0.15);
\fill (3.6,-3.5) circle (0.15);

\draw[line width=0.7pt] (7,0) to (7,-2.25);
\draw[line width=0.7pt] (7.6,0) to (7.6,-2);
\draw[line width=0.7pt] (7,-2) circle (0.25);
\draw[line width=0.7pt] (7.6,-1.5) circle (0.25);
\fill (7,0) circle (0.15);
\fill (7.6,0) circle (0.15);
\fill (7,-1.5) circle (0.15);
\fill (7.6,-2) circle (0.15);

\draw[line width=0.7pt] (9,0) to (9,-1.75);
\draw[line width=0.7pt] (9,-1.5) circle (0.25);
\fill (9,0) circle (0.15);

\node at (-1.55,0) {\textsf{\small vertex qubit}};
\node at (-3.4,-2.75) {\textsf{\small qubits for Hamming weight encoding}};
\node at (-0.3,-2.75)[rotate=90] {$\overbrace{\hspace{6.7em}}^{}$};
\fill (5.3,-2.5) circle (0.03);
\fill (5.6,-2.4) circle (0.03);
\fill (5,-2.6) circle (0.03);

\end{tikzpicture}.
\end{figure}

\section{Trade-off between the gap and the output probability}

 It remains to determine the gap $\gamma$ of the Hamiltonian and the probability $p$ to obtain the correct computation output $U_{L'} \dots U_1 \ket{\mathbf{0}}$ after measuring on the graph subsystem of the final ground-state
\begin{align*}
\ket{\eta(1)}=\frac{1}{\sqrt{2^{L'}}}\sum_{t=0}^{L'} \sum_{\substack{z\in \{0,1\}^{L'}\\ \vert z\vert_1=t}} \ket{z} \otimes \ket{t} \otimes U_t U_{t-1} \dots U_1 \ket{\mathbf{0}}
\end{align*}
in the extended vertex basis and then discarding this subsystem. According to  Proposition \ref{prop:outputProb}
\begin{align*}
p = \sum_{t=L'-L_f}^{L'}\frac{\vol(V_t)}{\vol(V)} = \frac{1}{2^{L'}} \sum_{t=L'-L_f}^{L'}\binom{L'}{t}
= \frac{1}{2^{L'}} \sum_{t=0}^{L_f}\binom{L'}{t}\text{.}
\end{align*}

The gap $\frac{1}{L'}$ of $H_{graph}$ and the gap $\frac{1}{n}$ of $H_{in\vert D}$ are contained in $\Omega(\lambda)$ with $\lambda=\frac{2}{L'}$ being the spectral gap of our normalized Laplacian $H_{prop\vert D}$. Hence the gap $\gamma$ of $H(s)$ is determined according to Lemmata \ref{lem:gaps} and \ref{lem:angle} by
\begin{align*}
\Omega\left( \lambda \sin^2\left(\frac{\theta}{2}\right) \right) \ni \gamma \in \mathcal{O}\big(\sin^2(\theta)\big)
\end{align*}
with $\theta$ denoting the gap angle.

Proposition \ref{prop:angleExpl} allows us to write for the gap angle
\begin{align*}
\sin^2(\theta)&= \frac{\sum_{t=0}^{L_i}\vol(V_t)}{\vol(V)}
= \frac{1}{2^{L'}} \sum_{t=0}^{L_i}\binom{L'}{t}\text{,}
\end{align*}
which is exactly the expression for the probability $p$ just with the variable $L_i$ instead of $L_f$. If we did not add any initial or final identity gates, both quantities $p$ and $\sin^2(\theta)$ would decrease like $\frac{1}{2^{L'}}\le \frac{1}{2^{L}}$, which violates the efficieny requirement that $p$ and $\gamma$ are lower bounded by an inverse polynomial in the number of original circuit gates $L$.

The question now is whether the number of initial and final identity gates, $L_i$ and $L_f$, respectively, can be choosen large enough, such that both expressions are lower bounded by an inverse polynomial. Let's investigate the case $L_i=L_f \ge 1$, since if the equality case does not lead to lower bounds by inverse polynomials, then no case with $L_i \ne L_f$ does it either. Since $2L_i + L = L'$ we investigate now the expression
\begin{align*}
p=\sin^2(\theta) = \frac{1}{2^{L'}} \sum_{t=0}^{\frac{L'-L}{2}}\binom{L'}{t}\text{.}
\end{align*}

Unfortunatly it is difficult to transform this sum of binomial coefficients into a useful exact expression, but the next lemma, whose proof we simply refer to \cite[lemmata 2.5 and 3.3]{worsch}, helps us at least to derive a bound:

\begin{lem}\label{lem:worsch}
\quad

\begin{enumerate}[label={\normalfont(\roman*)}]
\item
${\displaystyle \quad\quad 
\sum_{t=0}^{\frac{L'}{a}} \binom{L'}{t} \in \mathcal{O}\left( \sqrt{L'} C(a)^{L'} \right)\quad}$
with $C$ a function obeying $1<C(a)<2$ for all $a>2$.
\item
${\displaystyle \quad\quad 
\binom{L'}{\frac{L'}{2}} < (1+ \epsilon) \sqrt{\frac{2}{\pi}} \frac{1}{\sqrt{L'}} 2^{L'} \quad\,\,}$
for all $\epsilon >0$.
\end{enumerate}
\end{lem}

Let's first assume that we add linearly many identities gates to the original circuit to keep the spectral gap $\lambda \in \Theta\left(\frac{1}{L'}\right)=\Theta\left(\frac{1}{L}\right)$, so
\begin{align*}
L' &= cL  \quad &&\text{with }c> 1\\
\frac{L'-L}{2}  &= \frac{L'}{a} \quad &&\text{with }a=\frac{2c}{c-1}\text{.}
\end{align*}

Since $a>2$ it holds $C(a)<2$ for the function in Lemma \ref{lem:worsch} (i). Hence there exists an $\epsilon>0$ such that $q + \epsilon <1$ with $q:=\left(\frac{C(a)}{2}\right)^c$. We can now derive 
\begin{align*}
p = \sin^2(\theta) &= \frac{1}{2^{L'}}\sum_{t=0}^{\frac{L'}{a}} \binom{L'}{t}\\
&\in \mathcal{O}\left( \sqrt{L'} \left(\frac{C(a)}{2}\right)^{L'}\right)\\
&= \mathcal{O}\left( \sqrt{L} \left(\left(\frac{C(a)}{2}\right)^{c}\right)^{L} \right)\\
&\subseteq \mathcal{O}\left( \left(\frac{1}{q+\epsilon}\right)^{L} q^{L} \right)\\
&= \mathcal{O}\left( \left(\frac{q}{q+\epsilon}\right)^{L} \right)\text{.}
\end{align*}
We see that the gap angle and hence the energy gap $\gamma$ as well as the output probablity $p$ still decrease exponentially with $L$ and can therefore not be lower bounded by an inverse polynomial as required for an efficient adiabatic quantum computation.

Of course one can easily find a non-uniform distribution of initial and final identity gates that allows at least one of the quantities to be lower bounded by an inverse polynomial. We can reach even the constant probability $p=\frac{1}{2}$ by $L_i=0$ and $L_f=L$ or a constant angle $\sin^2(\theta)=\frac{1}{2}$ and therefore a gap $\gamma \in \Omega\left(\frac{1}{L}\right)$ by $L_i=L$ and $L_f=0$. But unfortunately we cannot achieve a sufficient large initial and final vertex set at the same time to balance the exponential volume of the whole vertex set. This is exactly the behaviour described in Section \ref{sec:limitations} for graphs with nonpolynomial vertex set.

Let's consider now to add more than linearly many identity gates symmetrically ($L_i=L_f$) to the original circuit. With the help of Lemma \ref{lem:worsch} (ii) we can derive a sufficiently large lower bound for the output probability and the gap angle:
\begin{align*}
p=\sin^2(\theta)&= \frac{1}{2^{L'}} \sum_{t=0}^{\frac{L'-L}{2}}\binom{L'}{t}\\
&= \frac{1}{2}- \frac{1}{2^{L'}}\sum_{t=\frac{L'-L}{2}+1}^{\frac{L'}{2}}\binom{L'}{t}\\
&\ge\frac{1}{2} -  \frac{1}{2^{L'}} \frac{L}{2} \binom{L'}{ \frac{L'}{2}}\\
& \ge \frac{1}{2} - \frac{1}{2^{L'}}\frac{L}{2} (1+\epsilon) \sqrt{\frac{2}{\pi}} \frac{1}{\sqrt{L'}}2^{L'}\\
&\ge \frac{1}{2} - \frac{L}{\sqrt{L'}}\frac{1+\epsilon}{\sqrt{2\pi}}\text{.}
\end{align*}
If we choose $L'=L^2$, there is an $\epsilon>0$ such that the expression is always lower bounded by a non-vanishing constant, hence
\begin{align*}
p = \sin^2(\theta) \in \Theta(1)\text{.}
\end{align*}
So the probability of measuring the correct output is constant and the gap $\gamma$ of the Hamiltonian $H(s)$ is lower bounded by the spectral gap of the graph Laplacian according to Lemma \ref{lem:angle}. Unfortunately in terms of the origianl circuit this lower bound is now $\gamma \in \Omega\left(\frac{1}{L'}\right) = \Omega\left(\frac{1}{L^2}\right)$, which is not an improvement over the path graph. 

We re-iterate that in this chapter we have set-up a standard graph Hamiltonian based on a hypercube graph that easily fulfills requirements 1--5 from the efficiency and practicability list in Section \ref{sec:circuitSimReq}. Hence the Hamiltonian can be used for an adiabatic quantum computation, but the probability of measuring the correct computation output and its running time, determined by the energy gap, are not ensured to be efficient. If we choose the hypercube of degree $L' \in \Theta(L)$ with $L$ the length of the simulated circuit, one of the mentioned quantities is certainly not efficient. For larger degrees the performance improves, the edge for a provable efficient adiabatic quantum computation is passed at $L' \in \Theta(L^2)$ which implies an energy gap $\gamma \in \Omega\left(\frac{1}{L^2}\right)$ of the Hamiltonian.  Perhaps it is more than just a coincidence that this gap exactly equals the one of the standard Kitaev Hamiltonian construction. Also in the next chapter, when we undertake a short attempt to improve the spectral gap by looking at weighted graphs, we will finally end up again with $\gamma \in \Omega\left(\frac{1}{L^2}\right)$. Apart from this it is at least notable that we have found in this chapter a Hamiltonian based on an exponentially sized graph that can be used for efficient adiabatic quantum computation.

\chapter{Some Ideas about Weighted Graph Hamiltonians}\label{chap:weighted}
\ihead{CHAPTER 8.\,\, SOME IDEAS ABOUT WEIGHTED GRAPH HAMILTONIANS}

\section{Motivation for weighted path reductions}

It is a challenging task to find a Hamiltonian construction for an efficient adiabatic quantum computation that exceeds the gap $\lambda \in \Theta\left(\frac{1}{L^2}\right)$ known from the standard Kitaev construction. From Theorem \ref{thm:linearDiam} we know that any nonweighted parallel transport network with $\lambda \in \Omega\left(\frac{1}{L^k}\right)$, $k<2$, has an vertex set whose size is not upper bounded by a polynomial, hence local implementation of such a network with polynomial many qubits is a non-trivial, to impossible task. In the previous chapter we found a local implementation for one graph of this family, namely the hypercube. Unfortunately it turned out that this graph has a vertex expansion that makes it impossible to have an inverse polynomial output probability and an inverse polynomial gap angle of Laplacian and input penalty term at the same time.

In this chapter we avoid the difficulty to handle nonpolynomial sized graphs by broadening our focus to include weighted graphs. The connection between a large gap and vertex expansion just holds for unweighted graphs, hence among weighted graphs there might be (and indeed are) polynomial sized graphs with arbitrary gaps, even constant ones. Unfortunately the research area of spectral graph theory mostly concentrates on unweighted graphs and therefore weighted graph families with interesting gaps and properties are not widely known. That's why we concentrate in this chaper on transforming unweighted graph families with interesting gaps into weighted graph families by keeping a lower bound on the spectral gap and the neccessary diameter. Section \ref{sec:contraction} introduced two tools for this purpose: contraction and covering.

The easiest graph that can be implemented locally, even if weights are assigned to its edges, is a path graph like in the Kitaev contruction. Since both contraction and covering only affect edges touching the combined vertices, they both offer us a controlled method of reducing the underlying graph of a parallel transport network to a path graph by keeping the diameter and hence the capability to implement the same quantum circuit.

\section{Simplification of local implementation via path contraction}

\begin{definition}
Let $\mathcal{G}=(G,\mathcal{T}',\mathcal{U},\pi)$, $G=(V',E',w')$, be a parallel transport network with time map $\mathcal{T}: V' \rightarrow \{0,\dots L'\}$. Then the parallel transport network $\mathcal{H}=(V,E,w,\mathcal{T},\mathcal{U},\pi)$ with $V=\{0, \dots L'\}$, $\mathcal{T}(t) = t$ for all $t\in V$ and $H=(V,E,w)$ the contraction of $G$ via the contraction function
\begin{align*}
c(a)&= \mathcal{T}'(a)
\end{align*}
for all $a\in V'$ is called the \textsf{path contraction} of $\mathcal{G}$.
\end{definition}

Informally speaking all vertices belonging to the same time step are contracted and the new vertex represents that time step.
Since there is exactly one vertex in $V$ for every time step, we call $H$ a path graph with weights. Of course the definition of the new time map $\mathcal{T}$ is consistent with Definition \ref{def:PTN} of a parallel transport network: vertices of the contracted graph are not connected if their time steps differ more than by $1$. But of course the graph stays connected as a whole,  hence we can use it for an implementation of the same circuit. 

We have seen that networks with a spectral gap $\lambda \in \Omega\left(\frac{1}{L^k}\right)$, $k<2$, may have an output probability or gap angle that is to small for an efficient adiabatic quantum computation because of their nonpolynomial vertex set size. Unfortunately if the original parallel transport network shows this behaviour, then the path contracted network will also:
\begin{lem}
Let $H(s)$ be a standard graph Hamiltonian with underlying parallel transport network $\mathcal{G}$. Then a standard graph Hamiltonian whose underlying parallel transport network is the path contraction of $\mathcal{G}$ has the same gap angle $\theta$ and probability $p$ of measuring the correct computation output as $H(s)$.
\end{lem}
\begin{proof}
The reasoning beyond the inheritance of gap angle and output probability is that path contraction preserves the volume of each time cluster. Denote by $G=(V',E',w')$ the underlying graph of $\mathcal{G}$, by $H=(V,E,w)$ the underlying graph of the path contraction $\mathcal{H}$ and by $c$ the contraction function. Then it holds for every time step $t$ according to Lemma \ref{lem:contractionDegree}:
\begin{align*}
\vol(V_t) &= d_t
=  \sum_{a \in c^{-1}(t)} d'_a
=  \sum_{a \in V'_t} d'_a
=\vol(V'_t)\text{,}
\end{align*}
hence the expressions for the gap angle according to Proposition \ref{prop:angleExpl} and for the output probability according to Proposition \ref{prop:outputProb} give the same result for both standard graph Hamiltonians.
\end{proof}
 
We know from Theorem \ref{thm:tradeoff} that for standard graph Hamiltonians with spectral gap $\lambda \in \Omega\left(\frac{1}{L^k}\right)$, $k<1$, it is always the case that either the output probability $p$ or the gap angle $\theta$ cannot be lower bounded by an inverse polynomial. Consequently standard graph Hamiltonians based on path contractions of such parallel transport networks cannot be used for an efficient adiabatic quantum computation either. The only new helpful tool that path contraction offers us is an easier local implementation: Assume one finds a parallel transport network with an attractive spectral gap (for example $\lambda\in \Theta\left(\frac{1}{L}\right)$) and sufficiently large gap angle and output probability but no way to write its normalized Laplacian as a local Hamiltonian. Then one can instead just take the normalized Laplacian of its path contraction, which is the local Kitaev Hamiltonian from Chapter \ref{chap:path} with adjusted prefactors of the interaction terms. Of course one has to ensure that the prefactors now do not contradict requirement 3 from the effecieny list in Section \ref{sec:circuitSimReq}.

As an example let's write down the standard graph Hamiltonian based on the path contraction of the hypercube of degree $L'=L^2$ with symmetrically appended identity gates as we have presented it at the end of the last chapter. Of course we found a way to write a local Hamiltonian for this construction and furthermore its energy gap $\gamma=\Theta\left(\frac{1}{L^2}\right)$ is the same as for the Kitaev Hamiltonian, but it can serve as example to demonstrate the easier local implementation.

Definition \ref{def:contraction} of path contraction tells us that the weights and degrees of the path contracted hypercube equal the following:
\begin{align*}
w(t-1,t)&=t\binom{L'}{t}\\
d_t&=w(t-1,t) + w(t,t+1)= (L'-1) \binom{L'}{t}\text{.}
\end{align*}
Weights and degrees get exponentially large around the center time $t=\frac{L'}{2}$, but the relevant prefactors in the normalized Laplacian
\begin{align*}
\frac{w(t-1,t)}{\sqrt{d_{t-1} d_t}} &=\frac{\sqrt{t(L'-t)}}{L'-1}
\end{align*}
are still lower bounded by polynomials. Hence the standard graph Hamiltonian $H(s)=H_{prop}(s) + H_{in} + H_{graph}$ with
\begin{align*}
H_{prop}(s):=&\left(\sum_{t=0}^{L'} \overline{\ket{t}\bra{t}}\otimes \mathbb{I} \right)
-\left(\sum_{t=1}^{L'} \frac{\sqrt{t(L'-t)}}{L'-1} \left(\overline{\ket{t}\bra{t-1}} \otimes U_t(s)
+\overline{\ket{t-1}\bra{t}} \otimes U^\dagger_t(s)\right)
\right)\\
H_{in} :=& \frac{1}{n} \sum_{t=0}^{L_i}\sum_{i=1}^n \overline{\ket{t}\bra{t}} \otimes \ket{1}\bra{1}_i\\
H_{graph}:=& \frac{1}{L'} \sum_{i = 2}^{L'} \ket{01}\bra{01}_{i-1,i} \otimes \mathbb{I}
\end{align*}
can be used for an efficient adiabatic quantum computation with energy gap $\gamma=\Theta\left(\frac{1}{L^2}\right)$ and constant output probability. Compared to the original hypercube construction this Hamiltonian is only defined on $L'+n$ instead of $2L'+n$ qubits since we got rid of the need to define the artifical Hamming weight space that originally ensured us a local implementation.

\begin{center}
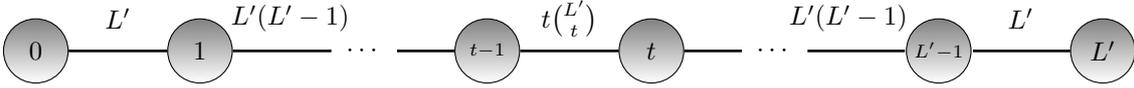
\begin{figure}[h]
\begin{tikzpicture}[scale=1.08]
\node (0) at (0,0) [circle,shade,draw,minimum size=8.5mm] {$0$};
\node (1) at (2,0) [circle,shade,draw,minimum size=8.5mm] {$1$};
\node (P) at (4,0) [minimum size=8.5mm] {$\cdots$};
\node (M1) at (5.5,0) [circle,shade,draw,minimum size=8.5mm] {$\,$};
\node (M1*) at (5.5,0) [minimum size=8.5mm] {$\scriptstyle t-1$};
\node (M2) at (7.5,0) [circle,shade,draw,minimum size=8.5mm] {$\,$};
\node (M2*) at (7.5,0) [minimum size=8.5mm] {$t$};
\node (P2) at (9,0) [minimum size=8.5mm] {$\cdots$};
\node (L1) at (11,0) [circle,shade,draw,minimum size=8.5mm] {$\,$};
\node (L1) at (11,0) [minimum size=8.5mm] {$\scriptstyle L'-1$};
\node (L) at (13,0) [circle,shade,draw,minimum size=8.5mm] {$L'$};
\draw[line width=1pt] (0) to (1);
\draw[line width=1pt] (1) to (P);
\draw[line width=1pt] (P) to (M1);
\draw[line width=1pt] (M1) to (M2);
\draw[line width=1pt] (M2) to (P2);
\draw[line width=1pt] (P2) to (L1);
\draw[line width=1pt] (L1) to (L);
\node at (1,0.4) {$ L'$};
\node at (3.1,0.4) {$ L'(L'-1)$};
\node at (6.5,0.4) {$ t\binom{L'}{t}$};
\node at (9.9,0.4) {$ L'(L'-1)$};
\node at (12,0.4) {$ L'$};
\end{tikzpicture}
\caption{Path contraction of a hypercube with degree $L'$.}
\end{figure}
\end{center}

\section{Hypercube of linear degree as covering graph}

In the previous section we investigated contractions of arbitrary graphs to weighted path graphs to ensure local implementation while keeping the spectral gap. Unfortunately too small gap angles and output probabilities are preserved, too, since the degrees of contracted vertices simply sum up. In this section we therefore want to consider a graph as a covering of a weighted path graph, since in this method the degrees of vertices get adjusted differently -- perhaps for the efficiency benefit of gap angle and output probability?

The weight definition \ref{def:covering} of a covered graph makes the formulation of general statements for path coverings more difficult than for path contractions. Moreover not every graph can be seen as covering of a weighted path graph in contrast to the contraction method. Therefore we will concentrate on the hypercube with a degree linear in the length $L$ of the simulated circuit as an example of a graph with attractive gap $\lambda\in \Theta\left(\frac{1}{L}\right)$ that can be understood as covering of a weighted path graph. It would be a remarkable progress if we could construct via covering a Hamiltonian construction that inherits the gap while achieving an improved output probability and gap angle, which turned out not to be efficient for the hypercube of linear degree.

\begin{definition}\label{def:pathCovering}
Let $\mathcal{G}=(G,\mathcal{T}',\mathcal{U},\pi)$, $G=(V',E',w')$, be a parallel transport network with time map $\mathcal{T}: V' \rightarrow \{0,\dots L'\}$ and $\mathcal{H}=(H,\mathcal{T},\mathcal{U},\pi)$, $H=(V,E,w)$, a parallel transport network with $V=\{0, \dots L'\}$, $\mathcal{T}(t) = t$ for all $t\in V$. Iff $G$ is the covering of $H=(V,E,w)$ via the covering function
\begin{align*}
c(a)&= \mathcal{T}'(a)
\end{align*}
for all $a\in V'$, then $\mathcal{H}$ is called the \textsf{covered path} of $\mathcal{G}$ and $\mathcal{G}$ the \textsf{path covering} of $\mathcal{H}$.
\end{definition}

Similar to the case of path contraction also path covering ensures that the covered network $\mathcal{H}$ is connected and vertices whose assigned time steps differ by more than $1$ are not adjacent. Since there is again exactly one vertex for every time step, $H$ has indeed the form of a weighted path graph. 

Now let $\mathcal{G}=(V',E',w',\mathcal{T}',\mathcal{U},\pi)$ be the parallel transport network with underlying hypercube as it was presented in the previous chapter. According to Definition \ref{def:covering} its covered path $\mathcal{H}=(V,E,w,\mathcal{T},\mathcal{U},\pi)$ has only edges between successive time steps with weights
\begin{align*}
w(t,t+1) &= \frac{1}{\sqrt{\left\vert V_t \right\vert \left\vert V_{t+1}\right\vert}}\sum_{\substack{a \in V'_t \\b \in V'_{t+1}}} w'(a,b)\\
&= \frac{1}{\sqrt{\binom{L'}{t} \binom{L'}{t+1}}} \sqrt{\binom {L'}{t} (L'-t)} \sqrt{\binom{L'}{t+1}(t+1)}\\
&=  \sqrt{(L'-t)(t+1)}\text{.}
\end{align*}

Since this weight definition also fulfills the second requirement in Definition \ref{def:covering}
\begin{align*}
\sum_{b \in V_{t}} w'(a,b) =\sum_{b \in V_{t}} w'(a^*,b) \quad\quad \text{ for all } a,a^* \in V'_{t'} \text{ and for all }t \in V\text{,}
\end{align*}
$\mathcal{G}$ is indeed a correct covering of $\mathcal{H}$.

We already see that the weights grow polynomially and not exponentially as was the case with the path contraction in the previous section. As our covered path has only polynomial many vertices the volume of any vertex set is upper bounded by a polynomial and that is according to Propositions \ref{prop:angleExpl} and \ref{prop:outputProb} already enough to conclude that gap angle and output probability are efficient in contrast to the covering hypercube network.

Let's derive some explicit bounds for the gap angle $\theta$ and the output probability $p$. The weight $w(t,t+1)$ as a function of $t\in \{0, \dots L'-1\}$ is concave and symmetric around its maximum $\frac{L'+1}{2}$ at $t= \frac{L'-1}{2}$ (Figure \ref{fig:hypercube_covering_w}). With this knowledge we can bound
\begin{align*}
d_{max}&\le 2 \max_{0\le t\le L'-1} w(t,t+1) =  L'+1\\
d_{min}&=d_0=w(0,1)= \sqrt{L'}\\
\vol(V) &\ge (L'+1) d_{min} = (L'+1)\sqrt{L'} \ge L'^{1.5}\\
\vol(V) &\le (L'+1) d_{max} \le (L'+1)^2\text{.}
\end{align*}

\begin{figure}[h]
\centering
 \includegraphics[width=0.61\textwidth]{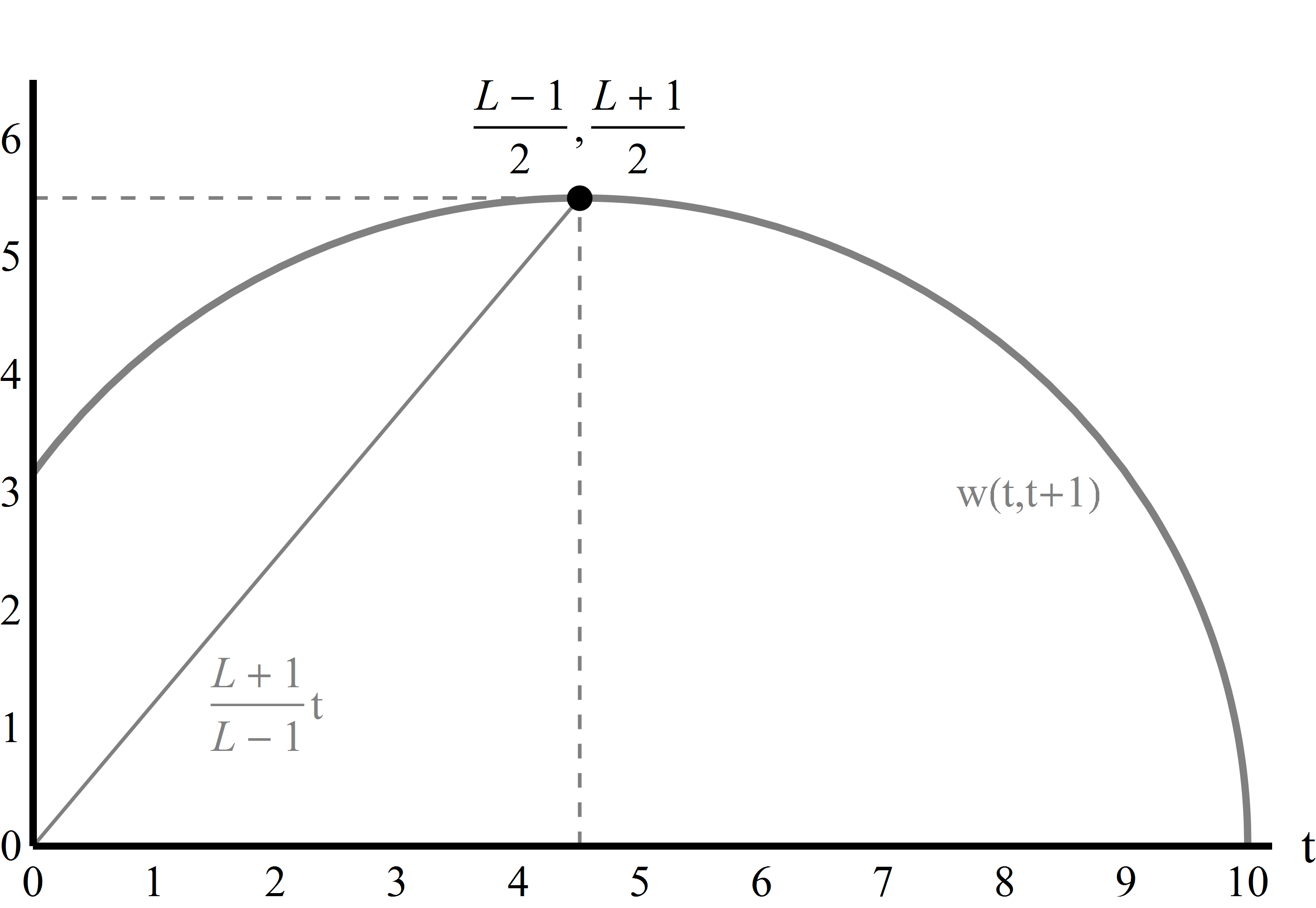}
  \caption{The weight function $w(t,t+1)$ for $L'=10$.}
  \label{fig:hypercube_covering_w}
\end{figure}

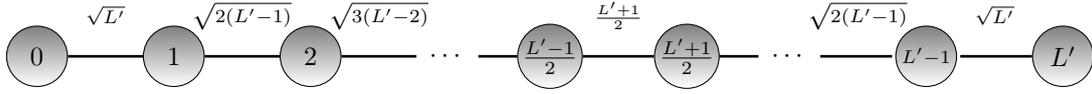
\begin{figure}[ht]
\centering
\begin{tikzpicture}[scale=0.9]
\node (0) at (0,0) [circle,shade,draw,minimum size=8mm] {$0$};
\node (1) at (2,0) [circle,shade,draw,minimum size=8mm] {$1$};
\node (2) at (4,0) [circle,shade,draw,minimum size=8mm] {$2$};
\node (P) at (6,0) [minimum size=8mm] {$\cdots$};
\node (M1) at (7.5,0) [circle,shade,draw,minimum size=8.5mm] {$\,$};
\node (M1*) at (7.5,0) [minimum size=8.5mm] {$\frac{L'-1}{2}$};
\node (M2) at (9.5,0) [circle,shade,draw,minimum size=8.5mm] {$\,$};
\node (M2*) at (9.5,0) [minimum size=8.5mm] {$\frac{L'+1}{2}$};
\node (P2) at (11,0) [minimum size=8mm] {$\cdots$};
\node (L1) at (13,0) [circle,shade,draw,minimum size=8mm] {$\,$};
\node (L1) at (13,0) [minimum size=8mm] {$\scriptstyle L'-1$};
\node (L) at (15,0) [circle,shade,draw,minimum size=8mm] {$L'$};
\draw[line width=1pt] (0) to (1);
\draw[line width=1pt] (1) to (2);
\draw[line width=1pt] (2) to (P);
\draw[line width=1pt] (P) to (M1);
\draw[line width=1pt] (M1) to (M2);
\draw[line width=1pt] (M2) to (P2);
\draw[line width=1pt] (P2) to (L1);
\draw[line width=1pt] (L1) to (L);
\node at (1,0.6) {$\scriptstyle\sqrt{L'}$};
\node at (3,0.6) {$\scriptstyle\sqrt{2(L'-1)}$};
\node at (5,0.6) {$\scriptstyle\sqrt{3(L'-2)}$};
\node at (8.5,0.6) {$\scriptstyle\frac{L'+1}{2}$};
\node at (12,0.6) {$\scriptstyle\sqrt{2(L'-1)}$};
\node at (14,0.6) {$\scriptstyle\sqrt{L'}$};
\end{tikzpicture}
\caption{Covered path of the hypercube with degree $L'$.}
\end{figure}

If we assume now that there are no initial and final identity gates appended to the original circuit, $t=0$ is the only initial and $t=L=L'$ the only final vertex and hence according to Propositions \ref{prop:angleExpl} and \ref{prop:outputProb}
\begin{align*}
\sin^2(\theta) = \frac{d_0}{\vol(V)} =\frac{d_{L}}{\vol(V)} = p
\end{align*}
which results after inserting the derived bounds in 
\begin{align*}
\Omega\left(\frac{1}{L^{1.5}}\right) \ni \sin^2(\theta)=p \in \mathcal{O}\left(\frac{1}{L} \right)\text{.}
\end{align*}

This is a remarkable improvement over the original hypercube construction where either $\theta$ or $p$ could not be lower bounded by an inverse polynomial. The bounds can be further improved by appending initial and final identity gates to the original quantum circuit. Let's add again symmetrically linear many identiy gates, so $L_i=L_f\le \frac{L'}{a}$, $a > 2$:
\begin{align*}
p=\sin^2(\theta) &= \frac{\sum\limits_{t=0}^{L_0} d_t}{\vol(V)}\\
&\ge \frac{ \sum\limits_{t=0}^{\frac{L'}{a}} w(t,t+1)}{\vol(V)}\\
&\ge \frac{ L' w^*}{a\vol(V)}
\end{align*}
with $w^*$ denoting the average of the value $w(t,t+1)$ for the integers $t$ from $0$ to $\left\lfloor \frac{L'}{a} \right\rfloor$. Since $w(t,t+1)$ is concave, this average can be lower bounded by the average of the straight line $f(t) = \frac{L'+1}{L'-1}t$ which turns out to be $\frac{L'+1}{L'-1}\frac{1}{2}\left\lfloor\frac{L'}{a}\right\rfloor \ge \frac{L'+1}{4a}$ for large enough $L'$. Inserting this and the derived upper bound for $\vol(V)$ leads to:
\begin{align*}
p=\sin^2(\theta)&\ge\frac{1}{4a^2}\frac{L'}{(L'+1)}\\
&=\frac{1}{4a^2} \left(1-\frac{1}{L'}\right)\\
p=\sin^2(\theta) &\in \Theta(1)\text{.}
\end{align*}

In the orginal hypercube contruction even adding linearly many initial and final identity gates did not resolve the problem that either $p$ or $\theta$ could not be lower bounded by an inverse polynomial, whereas now in the case of its covered path graph these quantities increase to a constant. Because the gap angle is constant, the gap of the Hamiltonian equals the spectral gap of the weighted path. If the spectral gap $\Theta\left(\frac{1}{L}\right)$ of the hypercube is preserved, then the gap of the Hamiltonian will be a significant improvement over the lower bound $\Omega\left( \frac{1}{L^2}\right)$ of the previous constructions.

Unfortunately a graph does not inherit the Laplacian eigenvalues of its covering graph, but the adjacency eigenvalues (see Theorem \ref{thm:covering}). The second largest adjacency eigenvalue of the $L'$-regular hypercube and hence the covered path graph is $L'-2$ as we can conclude from Lemma \ref{lem:hypercubeGap}. With Theorem \ref{thm:relationship} regarding the relationship between the different spectra we can derive the upper bound
\begin{align*}
\lambda \in \mathcal{O}\left(\frac{1}{L'}\right) = \mathcal{O}\left(\frac{1}{L}\right)
\end{align*}
for the spectral gap of the weighted path graph. This bound is not a contradiction to the desired improvement. Unfortunately we cannot derive a meaningful lower bound on the spectral gap by Theorem \ref{thm:relationship}. But a simple numerical plot in Figure \ref{fig:hypercube_covering_gap} reveals to us that the gap rather scales like $\Theta\left(\frac{1}{L^2}\right)$, which is the well-known behaviour from all our constructions so far. So unfortunately path covering worsens the orignal gap $\Theta\left(\frac{1}{L}\right)$ of the hypercube to $\Theta\left(\frac{1}{L^2}\right)$ and our hope to find a Hamiltonian for a more efficient adiabatic quantum computation using this method is unfulfilled, too.

There is also an intuition explaining why the spectral gap of the weighted path behaves like $\Theta\left(\frac{1}{L^2}\right)$. The weights of the two edges of each vertices are not too different, so the off-diagonal elements of the normalized Laplacian are close to $-\frac{w(t,t+1)}{\sqrt{d_t d_{t+1}}} \approx - \frac{1}{2}$ as one can also see in Figure \ref{fig:hypercube_covering_off}. But a tridiagonal matrix with diagonal elements $1$ and off-diagonal elements $-\frac{1}{2}$ is exactly the normalized Laplacian of an unweighted path graph whose gap is known to be $\Theta\left(\frac{1}{L^2}\right)$ (this holds even if the path has no initial and final loop as one can see in \cite{yueh}).

\begin{figure}[b]
\begin{minipage}[h]{0.47\textwidth}
 \includegraphics[width=\textwidth]{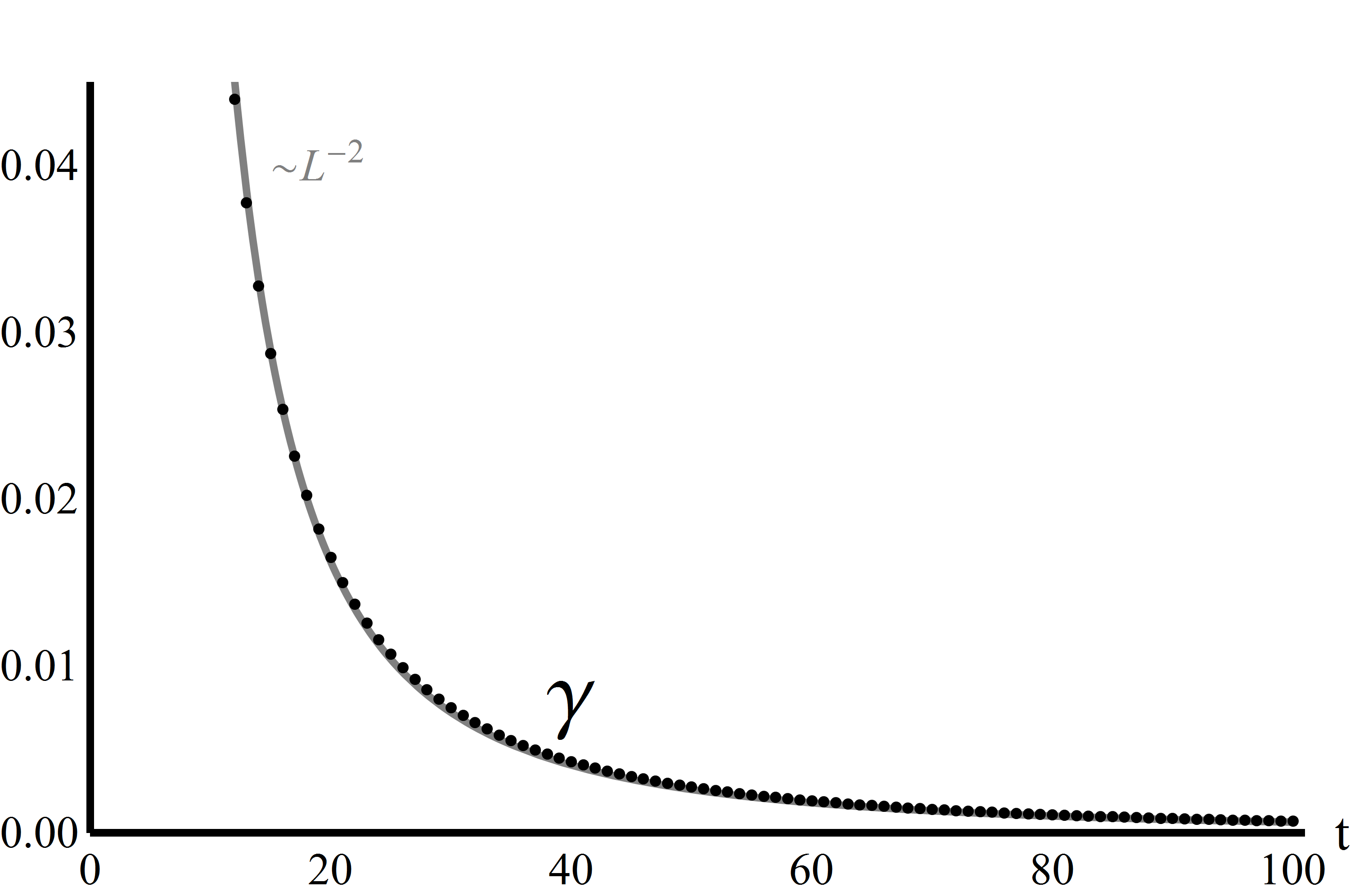}
  \caption{The gap $\gamma$ of the covered path Hamiltonian in comparision to $\Theta(L^{-2})$.}
  \label{fig:hypercube_covering_gap}
\end{minipage}
\begin{minipage}[h]{0.03\textwidth}
\quad
\end{minipage}
\begin{minipage}[h]{0.50\textwidth}
 \includegraphics[width=\textwidth]{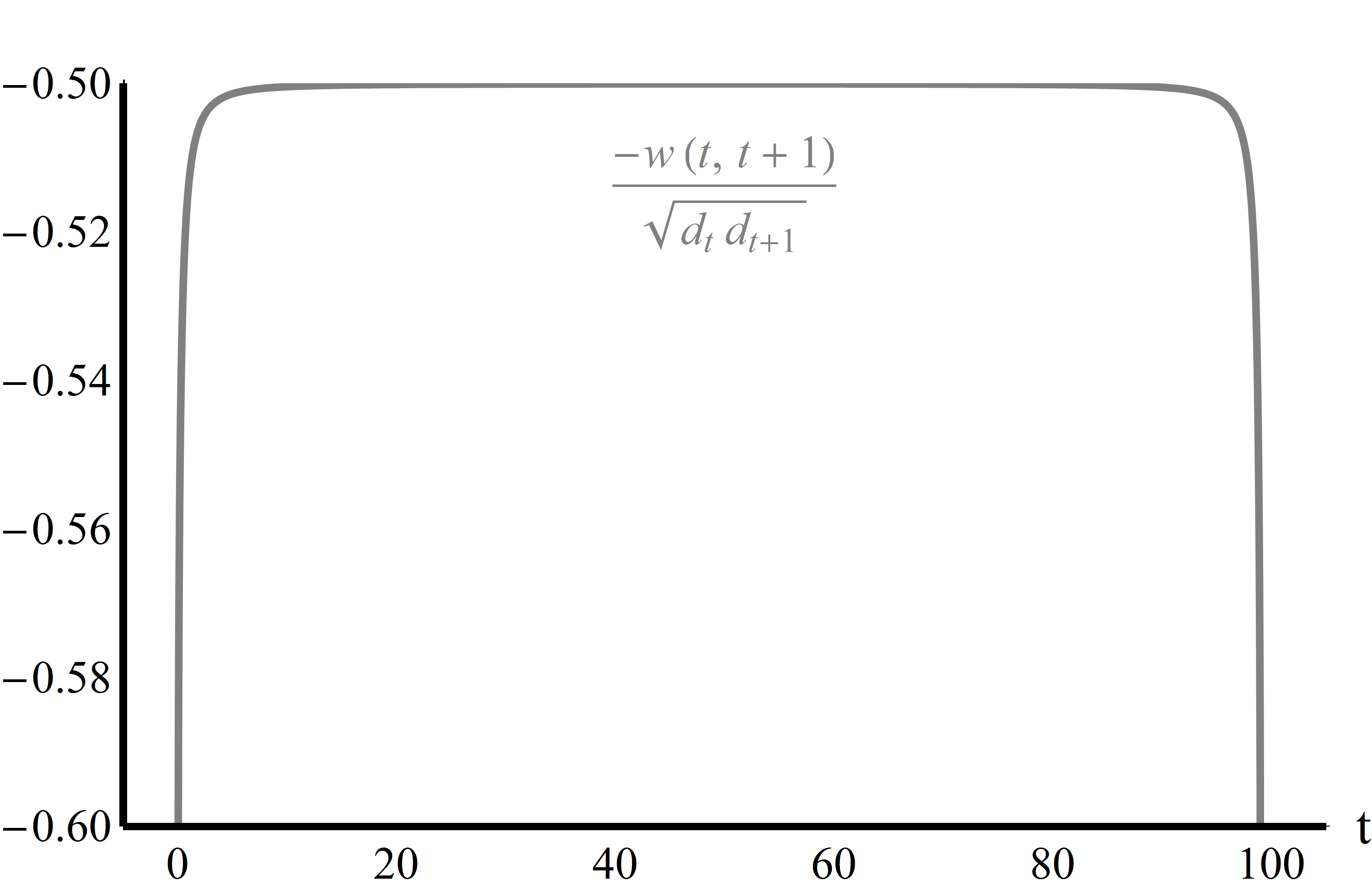}
  \caption{The off-diagonal elements of the normalized Laplacian for $L'=100$.}
  \label{fig:hypercube_covering_off}
\end{minipage}
\end{figure}

The above argument is not mathematically rigorous since the normalized Laplacian does not actually converge to the one of the unweighted path graph. But the idea that the similar neighboring weights lead to this behaviour can be extended to an idea of a further trade-off for weighted path graphs: In the previous section the weights of the path contraction of the hypercube differed exponentially which caused a gap angle $\theta$ or output probability $p$ which was too small. In contrast to this in this section the neighboring weights are too similiar such that the normalized Laplacian is very close to the one of an unweighted path graph and shows its gap behaviour of $\Theta\left(\frac{1}{L^2}\right)$. This raises the question of whether there is a weight assignment that balances both effects off such that the Hamiltonian construction based on the path graph has an efficient output probability and an energy gap better than $\Theta\left(\frac{1}{L^2}\right)$ at the same time.

As a covered path graph does not neccessarily inherit the exponential weight ratios of its covering graph, it is even worth investigating expander graphs as covering graphs. Perhaps one of them leads to a weight assignment as just described. However most expander graphs have a rather difficult structure for defining a suitable time map $\mathcal{T}:V\rightarrow \{0,\dots L'\}$ and hence calculating the weights of the covered path might not be an easy task. Finally we can at least note that we have found in this section again a truely weighted graph that can be used in a standard Hamiltonian construction for an efficient adiabatic quantum computation with energy gap $\gamma \in \Theta\left(\frac{1}{L^2}\right)$ and constant output probability.

\chapter{Summary and Outlook}
\ihead{CHAPTER 9.\,\, SUMMARY AND OUTLOOK}

In this thesis, we have investigated extensions of the Kitaev Hamiltonian for the problem of achieving efficient adiabatic quantum computation since the specific Hamiltonian structure allows us to change the efficiency by modifying the spectral gap of an underlying graph. Although we haven't succeeded in actually contructing a Hamiltonian which improves the spectral gap of $\lambda\in\Omega\left(\frac{1}{L^2}\right)$ ($L$: length of the simulated circuit) we have shown several useful restrictive results and presented new Hamiltonians of the same efficiency performance as the Kitaev Hamiltonian.

Already in the first three basic chapters of the thesis we have introduced some new methods, for example the concept of parallel transport networks in Chapter 2 and a quantum adiabatic theorem for projection operators in Chapter 4. Moreover we motivated a list of requirements for an efficient adiabatic quantum computation which was later specified in Chapter 5 under the assumption of standard graph Hamiltonians. We showed that the norm of their time derivatives is always constant and hence do not influence the evolution time of the adiabatic quantum computation. Furthermore the definition of a standard graph Hamiltonian directly led to some easy expressions for the relevant complexity quantities and allowed us to derive some restrictive results for graphs to be capable of an efficient adiabatic quantum computation. On the one hand we have seen that a spectral gap $\lambda\in\Omega\left(\frac{1}{L^k}\right)$, $k<2$, of a graph with constant degree ratio directly implies a vertex set size that scales faster than polynomial, on the other hand we could prove that graphs with spectral gap $\lambda\in\Omega\left(\frac{1}{L^k}\right)$, $k<1$, cannot even be used for an efficient adiabatic quantum computation at all.

After presenting the Kitaev Hamiltonian in Chapter 6, we constructed a standard graph Hamiltonian based on a hypercube graph in Chapter 7. In its first configuration (hypercube of linear degree) this Hamiltonian showed a vertex expansion that led to an inefficient adiabatic quantum computation. This possible behaviour had already been predicted earlier in Chapter 5. The second hypercube configuration (of degree $L^2$) scaled the spectral gap again down to the same as in the Kitaev contruction, but thereby became a Hamiltonian that can be used for an efficient adiabatic quantum computation.

The thesis finished with a short excursion to standard graph Hamiltonians based on weighted graphs in Chapter 8. It turned out that contraction does not resolve problematic expansion properties of a graph but might help for an easier local implementation of the Hamiltonian, whereas path covering suggests to be the more promising tool for turning unweighted graphs into weighted graphs with adjusted expansion properties.

In particular, the last Chapters motivate us to investigate some more graphs since the restriction results so far do not contradict a possible gap of $\lambda\in\Omega\left(\frac{1}{L}\right)$ and other graphs beside the path graph of the Kitaev construction have been presented to be capable of an efficient adiabatic quantum computation. For this purpose it would be desirable to find a straight forward way to contruct optimzied graphs candidates. The research field of spectral graph theory has developed a lot of tools for the adjustment of certain graph properties, inlcuding powering, replacing and certain graph products. Perhaps there is a way to extend some of these tools to parallel transport network, guaranteeing a certain diameter and time map, to transform a starting graph into a candidate for an effective adiabatic quantum computation.

From the other point of view a result proving that no improvement is possible can perhaps be derived by further investigating the tradeoff and implications between spectral gap, vertex expansion and gap angle. It is also conceivable that a combinatorical argument will show that the expansion properties of graphs with a spectral gap better than $\lambda\in\Omega\left(\frac{1}{L^2}\right)$ contradicts the required locality of the normalized Laplacian.

In this case one has to go beyond the model of standard graph Hamiltonians to improve the evolution time of the adiabatic quantum computation. As mentioned in the context of several acceptable input and output states, a concept that extends and assumes certain knowledge about the simulated quantum circuit might provide the oppertunity to define a different penalty term that results in an improved gap angle. Taking one step further, one could try to extend the idea of underlying graphs to hypergraphs or can overthrow the invariance assumption of parallel transport networks that different paths between two vertices always have to carry the same assigned unitary. Another suggestion is to find a totally different realization of a clock. It is conspicuous that in the original Kitaev construction the path graph directly corresponds to the function of a clock, whereas any other graph like the hypercube manifests a clear distinction between the actual vertex space and an artificial clock space.

Finally the whole model of adiabatic quantum computation is based on the fact that the ground-state carries the computation information. If a system does not stay in its ground-state it will most likely end up in one of the next lowest excited states. One can fundamentally change the computation model by considering not only the ground-state but also the first excited states as acceptable output states. Unfortunately no quantum adiabatic theorem that ensures the close evolution to states of different energy levels is known so far. But other computation models related to adiabatic quantum computation like measurement algorithms exist that are capable of evaluating such behaviour. One may also consider including error correction into such an algorithm to improve its efficiency.

\bibliography{master}
\bibliographystyle{plain}
\end{document}